\documentclass[a4paper,reqno]{amsart}
\usepackage{enumitem}
\usepackage[T1]{fontenc}
\usepackage{amsmath,amssymb,amsfonts,amsthm}
\usepackage{hyperref}
\usepackage{slashed,comment}
\usepackage{graphicx,makecell}
\usepackage{color}
\usepackage[dvipsnames]{xcolor} 
\usepackage[bbgreekl]{mathbbol}
\usepackage{mathtools}
\usepackage[all,cmtip]{xy}
\usepackage{stmaryrd} 
\usepackage[toc,page]{appendix}
\usepackage{tikz}
\usetikzlibrary{patterns}
\usetikzlibrary{arrows.meta}

\usepackage[backend=biber, giveninits=true, style=alphabetic, sorting=nyt,isbn=false, maxalphanames=5,url=false,doi=false, maxbibnames=99]{biblatex}
\usepackage{soul}
\definecolor{light-gray}{gray}{0.95}
\usepackage[colorinlistoftodos,textsize=footnotesize,color=light-gray]{todonotes}

\makeatletter
\newcommand{\circT}{\mathbin{\mathpalette\circT@{}}}
\newcommand{\circT@}[2]{%
  \vcenter{\offinterlineskip
    \hbox{$#1\ooalign{%
      \hfil\raise0.3ex\hbox{\scalebox{0.30}{$\m@th#1\mathcal{T}$}}\hfil\cr
      \hfil$\m@th#1\circ$\hfil\cr
    }$}%
  }%
}
\makeatother

\newcommand{\be}{\begin{equation}}
\newcommand{\ee}{\end{equation}}
 
\newcommand{\ph}{\varphi}
\newcommand{\reg}{\mathrm{reg}}

\newcommand{\rpar}[2]{\frac{#1\!\!\stackrel{\leftarrow}{\delta}}{\delta{#2}}}
\newcommand{\lpar}[2]{\frac{\stackrel{\rightarrow}{\delta}\!\!#1}{\delta{#2}}}
\DeclareMathOperator{\supp}{\mathrm{supp}}

\newcommand{\loc}{\mathrm{loc}}
\newcommand{\ml}{\mathrm{ml}}
\renewcommand{\top}{\mathrm{top}}
\newcommand{\std}{\mathrm{std}}
\newcommand{\oloc}{\Omega_{\loc}}
\newcommand{\iloc}{\Omega_{\int}}

\newcommand{\bom}{\boldsymbol{\omega}}
\newcommand{\bth}{\boldsymbol{\theta}}
\newcommand{\bF}{\mathbf{F}}
\newcommand{\bG}{\mathbf{G}}

\DeclareMathOperator{\im}{\textrm{im}}

\newcommand{\Rcal}{\mathcal{R}}

\newcommand{\Acal}{\mathcal{A}}
\newcommand{\Ocal}{\mathcal{O}}
\newcommand{\fA}{\mathfrak{A}}

\newcommand{\Pei}[2]{\lfloor #1, #2 \rfloor}

\newcommand{\NN}{\mathbb{N}}

\newcommand{\Tcal}{\mathcal{T}}
\newcommand{\Xcal}{\mathcal{X}}

\newcommand{\Kcal}{\mathcal{K}}
\newcommand{\Dcal}{\mathcal{D}}
\newcommand{\Ecal}{\mathcal{E}}
\newcommand{\Scal}{\mathcal{S}}

\newcommand{\Fcal}{\mathcal{F}}

\newcommand{\Ci}{\mathcal{C}^\infty}

\newcommand{\ren}{\mathrm{ren}}

\newcommand{\src}{\mathrm{src}}

\newcommand{\TT}{\mathcal{T}}
\newcommand{\T}{\cdot_{\TT}}

\newcommand{\Lap}{\triangle}

\newcommand{\eihbV}{e^{\frac{i}{\hbar}V}}

\usepackage{thmtools}

\newtheorem{definition}{Definition}[section]
\newtheorem{theorem}[definition]{Theorem}
\newtheorem{proposition}[definition]{Proposition}

\newtheorem{lemma}[definition]{Lemma}
\newtheorem{corollary}[definition]{Corollary}

\declaretheoremstyle[
  spaceabove=5pt, spacebelow=5pt,
  headfont=\itshape,
  notefont=\normalfont, notebraces={(}{)},
  bodyfont=\normalfont,
  postheadspace=1em,
  qed=$\spadesuit$
]{pluto2}
    \declaretheorem[style=pluto2,name=Remark,    sibling=definition]{remark}

\newcommand{\de}{\partial}

\title[BV-BFV in pAQFT]{Perturbative algebraic quantum field theory with smoothened boundary}
\author{Kasia Rejzner}
\address{Department of Mathematics, University of York, Heslington, York YO10 5DD, United Kingdom}
\email{kasia.rejzner@york.ac.uk}

\author{Michele Schiavina}
\address{Department of Mathematics, University of Pavia, Via Ferrata 5, 27100 Pavia, Italy}
\address{INFN Sezione di Pavia, via Bassi 6, 27100 Pavia, Italy}
\email{michele.schiavina@unipv.it}

\begin{document}

\begin{abstract}
We formulate quantisation of gauge field theories on globally hyperbolic Lorentzian manifolds with marked hypersurfaces within the framework of perturbative algebraic quantum field theory (pAQFT) enriched by the Batalin, Fradkin, Vilkovisky formalism (BV/BFV). This allows one to incorporate some crucial aspects of the local functorial approach to gauge field theory on manifolds with boundary within the algebraic setting. In particular, we provide a pAQFT-formulation of the modified classical and quantum master equations (after Cattaneo, Mnev and Reshetikhin CMR), as well as a constructive way to build a renormalised quantum BFV operator correcting the failure of the Quantum Master Equation by means of boundary terms (in the appropriate sense). We find that the renormalised quantum homotopy dg Lie algebra arising from the Anomalous Master Ward Identity becomes curved when boundaries are considered. The failure of the quantum master equation is thus encoded by a nontrivial curvature term in an $L_\infty$ algebra. As a byproduct, we recover previous results of Hollands' on the relation between the (boundary) BRST charge and the BV operator, and we recover CMR's ansatz for the quantum BFV operator at leading perturbative order on causal cylinders in Abelian Yang--Mills theory.
\end{abstract}

\maketitle

\tableofcontents

\section{Introduction}
Defining quantum field theory (QFT) or even articulating what it seeks to achieve is no simple task. This difficulty stems largely from the fact that the backbone of modern QFT currently consists of a large variety of methodologies developed by different communities, each extending the quantisation of classical mechanics to the richer and more technically involved realm of local field theory.

At present, no single non-perturbative formulation is known that simultaneously captures all the features of quantum mechanics and accommodates four-dimensional interacting models of physical relevance. Although the existing formalisms each illuminate specific aspects of the theory, none may yet claim to provide a complete and unified description of QFT (see e.g.\ \cite{Dedushenko_Snowmass}).

Nevertheless, the past two decades have witnessed substantial progress in mathematical quantum field theory, whose goal is to provide a rigorous framework for the quantisation of local (Lagrangian) field theories, with special attention being dedicated to gauge theory, which includes gravity. A pivotal development in this programme has been the introduction of cohomological techniques due to Batalin and Vilkovisky (BV), and their Hamiltonian counterpart BFV (incorporating the contributions of Fradkin) \cite{BFV0,BV1,BV2,BFV1,BFV2}. Originally formulated in the 1970s and 1980s, these ideas have been since then largely extended in rigor, scope and applicability. (Among the many contributors to said developments we mention \cite{HenneauxTeitelboim,BarnichBrandtHenneaux,Stasheff,Stasheff1997,Grigorievparent,GrigorievKotov,Grigoriev2022PresymplecticGP}. See also \cite{Grigoriev} for a historically accurate summary of recent contributions.)  At present, three main approaches to \emph{quantisation} within the BV formalism stand out, for the purposes of this paper: 
\begin{enumerate}
    \item the perturbative algebraic QFT (pAQFT) framework \cite{Rej11b,FR,FR3,FredenhagenRejzner2015,Book} (based on earlier works like \cite{DFqed,Hollands}), whose emphasis lies in constructing nets of algebras of observables via perturbative renormalization;
    \item the factorization algebra approach \cite{costello2011renormalization,CoGw,CG2}, which builds observables from BV data through cosheaf-like structures;
    \item the local functorial BV–BFV programme \cite{CMR1,CMR2,IrasoMnev,CMRCell,MSW,CMW,CattaneoMnevGluing}, combining the Lagrangian BV formalism with its ``Hamiltonian'' counterpart, the BFV formalism. It places spacetime locality and codimension-hierarchies at the forefront and treats gluing and functoriality as structural principles. This approach is devised to study higher-codimension data in a field theory.
\end{enumerate} 

Given these parallel developments, it is crucial to seek compatibility and interplay between frameworks if we are to move toward a coherent and unified picture of QFT in which tools can be interchanged and foundational features clarified. Progress in this direction has already been made, e.g. in connecting pAQFT and factorization algebras in the Lorentzian setting \cite{GR20,GR22}, as well as in categorifying aspects of AQFT \cite{BFV,FV12,beninischenkelcomparison,benini2024equivalence}. In addition, a framework has been developed in \cite{BMS23} that allows the formulation of AQFT models as functorial field theories. However, a direct implementation of the higher codimension structures prominent in the BV-BFV approach into the pAQFT framework
has not yet been developed. (A first step in bridging between approaches (2) and (3) has been taken by Rabinovich \cite{Rabinovich}, and a semiclassical attempt at this programme was initiated in \cite{RejznerSchiavina_Asymptotic}.)

In this work, we close this gap by bringing higher-codimension strata and their associated field-theoretic data into the pAQFT setting. To our knowledge, this piece of the mosaic has not yet been placed, and we believe that the ensuing insights will be valuable for the mathematical QFT community.

A key idea is the treatment of higher-codimension surfaces as \emph{smoothened}, allowing us to retain access to the BFV data, while embedding it as a modification of the standard pAQFT net-of-algebras construction.

Using local Lagrangians formulated in the variational bicomplex langauge and enhanced by the BV cohomological machinery (mostly following \cite{MSW} and \cite{CattaneoMnevGluing}, but see also \cite{Grigoriev,Grigorievparent,GrigorievKotov,Grigoriev2022PresymplecticGP,SchiavinaSchnitzer,RielloSchiavina}), we develop a perturbative algebraic quantisation scheme capable of incorporating theories on globally hyperbolic spacetimes with \emph{smoothened marked hypersurfaces} (understood as regularized or thickened boundaries). We propose modified versions of the \emph{classical and quantum master equations} for theories with boundary contributions in the spirit of \cite{CMR1,CMR2}, and we showcase the results, in the particular example of Abelian Yang--Mills theory (Theorem \ref{thm:PAQFT-CMR}). 

For theories whose equations-of-motion operators are Green-hyperbolic (see e.g.\ \cite{GreenBear} and \cite{BeniniMusanteSchenkel} for a recent take on this notion), we obtain corrections to the quantum BV operator, given in terms of codimension $1$ generalized currents, quantized via deformation and suitable Møller maps.\footnote{In \cite{beninischenkelmoller} a no-go theorem for the existence of M{\o}ller maps has been presented. However, that analysis concerns hypothetical M{\o}ller maps for interactions with no spacetime cutoff, while our approach makes crucial use of such cutoffs, which invalidates the applicability of the no-go theorem for our results. More details and an explicit analysis of this difference is presented in the upcoming paper \cite{NewHRV}.}

{Specifically, if $\mathsf{QME}(V)$ denotes the expression controlling the quantum master equation for an interaction $V$, the \emph{modified} quantum master equation for some current $J$, a map from compactly supported one-forms to functionals, reads
\[
\mathsf{QME}(V) + J = 0.
\]
We show that this is equivalent to (Theorem \ref{thm:mQMEStates})
\[
\widetilde{Q}^R_{J,H} \TT\eihbV  = 0,
\]
where 
\[
\widetilde{Q}^R_{J,H}\doteq \left(Q_0 + \frac{i}{\hbar}(\cdot) \star_H R_{V,H}(J)\right)
\]
is the retarded, modified, quantum BV operator, defined by the free BV operator $Q_0$, a choice of a star product $\star_H$ and retarded M\o ller map on functionals $R_{V,H}$, as well as a time ordering map $\TT$.\footnote{Note that $H$ denotes the dependence on a Hadamard state, which also determines $\TT$ and, essentially, $R_{V,H}$.}
}

This shows how the perturbative algebraic QFT machinery can be successfully applied to bounded regions, {which are detected through the support of the one form fed into $J$. The effect of adding boundaries is that the Epstein--Glaser-renormalised BV machinery produces a \emph{curved} (homotopy) dg Lie algebra, generalising the results of \cite{Froeb,BDFR_infty}, where the curvature is controlled by the current $J$ (see Theorem \ref{thm:AMWI_curved}). We also provide a generalisation of the extended Wess--Zumino consistency condition for the Anomaly term in the Master Ward Identity, modified by boundary terms}.

{In terms of the otherwise-standard interacting (retarded/advanced) BV operators $Q^{R/A}_V$, obtained from the classical BV operator via adjunction by M\o ller maps, the introduction of a (smoothened) boundary produces the modificaton of the operators (see the precise statement in Theorem \ref{thm:R/AqBV}):
\begin{subequations}
\begin{align}
    Q^R_V(F) 
        &= (Q_0 + \{V,F\} -i\hbar \Lap F) + \frac{i}{\hbar} \overrightarrow{\Omega}_J \\\notag
        &\doteq (Q_0 + \{V,F\} -i\hbar \Lap F) + \frac{i}{\hbar}(J{\star_V} F - J\cdot F)\\
    Q^A_V(F) 
        &= (Q_0 + \{V,F\} -i\hbar \Lap F) + \frac{i}{\hbar} \overleftarrow{\Omega}_J \\
        &\doteq (Q_0 + \{V,F\} -i\hbar \Lap F) + \frac{i}{\hbar}(F{\star_V} J - J\cdot F)\notag
    \end{align}
\end{subequations}
whereby, differently from the unbounded case, we have a nonzero difference
\[
(Q^R_V- Q^A_V)(F) = \frac{i}{\hbar}\left(\overrightarrow{\Omega}_J(F) - \overleftarrow{\Omega}_J(F)\right)= \frac{i}{\hbar}[J,F]_{{\star_V}},
\]
for $\star_V$ the interacting star product (cf.\ Definition \ref{def:intstarprod}, Lemma \ref{lem:Vproduct} and Remark \ref{rem:starVwelldef}). The operators $\overleftrightarrow{\Omega}$ are called (retarded/advanced) quantum BFV corrections.
} 

The boundary-modified BV operators depend on the choice of splitting between free and interacting parts of the theory. In what we call \emph{the minimal splitting} (defined so that the free part does not contain antifields), we recover (within our BV setting) Hollands’s BRST-based on-shell equivalence between the “boundary” interacting BRST charge and the “bulk” quantum BV operator (see Section~\ref{sec:Hollands} and cf.\ with \cite{Hollands}). This `holographic correspondence' between bulk and boundary operators hints at a broader class of correspondences that are in principle accessible in our formalism, obtained by choosing alternative free-interacting splittings.

For comparison with the BV-BFV formalism of \cite{CMR1,CMR2} (see Section~\ref{sec:BVBFVPAQFT}), we show (at lowest order in $\hbar$) that our boundary corrections coincide with the quantized BFV boundary action of  \cite{CMR1}. This is a significant improvement beyond the state of the art, as our expression for the boundary corrections is given to all orders in $\hbar$ and can be applied in a broad class of physically relevant examples. We illustrate the construction explicitly in the case of electromagnetism (Abelian Yang-–Mills theory without matter). 

Finally, we want to stress that our framework also provides a clear connection between the BFV quantisation, where the quantum BFV correction $\overrightarrow{\Omega}_J$ contains the quantisation of the boundary action at lowest order (see Section \ref{sec:CMR}), to the Kugo-Ojima formalism, where $\overrightarrow{\Omega}_J$ is the commutator with the BRST charge, as we argue in Section \ref{sec:Hollands}. In both formalisms, $\overrightarrow{\Omega}_J$ is used to select \textit{physical states}.

\medskip
\noindent \textbf{Quantisation strategy overview.} In (perturbative) algebraic quantum field theory ((p)AQFT) on globally hyperbolic Lorentzian spacetimes one observes that, using a construction originally due to Peierls \cite{Pei}, the space of field configurations is endowed with a Poisson bracket whenever it is endowed with an operator admitting retarded and advanced Green’s functions (i.e. the operator is Green hyperbolic). When restricted to the space of solutions to the equations of motion associated to such operator (a.k.a.\ \emph{on-shell}), this bracket is non-degenerate. This structure is the starting point for quantisation via formal deformations, yielding a net of algebras interpreted as the quantum observables of the theory.

The perturbative philosophy underpinning pAQFT views a general Lagrangian field theory as a deformation of a free theory (defined by a quadratic Lagrangian) by introducing an interaction term. Quantisation is performed for the free (Green-hyperbolic) system, with the interaction introduced subsequently by a twist to the star product. The algebra of interacting observables is then obtained from the free algebra via intertwining maps known as Møller maps. See \cite{Book} for more detail and \cite{FR24} for the history of the subject.

A number of technical and conceptual difficulties complicate this strategy. When gauge symmetries are present, the operator defining the free theory is typically degenerate, and Green-hyperbolicity becomes available only after gauge fixing. Moreover, since the center of the Peierls Poisson bracket is generated by functionals vanishing on-shell, it is natural and advantageous to formulate the problem directly in the Batalin–Vilkovisky (BV) framework, which cohomologically resolves the moduli space of solutions. This has the advantage of allowing one to deal with gauge theories in one fell swoop, as the basic strategy outlined above goes through unchanged and one constructs a gauge-fixed Green-hyperbolic operator, which leads to a deformed gauge-fixed Peierls bracket. For an alternative approach, where the Peierls bracket is constructed without adding an explicit gauge fixing term to the action, see \cite{BeniniMusanteSchenkel}, where the gauge-fixing is enforced at the cochain level by replacing Green's functions with Green's chain homotopies. This is similar---and in some cases coincident---with the concept of gauge fixing operators as seen in \cite{costello2011renormalization} (see also \cite{SchiavinaStucker}).
For a comparison between the \cite{BeniniMusanteSchenkel} approach and the gauge-fixed construction using the BV framework of \cite{FR,FR3}, see the upcoming paper \cite{NewHRV}. 

Irrespective of the approach taken, one typically realizes classical observables as functionals on some odd-symplectic graded manifold, equipped with the (gauge-fixed) Peierls bracket, and without the loss of generality we may assume this to be the starting point for perturbative quantisation.

In order to realise this perturbative description of the quantum theory on (appropriate) functionals of the fields, we follow \cite{HR}. Star products can be defined subordinate to a choice of an integral kernel (a bidistribution) and taking an exponential, i.e. 
\[
F\star G = m\circ e^{\hbar \langle K, \delta^{\otimes2}\rangle}(F\otimes G)
\]
where the antisymmetric part of $K$ is fixed by the kernel of the Peierls bracket. Changing its symmetric part gives a host of equally viable deformation quantisations, among which the Moyal--Weyl quantisation and the Wick quantisation. Underlying all of them is an abstract algebra\footnote{The limit of all the above, running over Hadamard states, see \cite{HR}.} with an abstract star product and an involution. One can identify this abstract algebra with its ``realisation'' within functionals of the fields, and this requires a choice of Hadamard state. Different identifications of this abstract algebra give rise to different ``quantisation maps''.

In order to realise Møller maps within this picture—and thereby the interacting theory—one introduces the \emph{time-ordered product} on appropriate classes of functionals. This class is given by applying the \emph{time-ordering operator} to the space of multilocal functionals. Time-ordered multilocal functionals are generically no longer multilocal, but they fit within a larger space of \emph{equicausal functionals} \cite{HRV}, where both the Peierls bracket and the free quantum star product are well defined.
(This is a graded commutative product, formally obtained using path integral, and the time ordered product is formally interpreted as path-integral convolution of the given functionals with the Gaussian measure of the free theory; made precise using Epstein-Glaser renormalization \cite{EpsteinGlaser}.)
Once again, time ordered products can be described abstractly, or realised in a particular incarnation (or presentation) of the quantum star algebra. 

Endowed with the time-ordered product, time-ordered (multilocal) functionals acquire the structure of a factorization algebra. Indeed, the free quantum star product can be recovered from the time-ordered product by ordering arguments according to their causal relationships. This observation firmly links pAQFT to the factorization algebra framework of Costello–Gwilliam \cite{CoGw}: time-ordering provides a bridge between the pAQFT viewpoint of \emph{deforming the product while keeping the BV differential fixed} and the factorization algebra viewpoint of \emph{deforming the BV differential while keeping the product fixed} (see \cite[Section 7]{GR20} and \cite[Section 6.1]{GR22} for an overview of this interpretation). This fact can be nicely reinterpreted as a link between \emph{realisations} of a quantisation map internal to the space of functionals \cite{HR}. We outline this point of view in Section \ref{sec:freequantum}.

 The existence of time-ordered products is, however, delicate and brings renormalisation to the forefront. The above-mentioned Epstein–Glaser renormalisation method \cite{EpsteinGlaser}, adapted to the BV formalism \cite{FR3}, provides a systematic way to construct such products, modulo cohomological obstructions. The resulting ambiguities are controlled by the renormalisation group. With renormalised time-ordered and star products in hand, the Møller map is formed to define interacting functionals, now acted upon by the quantum BV differential. Physical observables are subsequently described as cohomology classes of this deformed BV operator.

A central condition for the above strategy to work is the Quantum Master Equation (QME). Its original role and motivation \cite{BV1} was to ensure independence from the choice of gauge-fixing (see \cite{CattaneoMnevSchiavina} for a recent review). In the factorization algebra context and in pAQFT, the QME guarantees that the deformed interacting quantum BV differential remains local, yielding a quantum factorization algebra \cite{GR20,GR22}. A solution of the QME typically is a perturbation of a solution of its classical counterpart, the Classical Master Equation (CME). However, almost no nontrivial BV theory actually satisfies the CME in its strong form on manifolds with boundary. Boundary terms are both ubiquitous and essential, and up to now mostly neglected for a variety of reasons. Indeed, in the Lorentzian field theory community they are typically dealt with by imposing explicit boundary or falloff conditions, and in the topological field theory community the main interest is typically directed towards closed manifolds. Boundary conditions are also implemented in the factorisation algebra setting following \cite{Rabinovich} (see also \cite{MMST}).

In the last decade, however, Cattaneo, Mnev and Reshetikhin \cite{CMR1,CMR2} developed a systematic programme to deal with quantisation in the presence of boundaries, by introducing modified classical and quantum master equations (mCME / mQME) where the BV operator is corrected by a boundary term formulated within the BFV language, i.e.\ a degree $1$ function in a $0$-symplectic dg manifold, which satisfies the classical master equation. Their proposal has been tested in topological settings, prominently $BF$ and Chern-–Simons theory, leading to an ansatz for the required correction of the quantum BV operators and associated master equations \cite{CMR2,CMRCell,CMW}.

The paradigm thus must change: we cannot hope to find a solution of the QME starting from a solution of CME if the latter cannot be solved. Rather, one deals with modifications of both classical and quantum master equations (mCME and mQME respectively). In this paper, we find the correct placement of these insights within the perturbative AQFT framework. We work on finite globally hyperbolic regions with (smoothened) boundary, and address the failure of the CME due to boundary terms. We show that the pAQFT formalism is sufficiently flexible to accommodate general (potentially non-local) corrections to the BV operator, thereby providing a roadmap to perturbative quantisation of gauge field theory in the presence of boundaries.

\section*{Acknowledgements}
This paper took many attempts to complete, over many years. We are indebted to a number of people for support and comments that helped us out of the various ruts in which we got stuck along the way. We thank Alberto Cattaneo, Ezra Getzler, Owen Gwilliam, Eli Hawkins, Peter Michor, Pavel Mnev, Alexander Schmeding and Konstantin Wernli. We should also thank the many institutions that hosted our collaboration in the form of scientific visits, conference invitations and other collaborative formats. Certainly Perimeter Institute, where several interactions took place, and the Montana State University, where a crucial piece of the puzzle emerged (we especially want to thank Ryan Grady for setting up the meeting that allowed us to find it), but also ETH Zurich, the University of Notre Dame, the University of Massachussets at Amherst, Amherst College, the Humboldt University in Berlin, Institut Henri Poincar\'e, as well as our own home institutions. Last but not least, we thank Mittag-Leffler Institute and the organisers of the conference ``Cohomological aspects of quantum field theory'' where the final push towards completion of this paper took place (even though it still required a significant amount of time after that).

\section{Preliminaries and Background}
In this section we review the results of \cite{FR3,FredenhagenRejzner2015} on the perturbative algebraic quantum field theory (pAQFT) framework applied to gauge theories, as well as the results of \cite{CMR1,CMR2} on classical and quantum local field theory on manifolds with boundary. The framework we propose here is a combination of the two.

\subsection{Green's functions and functionals}\label{sec:functionals}

Throughout, we will consider a (possibly graded) vector bundle $E\to M$ over a smooth manifold, and denote by $E^! \doteq E^* \otimes \mathrm{Dens}(M) \to M$ the densitised dual bundle. We denote by $\Dcal(M,E) \equiv \Gamma_c(M,E)$ the space of smooth compactly supported sections, so that the space of distributional sections of $E$ is the strong dual $\Dcal'(M,E)\doteq \Dcal(M,E^!)'$. The space of smooth sections of $E$ is denoted by $\Ecal=\Gamma(M,E)$, and it is endowed with a natural Fr\'echet topology. 

Even though a large part of this preliminary section is general, our main interest will be directed towards the case of $M$ a globally hyperbolic Lorentzian manifold, which can be characterized as a Lorentzian (typically 4d) manifold that admits a foliation where every leaf is a Cauchy hypersurface. (See \cite{FR} for an extensive introduction.) Recall that, in that case, for any point $p\in M$ the set $\mathcal{J}^+(p)$ (resp. $\mathcal{J}^-(p)$) contains the points $p'\in M$ that can be connected to $p$ by a future (resp. past) causal curve $\gamma^+\colon I\to M$. These are called the causal future (resp. past) of $p$. This allows us to say that two sets $O_1,O_2 \subset M$ are causally separated iff $\forall p\in \overline{O}_1$ we have that $\mathcal{J}^\pm(p)\cap O_2=\emptyset$. Let us denote by $\mathrm{Caus}(M)$ the set of causally convex, open sets in $M$. This becomes a category with causality-preserving embeddings. The differential operators we will consider acting on $\Ecal$ belong to the following class.

\begin{definition}\label{def:propagators}
    Let $D$ be a differential operator on $\Ecal$. It is said to be ``Green hyperbolic'' if it admits retarded/advanced Green functions, i.e.\ maps $\Delta^{R/A}: \Ecal_c^!\rightarrow \Ecal $, such that $D\Delta^{R/A} = \Delta^{R/A}D\vert_{\Ecal_c} = \mathrm{id}\vert_{\Ecal_c}$, uniquely defined by their support properties:
    \[
    \supp(\Delta^{R/A}[f])\subset J^{\pm}(\supp(f))\,
    \]
    for $f\in \Ecal_c^!$ (see \cite{FR,Rej13}).

    \noindent The Dirac and  the Pauli-Jordan (or causal) propagators associated to $D$ are, respectively 
    \begin{align*}
        \Delta^{\rm D}\doteq \frac12 (\Delta^{R}+\Delta^{A}), \qquad \Delta\doteq  \Delta^{R}-\Delta^{A}.
    \end{align*}

    \noindent Given a Hadamard 2-point function\footnote{This means a symmetric, distributional bi-solution for $D$, with additional wavefront set properties, see e.g.\ \cite{FredenhagenRejzner2015,Book}.} for  $H\in \Ecal_c^!\rightarrow \Ecal$,
    $$\Delta^+ \doteq \frac{i}{2}(\Delta^R - \Delta^A) + H = \frac{i}{2} \Delta + H,
    $$ 
    the corresponding Feynmann propagator is given by \cite{Rej13,Book}
    \begin{align*}
\Delta^F \doteq \frac{i}{2}(\Delta^R + \Delta^A) + H = i\Delta^D + H. 
    \end{align*}
\end{definition}

\begin{remark}
It is often convenient to think of
retarded and advanced Green's functions, as well as the Dirac and Pauli--Jordan propagators, and the Feynmann propagator as distributions in $\Dcal'(M^2,E^{\boxtimes 2})$ as well as operators $\Ecal_c^!\rightarrow \Ecal$. We silently use this identification throughout this work.
\end{remark}

We now introduce several spaces of functionals on $\Ecal$, distinguished by additional regularity conditions imposed on the derivatives of maps $F\colon \Ecal\to Y$, where $Y$ is a locally convex topological vector space. Derivatives are understood in the sense of Bastiani~\cite{Bastiani}, and in all cases of interest the Bastiani notion of smoothness coincides with \emph{convenient smoothness} in the sense of Kriegl--Michor~\cite{KrieglMichor}. (The definitions that follow extend verbatim to any smooth manifold $M$.)

\begin{definition}
    The spacetime support of a functional $F\colon \Ecal \to Y$ is defined as
    \begin{multline*}
    \supp(F) = M\backslash \{x\in M\ |\ \forall U\subset M\ {\rm open}\ \mathrm{with}\ x\in U,\\ \exists \ph,\psi\in\Ecal,\supp\psi\subset U,\ \mathrm{s.t.} F(\ph+\psi)\neq F(\ph)\}.
    \end{multline*}
\end{definition}

\begin{definition}[Regular functionals]\label{def:regfun}
    Denote by $\mathcal{F}(\Ecal)$ the space of smooth compactly supported $\mathbb{C}$-valued functionals on $\Ecal$.
    A functional in $\mathcal{F}(\Ecal)$ is called \emph{regular} iff
    \[
    F^{(n)}(\varphi) \in \Dcal\Big(M^n,{E^!}^{\boxtimes n}\Big), \qquad \forall \varphi\in\Ecal\,.
    \]
    and there exists an $N\in\NN$ such that $F^{(N)}\equiv 0$ (i.e. we consider only polynomial functionals). 
    We denote the space of regular functionals by $\mathcal{F}_{\rm reg}(\Ecal)\equiv \Fcal_\reg$.
\end{definition}

\begin{remark}\label{rmk:regulareihbv}
    Note that, if $F,V\in \Fcal_\reg$, then both $\eihbV$ and $F\eihbV$ are regular.
\end{remark}

\subsection{Densitites, Currents and generalized Lagrangians\label{sec:densities currents lagrangians}}
Intuitively, local densities are functions constructed at each spacetime point from the fields and their derivatives (up to a finite order). This notion also generalizes from functions to local forms. 

More precisely, we define the space of local densities over $\Ecal\times M$ by taking the pullback along the (infinite) jet evaluation map $j^\infty\colon \Ecal \times M \to J^\infty E$. Since on $J^\infty E$ the space of differential forms admits the structure of a bicomplex (the variational bicomplex \cite{Anderson,BlohmannLFT}, see also \cite{SchiavinaSchnitzer}), we can transfer the bicomplex structure onto $\Ecal \times M$:
\begin{definition}[Local densities]
    The bicomplex of local $(p,q)$ densities is the vector space
    \[
    \oloc^{p,q}(\Ecal\times M) = (j^\infty)^*\Omega^{p,q}(J^\infty E)
    \]
    endowed with the (anticommuting) \emph{variational} and \emph{spacetime} differentials
    \begin{subequations}\begin{align}
        \delta (j^\infty)^*\alpha \doteq (j^\infty)^* d_V \alpha \\
        d(j^\infty)^*\alpha \doteq (j^\infty)^* d_H \alpha,
    \end{align}\end{subequations}
    where $d_V$ and $d_H$ are, respectively, the vertical and horizontal differentials in the variational bicomplex.
\end{definition}

    Given a $p,\top-k$ form it is possible to wedge it with a local $0,k$ form with compact support and integrate over the whole manifold. This leads to the following definition.

\begin{definition}
    Given a  $(p,\top-k)$ local form $\mathbf{K}\in\oloc^{p,\top-k}(\Ecal\times M)$, we can define the \emph{generalised $p$-form $k$-current\footnote{Note that every generalised $0$-form $k$ current defines a $k$ current in the standard sense by evaluation on a field configuration: $K(\phi)\colon \Omega_c^k(M)\to \mathbb{R}$.} associated to $\mathbf{K}$} as the map
    \[
    K\colon \Omega^k_c(M) \to \Omega^p(\Ecal), \qquad \beta \mapsto K(\beta)\doteq\int_M \mathbf{K} \wedge \beta.
    \]
    We denote the space of generalised $p$-form $k$-currents associated to local densities via the integration map as $\iloc^{p,(k)}(\Ecal)$.

    \noindent A generalised $0$-form $0$-current associated to a local density $\mathbf{L}\in \oloc^{0,\top}(\Ecal\times M)$ is called \emph{generalised Lagrangian associated to $\mathbf{L}$}:
    \begin{align}
        L\in \iloc^{0,(0)}(\Ecal), \quad  L\colon C_c^\infty(M) \to C^\infty_\loc(\Ecal) \qquad 
        f \mapsto L[f]\doteq\int_M f \mathbf{L},
    \end{align}
    iff, additionally, $\supp(L[f])\subset \supp(f)$ and the Hammerstein property holds:
\be\label{e:Hammerstein}
L(f+g+h)=L(f+g)-L[g]+L(g+h)\,,
\ee
if $\supp f\cap \supp g=\varnothing$. 
\end{definition}
 
Generalised Lagrangians were originally introduced in \cite{BDF}, they can be seen as maps between test functions and \emph{local functionals}. The space of local functionals $\mathcal{F}_\loc(\Ecal)$ can be given an abstract definition, using the appropriate version of Hammerstein property \eqref{e:Hammerstein} \cite{BDGR}, but we will not be interested in this level generality so for us local forms on $\Ecal$ will be just generalised currents associated to local densities. 

\begin{remark}[Currents as natural transformations]
One can also view generalised Lagrangians as natural transformations  between two functors from Lorentzian manifolds to topological algebras: the test functions functor $C_c^\infty$ and local functionals functor $C^\infty_{\loc}$. This point of view has been introduced by Brunetti, Fredenhagen and Verch in \cite{BFV}, where also locally covariant fields (generalized fields) are seen as natural transformations. Generalised $k$-currents associated to local densities extend the notion of generalised Lagrangians and can be promoted to natural transformations between functors that respectively assign test \emph{forms} and local functionals to spaces. 
\end{remark}

\begin{remark}[Multilocal functionals]
    Notice that $\mathcal{F}_\loc(\Ecal)$ is not closed under multiplication. Hence, we algebraically complete it to $\mathcal{F}_{\ml}(\Ecal)$, the space of \textit{multilocal functionals}, as defined in \cite[Definition~3.2]{GR22} (following the ideas of \cite{FR}). This space contains in particular sums of finite products of local functionals. The generalisation to multilocal forms is straightforward. In essence\footnote{Here we assume for the sake of simplicity that one can work globally. The more precise functorial definition is presented in \cite[Sections 3.1 and 3.2]{GR20}, and in particular  \cite[Definition~3.12]{GR22}.} one can look at 
    \[
    j^\infty_{[k]} \colon \Ecal\times M^k \to J^\infty E^{\boxtimes k}
    \]
    so that one defines
    \[
    \Omega_{\mathrm{loc},k}^{\bullet,\bullet}(\Ecal\times M) \doteq (j^\infty_{[k]})^*(\Omega^{\bullet,\bullet}(J^\infty E^{\boxtimes k})),
    \]
    together with the map 
    \[
    \int_{M^k} \colon \Omega_{\mathrm{loc},k}^{0,{\rm top}}(\Ecal\times M) \to C^\infty(\Ecal), \qquad \mathcal{F}_{\loc,k}(\Ecal) \doteq \mathrm{Im}\left(\int_{M^k}\right).
    \]
    $k$-local functionals  in $\mathcal{F}_{\loc,k}(\Ecal)$ are of the form
    \[
    F(\varphi_1,\dots ,\varphi_k) = \int_{M^k} \mathbb{F}(j^\infty_{x_1}\varphi_1,\dots,j^\infty_{x_k}\varphi_k),
    \]
    for $\mathbb{F}\in \Omega^{0,{\rm top}}(J^\infty E^{\boxtimes k})$. This space contains in particular sums of $k$-products of local functionals. Finally, we define the algebra of multilocal forms as\footnote{Recall that by the direct sum we mean the space of finite sums of finite products of local functionals.} 
    \[
    \mathcal{F}_{\mathrm{ml}}(\Ecal) = \bigoplus_{k=1}^\infty \mathcal{F}_{\loc,k}(\Ecal).
    \]
    This is the algebraic closure of the space of local functionals. (Observe that in \cite{FKR} a similar procedure is considered, effectively looking at an algebra of ``multilocal densities'' over unordered configuration spaces, before the integration map is applied.) Note that throughout, when unambiguous, we will always shorten $\Fcal_*(\Ecal)\equiv \Fcal_*$ for all functional subspaces.
    \end{remark}

\begin{definition}\label{def:eqrel}
    Two generalised Lagrangians $F,G$ are said to be bulk equivalent $F\sim G$ iff for every compactly supported test function $f\in C^\infty_c(M)$ we have
    \[
    \supp(F[f] - G[f])\subset \supp(df).
    \]
\end{definition}

We will need an important result in the theory of local forms, which stems from the following definition:
\begin{definition}[Source forms and derivatives]\label{def:sourceform}
Given a local chart in $J^{|I|}E$, say $\{x^i,u^\alpha, u_I^\alpha\}$ for $I$ a finite multiindex, we define \emph{source forms} as local forms ${\bF}\in\oloc^{1,\text{top}}(\Ecal\times M) $ such that 
\[
{\bF} = (j^\infty)^*(F_\beta(x^i, u^\alpha, u_I^\alpha) d_Hu^\beta \wedge dx^1 \wedge  \cdots \wedge dx^{\text{top}})
\]
for some smooth functions $F_\beta(x^i, u^\alpha, u_I^\alpha)\in C^\infty(J^{|I|}E)$. We denote the space of local forms of source type by $\Omega_{\text{src}}^{1,\text{top}}(\Ecal\times M)$.
\end{definition}

Source forms are often called functional forms, e.g.\ by Anderson in \cite{Anderson}. They depend on a finite number of derivatives of $\varphi\in\Ecal$ and only on the vertical differential of the base section $\delta\varphi$ and none of its jets of the form $\delta\partial\varphi$. (See \cite{DelgadoPhD,BlohmannLFT,Anderson} for more details.)

\begin{theorem}[Lemma 1 of \cite{Zuckerman}, see also \cite{Takens,Takens77}]
The space of $(\text{top},1)$-local forms admits the decomposition:
\[
\oloc^{\text{top},1}(\Ecal\times M) = \Omega_{\text{src}}^{\text{top},1}(\Ecal\times M) \oplus d \oloc^{\text{top}-1,1}(\Ecal\times M).
\]
We denote by $\Pi_{\mathrm{src}}$ and $\Pi_{\mathrm{bdr}}$ the projections respectively onto $\Omega_{\text{src}}^{\text{top},1}(\Ecal\times M)$ and $d \oloc^{\text{top}-1,1}(\Ecal\times M)$.
\end{theorem}

This defines a local map\footnote{Note that the map below really should be defined on $T_\phi\Ecal$. Since we are using vector bundles we identify the tangent spaces with the space $\Ecal$ for simplicity.}
\[
\Ecal \to \oloc^{0,\top}(\Ecal\times M), \qquad \psi\mapsto (j^\infty)^*(F_\beta(x^i,u_\alpha,u^\alpha_I)dx^1 \wedge  \cdots \wedge dx^{\text{top}})\psi^\beta
\]
called the source variation of the local form ${\bF}$ in the direction of $\psi$, denoted by $(\delta^\mathrm{src}{\bF})(\psi)$. More precisely

\begin{definition}
The source variation\footnote{This is often called exterior Euler operator \cite{BlohmannLFT}. See also \cite[Remark 3.3]{SchiavinaSchnitzer} for its homological properties.} of a local form is
\[
\delta^{\mathrm{src}}{\bF} = \Pi^{\src}\delta{\bF},
\]
and we write
\[
\delta {\bF} = \delta^{\src}{\bF} + d{\bth}_{{\bF}}
\]
for ${\bth}_{{\bF}}\in\Omega^{1,\top-1}(\Ecal\times M)$.
Then, for a generalised Lagrangian $L$ determined by a local form ${\bF}$ we define the generalised 1-form current $\theta_F$ assocated to $F$ as\footnote{The order will matter for sign reasons, later on.}
\begin{equation}\label{e:smearedtheta}
    \theta_F[df] = \int  df {\bth}_{{\bF}}.    
\end{equation}

Given a generalised current $K$ associated to the local density $\mathbf{K}$ we define its \emph{source derivative} as
\[
\frac{\delta^{\mathrm{src}}K[f]}{\delta \phi} \doteq \int_{\Sigma}f \delta^{\mathrm{src}}\mathbf{K}\colon\quad  \psi \mapsto \frac{\delta^{\mathrm{src}}K[f]}{\delta \phi}(\psi)=\int_{\Sigma}f (\delta^{\mathrm{src}}\mathbf{K})(\psi).
\]    
\end{definition}

\begin{remark}
    Observe that $\boldsymbol\theta_{{\bF}}$ is defined, in principle, up to $d$-closed local forms (which are then $d$-exact). One way to deal with these ambiguity is by choosing a homotopy for the variational bicomplex, as explained in \cite{SchiavinaSchnitzer}, for example the horizontal homotopy of Anderson's \cite{Anderson}. We will glide over this aspect, in this manuscript, as the effects of a redefinition of ${\bth}$ would only show up when investigating codimension $2$ data (see also \cite{RielloSchiavinaPS}). 
\end{remark}

\begin{lemma}
    Let $F$ be a generalised Lagrangian, then, if $\partial M=\emptyset$
    \[
    \frac{\delta F[f]}{\delta \phi} \sim \frac{\delta^{\src} F[f]}{\delta \phi}
    \]
\end{lemma}
\begin{proof}
    We compute
    \begin{multline*}
    \frac{\delta F[f]}{\delta \phi}(\psi) = \int(\delta {\bF}f)(\psi) =  \int((\delta^\src {\bF} + d\theta_{{\bF}})f)(\psi) = \int(\delta^\src {\bF})(\psi)f - \int (\theta_{{\bF}}(\psi)df).
    \end{multline*}
    Then, we have
    \[
    \frac{\delta F[f]}{\delta \phi} = \frac{\delta^{\src} F[f]}{\delta \phi} + X
    \]
    where $\supp(X) \subset \supp(df)$. 
\end{proof}

\begin{definition}[Symplectic data]\label{def:symplecticdata}
    A local symplectic density on the space of sections $\Ecal$ of a graded vector bundle $E\to M$ is an element ${\bom}\in \oloc^{2,\top}(\Ecal\times M)$ such that it is vertically closed $\delta{\bom}=0$ and such that the map
    \[
    {\bom}^\flat_{\src}\doteq \Pi^{\src}\circ {\bom}^\flat \colon \mathfrak{X}_{\loc}(\Ecal) \to \Omega^{1,\top}_{\src}(\Ecal\times M) \qquad X\mapsto \Pi^{\src}\iota_X{\bom}
    \]
    is injective. We denote by $|{\bom}|$ its internal degree, determined by the grading on the vector bundle $E\to M$. The local symplectic density $\bom$ is \emph{ultralocal} if it descends to a 2 form on the 0th jet bundle.\footnote{In other terms, it does not depend on any higher jet of sections, but only on the value of the section at $x$.} A local weak-symplectic form is a two form $\omega\in \Omega^2(\Ecal)$ for which there exist a local symplectic density $\bom$ such that $\omega = \int_M \bom$.
\end{definition}

\begin{definition}
    Let $\Ecal$ be endowed with a local symplectic density ${\bom}$. A local $(0,\top)$ density ${\bF}\in\oloc^{0,\top}(\Ecal\times M)$ is said to be Hamiltonian iff $\delta^{\src}{\bF}\equiv\Pi^{\src}\delta {\bF}$ is in the image of $\omega^\flat_{\src}$, i.e. if there exists $X_{\bF}$  such that
    \[
    \Pi^{\src}\iota_{X_{\bF}} {\bom} = \delta^{\src}{\bF}.
    \]
    This means that:
    \[
    \iota_{X_{\bF}} {\bom} - \delta {\bF} = d{\bth}_{{\bF}}.
    \]
    We will then say that the density $\mathbf{L}$ is Hamiltonian, and that\footnote{\label{fnt:verticalvf}Truly we should consider $\mathrm{pr}(X_{{\bF}}$ as the Hamiltonian vector field. We will abuse notation henceforth and assume that $X_{{\bF}}$ is evolutionary, a.k.a. strictly vertical, and that in particular it commutes with $d: [\iota_{X_{{\bF}}},d]=0$.} $X_{\boldsymbol{F}}$ is the \emph{Hamiltonian vector field} of $\boldsymbol{F}$, and ${\bth}_{{\bF}}$ is its \emph{Hamiltonian one-form}. The triple $(\bF,X_{\bF},\bth_\bF)$ is called a Hamiltonian triple. We denote the internal degrees by $|X_{{\bF}}|$ and $|{\bth}_{{\bF}}|$ and we have
    \[
    |X_{{\bF}}| + |{\bom}| = |{\bth}_{{\bF}}|=|{\bF}|.
    \]
\end{definition}

\begin{lemma}\label{lem:Hamvfequiv}
    Let $F$ be a generalised Lagrangian for the Hamiltonian density $\boldsymbol{F}\in \oloc^{0,\top}(\Ecal\times M)$ with Hamiltonian vector field $X_{\bF}$ and Hamiltonian one form $\bth_\bF$ with respect to a symplectic density ${\bom}$. Let $\omega$ be the weak symplectic form given by $\bom$.
    Then 
    \be\label{e:intHVF}
    X_{F[f]} = \sum_\phi\int (fX_{\boldsymbol{F}}(\phi))\frac{\delta}{\delta \phi},
    \ee
    where $\phi$ runs over the variables in $\Ecal$, is the Hamiltonian vector field of $F[f]$, in the sense that
    \[
    \iota_{X_{F[f]}}\omega = \delta F[f] + \theta[df]
    \]
    for some $1$-form current $\theta$. Moreover, if $\bom$ is ultralocal, then $\theta = \int \bth_{\bF}$.
\end{lemma}
\begin{proof}
    Let us compute
    \[
    \delta F[f] = \int \delta {\bF}f = \int \left(\iota_{X_{{\bF}}}{\bom} f + d{\bth}_{\bF} f\right) = \iota_{X_{F[f]}}\omega  - \int df {\bth}_{\bF}  \sim \iota_{X_{F[f]}}\omega.
    \]
    Indeed, one has, denoting by $\langle\cdot,\cdot\rangle$ the \emph{local} pairing between elements of $T_\phi\Ecal$ and $T_\phi^*\Ecal$
    \[
    \int\iota_{X_{{\bF}}}{\bom} f=\int\langle X_{{\bF}}, \bom\rangle f = \int\langle fX_{{\bF}}, \bom\rangle + \int  df {\bth}'
    \]
    for some ${\bth}'$, and the claimed expression is valid for $\bth =  {\bth}' -{\bth}_{\bF} $ and the associated $1$ form current.
    We can use the additional requirement of ultralocality to argue that $\int\langle X_{{\bF}}(\phi), \omega_\phi\rangle f = \int\langle fX_{{\bF}}(\phi), \omega_\phi\rangle$, i.e. $\bth' = 0$ and $\bth = \bth_\bF$.
\end{proof}
Note that, in what follows, we will always deal with ultralocal symplectic densities, as such are the canonical forms on the (shifted) cotangent bundle of the space of fields.

In Appendix \ref{app:densities} we will expand on useful relations concerning natural operations we can perform on local densities when given a local symplectic form. Our discussion, based on those identities leads to the following (known) result:
\begin{proposition}[Source Bracket]\label{def:sourcebracket}
    There is a graded skew-symmetric bilinear operation on Hamiltonian $(0,\top)$ forms and generalised Lagrangians called the \emph{source bracket}:
    \[
    \{\mathbf{L},\mathbf{G}\}^\src \doteq \iota_{X_{\mathbf{F}}}\iota_{X_{\mathbf{G}}}\bom, \qquad \{F,G\}^{\src}(f) \doteq \int (\iota_{X_{\bF}}\iota_{X_{\bG}}{\bom}) f.
    \]
    Moreover, if $\bom$ is ultralocal we have 
    \[
    \{F,G\}^{\src}(fg) = \iota_{X_{F[f]}}\iota_{X_{G[g]}}\omega.
    \]
\end{proposition}

\begin{remark}
    The source bracket defines a Lie algebra on Hamiltonian generalised Lagrangians only after taking the equivalence classes w.r.t. the equivalence relation $\sim$ of Definition \ref{def:eqrel}. This is because it can be verified that the cyclic permutations $\mathrm{Cyc}_{F,G,H}\{F,\{G,H\}^\src\}^\src (f) \sim 0$, which means that the associated local density is $d$-exact. This was studied in detail in \cite{SchiavinaSchnitzer}, where it was argued that the source bracket induces the (more general) structure of an $L_\infty$ algebra. We refer to the mentioned article (and references therein) for a primer on the subject. In particular, the source bracket coincides with the structure discussed in \cite[Remark 4.22]{SchiavinaSchnitzer} and preceding paragraphs. Note that a discussion of these topics is present also in \cite{BFLS}, and many others discussed related structures, among which we mention \cite{BarnichBrandtHenneaux,BarnichHenneaux}. See also \cite{Grigoriev} and references therein.
\end{remark}

Given a local symplectic density ${\bom}$ one can often also introduce a bilinear operation on generalised Lagrangians, by building a local bivector $\Pi$ on the algebra of Hamiltonian functions. This is manifestly the case on cotangent bundles of spaces of sections, i.e. when $\bom$ is in particular ultralocal. We will refrain from giving sufficient conditions for ${\bom}$ to define such a bivector but we will simply assume it exists. 

\begin{definition}[Standard Bracket]\label{def:stdbracket}
    Suppose the local symplectic density ${\bom}$ defines a local bivector on generalised $0$-currents
    \[
    \Pi\colon \iloc^{1,(0)}(\mathcal{E})\times \iloc^{1,(0)}(\mathcal{E}) \to \iloc^{0,(0)}(\mathcal{E}).
    \]
    We denote the associated biderivation on $\iloc^{0,(0)}(\mathcal{E})$, called the \emph{standard bracket}, by
    \[
    \{\cdot,\cdot\}^{\std}\colon \iloc^{0,(0)}(\mathcal{E})\times \iloc^{0,(0)}(\mathcal{E})\to \iloc^{0,(0)}(\mathcal{E}),
    \]
    namely
    \[
    \{F[f],G[f]\}^{\std} = \Pi(\delta F[f], \delta G[f])
    \]
    such that for any local one-form $\alpha$
    \[
    \Pi(\delta^{\src}F[f],\alpha) = \iota_{X_{F[f]}}\alpha.
    \]
\end{definition}

When both these bilinear operations exist on local forms on $\Ecal$ we have the following relation between them

\begin{align}\label{e:fullsrcBrel}
\{F[f],G[f]\}^{\std} &\sim \{F[f],G[f]\}^\src 
    \sim \frac12 \left(X_{F[f]}G[f] -\sigma_\omega({\bF},{\bG})X_{G[f]}F[f]\right)
\end{align}
A sharper version of this relation is given (and proven) in Prop.~\ref{prop:fullbracketvsaction} in Appendix~\ref{app:densities}.

\subsection{BV formalisms}
In this section we will recall the basics of the Batalin--Vilkovisky and the Batalin--Fradkin--Vilkovisky formalisms \cite{BV1,BV2}, as well as the relationship between the two, which often goes under the acronym of BV-BFV formalism. The specific link we are employing here is due to Cattaneo, Mnev and Reshetikhin \cite{CMR1,CMR2}, although several authors have investigated similar questions before, see e.g. \cite{KugoOjima,HenneauxTeitelboim,BarnichBrandtHenneaux,BarnichHenneaux, BarnichGrigoriev,Hollands,Grigorievparent}.

In the BV framework, one extends the space of classical field configurations $\Ecal$ by auxiliary fields, often called ghosts, antighosts and Nakanishi--Lautrup fields to encode gauge-fixing through a modification of the generalised Lagrangian of the theory.\footnote{The extension by Nakanishi--Lautrupp fields and antighosts is often called \emph{non-minimal extension}. It is not strictly necessary but it is often convenient to implement gauge-fixing explicitly.} This procedure allows us to geometrise the BV complex using Hamiltonian dg manifolds. It can be seen as the derived critical locus of a generalised Lagrangian, or the combination of the Koszul--Tate and Chevalley--Eileberg complexes for the (formal) moduli problem associated to the equations of motion of the generalised Lagrangian. (See e.g.\ \cite{Steffens}.)

From now on $\Ecal$ will denote the extended BV configuration space, which is the densitised shifted cotangent bundle of the graded vector space:

\[
\Ecal_0\oplus \Ecal_1[1]\oplus\dots\oplus \Ecal_k[k]\,
\]
where $\Ecal_i=\Gamma(M,E_i)$ for some vector bundles $E_i$, $i=1\dots k$. 
\begin{remark}
Note that the ``true'' contangent bundle of a Frech\'et space $\Xcal$ is the trivial bundle $\Xcal\times\Xcal'$, where $\Xcal'$ is the topological dual equipped with the strong topology. Here we use $T^!\Xcal\doteq\Xcal\times \Xcal^!$, instead. See \cite{KrieglMichor} for more details, and \cite{RielloSchiavina} for a review relevant to the applications contained here.
\end{remark}
Explicitly, we define the space of BV configurations by
\[
\Ecal\doteq T^![-1]\left(\Ecal_0\oplus\dots\oplus \Ecal_k[k]\right)\doteq \Ecal_0\oplus\dots\oplus \Ecal_k[k]\oplus \Ecal^!_0[-1]\oplus\dots\oplus \Ecal^!_k[-k-1]\,
\]
where $\Ecal_i^!\doteq\Gamma(M,E^*\otimes {\rm Dens}(M))$ and ${\rm Dens}(M)$ is the density bundle over $M$. 

An element of $\oplus_i(\Ecal_i)$ will be denoted by $\ph$ and we label its components by $\ph^\alpha$ (in a given basis, so that $\alpha$ runs through all the components of all the fields), while the corresponding antifields in $\Ecal_i^!$ are denoted by $\ph_\alpha^\ddagger$.

On $\Ecal$ there is a canonical, degree $(-1)$ ultralocal symplectic density:
\[
{\bom}_{BV}\in \oloc^{2,\mathrm{top}}(\Ecal\times M)\qquad {\bom}_{BV} = \delta \varphi_\alpha\wedge\delta \varphi^\ddagger_\alpha.
\]
One can also define a 1-shifted Poisson algebra structure on $C^\infty_{\loc}(\Ecal)$ as (see also \cite{Anselmi_1994,CattaneoFiorenzaLongoni,CattaneoSchaetz-super,costello2011renormalization,Qiu2011,FredenhagenRejzner2015}) as in Definition \ref{def:stdbracket}:
\begin{equation}\label{e:antibracket}
\{.,.\}^{\std}\doteq\left<\rpar{}{\ph^\ddagger_\alpha},\lpar{}{\ph^\alpha}\right> + \left<\rpar{}{\ph^\alpha},\lpar{}{\ph^\ddagger_\alpha}\right>\,.
\end{equation}

\begin{definition}[BV theory]
A BV theory over a manifold $M$ is the data $(\Ecal,\Omega,L)$ of a graded, local $(-1)$ symplectic vector space $(\Ecal,\omega)$, modeled over the space of sections of a graded vector bundle $\Ecal = \Gamma(\Sigma,E\otimes E^* \otimes \mathrm{Dens}(M))$ together with a degree $0$ Hamiltonian generalised Lagrangian $L$ on $\Ecal$ that satisfies the classical master equation 
\[
\{L[g],L[g]\}^\std\sim 0,
\]
for any test function $g\in C^\infty_c(M)$, where $\{\cdot,\cdot\}^\std$ is the bracket on Hamiltonian generalised Lagrangians induced by $\omega$.
The BV operator is the cohmological vector field $Q\equiv X_{\mathbf{L}}$ associated to the density $\mathbf{L}$ defining $L[g] = \int_{M} \mathbf{L} g$.
\end{definition}

\begin{remark}[On the need for cutoffs]
    Our philosophy on the definition of the fundamental fields of the theory is that one should consider all smooth sections, without compact support requirements. This has the advantage of allowing us to work with Fr\'echet spaces even in the noncompact setting, in the case of vector bundles. However, to make sense of ``action functionals'' one needs to smear the Lagrangian density with a suitable test function or cutoff. 
    
    It is important to notice that this \emph{generically} leads to a notion of classical master equation that holds only up to boundary terms, meaning that in the classical master equation that is commonly used in pAQFT is only valid for compactly supported fields and/or regions, while the densitised or \emph{modified} classical master equation (studied, among others, by Cattaneo, Mnev and Reshetikhin in \cite{CMR1}) is the true equation that physically relevant models actually satisfy.\footnote{With the important exception of the scalar field.} We will clarify this in the next section (see Proposition \ref{prop:CMEequivalence}).

    Surprisingly, the fact that generically CME fails due to boundary terms has been largely ignored by the (perturbative) quantisation community, or dealt with via the imposition of ``boundary conditions''. It is now increasingly clear that neglecting such terms cannot be a satisfactory approach to QFT. We shall not ignore them in what follows, but will need to devise an appropriate framework to deal with them in perturbation theory. 
\end{remark}

\subsection{Generalisations of BV formalism}
In the previous section we have introduced the BV formalism, by looking at the data of a $-1$ shifted symplectic (Hamiltonian) dg manifold. More generally, one can introduce the BF${}^k$V formalism, by looking at (sections of) graded vector bundles $\Ecal^{(k)}$ endowed with a local $(k-1)$-shifted symplectic density ${\bom}^{(k)}$, i.e. vertically closed $\delta \bom^{(k)}=0$ and such that the map $({\bom}^{(k)})^\flat\colon T\Ecal^{(k)}\to \oloc^{1,\top}(\Ecal^{(k)}\times M)$ is injective, with $k=0$ corresponding to the BV scenario. This is naturally attached to higher codimension submanifold, when looking at local field theory data \cite{CMR1,CMR2,MSW} (see also \cite{Grigoriev} and therein). 

\begin{definition}\label{def:densCME}
    A degree $k$, local $(0,\top)$ Hamiltonian density $\mathbf{L}^{(k)}$ on a $(k-1)$-shifted symplectic space $(\Ecal^{(k)},{\bom}^{(k)})$ is said to \emph{satisfy the densitised classical master equation} iff
    \begin{subequations}\label{e:dCME}\begin{equation}\label{e:dCME1}
    \Pi^\src\delta\{\mathbf{L},\mathbf{L}\}^\src \equiv \Pi^{\src}\delta\iota_{X_\mathbf{L}^{(k)}}\iota_{X_\mathbf{L}^{(k)}}{\bom}^{(k)} = 0,
    \end{equation}
    where $X_\mathbf{L}^{(k)}$ is the Hamiltonian vector field of $\mathbf{L}^{(k)}$: 
    \[
    \Pi^{\src}\iota_{X_\mathbf{L}^{(k)}} {\bom}^{(k)} = \Pi^{\src}\delta \mathbf{L}^{(k)},
    \]
    that is to say, iff there exists a degree $k+1$ local densitiy $\mathbf{J}^{(k+1)}\in\oloc^{0,\top -1}(\Ecal^{(k)},{\bom}^{(k)})$ such that
    \begin{equation}\label{e:dCME2}
    \frac12\{\mathbf{L}^{(k)},\mathbf{L}^{(k)}\}^\src = d \mathbf{J}^{(k+1)}.
    \end{equation}\end{subequations}
    Assuming that $\bom^{(k)}$ defines a degree $k+1$ Poisson algebra with bracket $\{\cdot,\cdot\}^\std$, a degree $k$ generalised Lagrangian $L^{(k)}$ on $\Ecal^{(k)}$ is said to \emph{satisfy the standard classical master equation} iff one of the two equivalent conditions hold
    \begin{subequations}\label{e:CME}
    \begin{gather}
    \supp(\{L^{(k)}[f],L^{(k)}[f]\}^\std) \subset \supp(df) \\
    \{L^{(k)}[f],L^{(k)}[f]\}^\std\sim 0,
    \end{gather}\end{subequations}
\end{definition}

We shall see that these notions are strongly related to one another, owing to the following result.

\begin{proposition}[Classical master equations]\label{prop:CMEequivalence}
    Let $\mathbf{L}^{(k)}\in \oloc^{0,\top}(\Ecal^{(k)}\times M)$ be a degree $k$ Hamiltonian local density on a locally $(k-1)$-shifted symplectic manifold $(\Ecal^{(k)},\bom^{(k)})$, with $L^{(k)}$ the generalised Lagrangian associated to $\mathbf{L}^{(k)}$, and $Q\equiv X_{\mathbf{L}^{(k)}}$ its Hamiltonian vector field. Then, the following are equivalent: 
    \begin{enumerate}
        \item $\mathbf{L}^{(k)}$ satisfies the densitised classical master equation (Equations \ref{e:dCME}),
        \item $L^{(k)}$  satisfies the standard classical master equation\footnote{Whenever it makes sense, i.e. assuming $\{\cdot,\cdot\}^\std$ is well defined.} (Equations \ref{e:CME}), 
        \item $Q\equiv X_{\mathbf{L}^{(k)}}$ and $Q[f]\equiv X_{L^{(k)}[f]}$ (Equation \ref{e:intHVF}) are degree 1 cohomological vector fields.
    \end{enumerate}
    Thus, we will henceforth conflate the two notions and simply say that $L^{(k)}$ or $\mathbf{L}^{(k)}$ satisfies the classical master equation (CME).
\end{proposition}

\begin{proof}
    The equivalence of $(1)$ and $(2)$ is a consequence of Corollary \ref{cor:srcnormalbracket}, under the tacit assumption that $\bom^{(k)}$ satisfies the conditions of Definition 
    \ref{def:stdbracket}. Turning to point $(3)$, let us drop the $k$ index for ease of notation. We have that
    \[
    \frac12\iota_{[Q,Q]}{\bom} = \iota_Q\delta \iota_Q{\bom} - \frac12\delta \iota_Q\iota_Q{\bom} \stackrel{\star}{=} \iota_Q\delta (\delta \mathbf{L} + d\theta_\mathbf{L}) - \delta d\mathbf{J} = d(-\iota_Q\delta \theta_\mathbf{L} + \delta\mathbf{J}),
    \]
    where $\stackrel{\star}{=}$ denotes that we used the densitised CME, together with $\iota_Q{\bom} = \delta \mathbf{L} + d\theta_\mathbf{L}$ for some $\theta_\mathbf{L}\in\oloc^{1,\mathrm{top}-1}(\Ecal\times M)$. This implies that, since $\Omega$ is source injective,
    \[
    \Pi^{\src}\iota_{\frac12 [Q,Q]}\omega = 0 \Leftrightarrow [Q,Q]=0.
    \]
    Hence $Q$ is a cohomological vector field iff $\mathbf{L}$ satisfies the densitised CME owing to the injectivity of $\omega$. This, together with the fact that $Q$ is cohomological iff $Q[f]$ is cohomological, shows that $(2)\iff (3)$. 
\end{proof}

\begin{definition}
    A BF${}^k\!$V theory on a manifold $M$ is the data $(\Ecal^{(k)},{\omega}^{(k)},{L}^{(k)})$ composed of a graded vector space\footnote{The space of sections of a graded vector bundle over $M$.} $\Ecal^{(k)}$ endowed with a degree $(k-1)$ local symplectic form ${\omega}^{(k)}$ (Definition \ref{def:symplecticdata}) and a degree $k$ Hamiltonian generalised Lagrangian $L^{(k)}$ that satisfies the classical master equation. The (classical) BF${}^k\!$V operator is the Hamiltonian vector field $Q^{(k)}$ of the generalised Lagrangian $L^{(k)}$ and it defines the BF${}^k\!$V complex $\left(\mathfrak{BFV}=C^\infty(\Ecal^{(k)}),Q^{(k)}\right)$. A BF${}^k$V theory is said to be \emph{exact} if ${\omega}^{(k)}$ is $\delta$-exact.
\end{definition}

\begin{remark}
    The terminology BF${}^k\!$V theory is inspired by the fact that when $k=0$ one obtains a BV theory as defined above, while for $k=1$ one obtains what is usually called a BFV theory. The generalisation is then straightforward.
\end{remark}

Note that, in virtue of Proposition \ref{prop:CMEequivalence}, the generalised Lagrangians defining BF${}^k\!$V theories satisfy a weaker version of what is sometimes taken to be the classical master equation, which otherwise requires $\{L^{(k)},L^{(k)}\}^\std = 0$ and not just up to the equivalence relation $\sim$ (cf. Definition \ref{def:eqrel} and Equations \ref{e:CME}). Precisely this fact is what allows us to link a BF${}^k\!$V and a BF${}^{k+1}\!$V theory \cite{CMR1,CMR2} (see also \cite{MSW} for a density version and \cite{CattaneoMnevGluing}).

\begin{definition}\label{def:relBVtheory}
    A BF${}^k\!$V theory $(\Ecal^{(k)},{\omega}^{(k)},{L}^{(k)})$ on the manifold $M$ is said to be relative to an exact BF${}^{(k+1)}\!$V theory $(\Ecal^{(k+1)},{\omega}^{(k+1)},L^{(k+1)})$ on a manifold $\Sigma$, with ${\omega}^{(k+1)}=\delta {\theta}^{(k+1)}$, iff there exists a local map $(\pi^{(k)})\colon\mathcal{E}^{(k)}\to \mathcal{E}^{(k+1)}$ such that\footnote{Here we implicitly assume that $\omega^{(k)}$ defines a bracket on functionals. See Definition \ref{def:stdbracket}.}
    \begin{subequations}\label{e:BVBFV}\begin{align}\label{e:BVBFV1}
    \iota_{Q^{(k)}}{\omega}^{(k)} = \delta {L}^{(k)} +  (\pi^{(k)})^*{\theta}^{(k+1)}, \\\label{e:BVBFV2}
    \frac12\{L^{(k)},L^{(k)}\}^{\std} = (\pi^{(k)})^*{L}^{(k+1)}
    \end{align}\end{subequations}
    and it is a chain map, i.e.\ $(\pi^{(k)})^*Q^{(k+1)} = Q^{(k)} (\pi^{(k)})^*$.
    Given relative BF${}^k\!$V data as above, the \emph{generalised BF${}^k\!$V Noether current} is the generalised current\footnote{Note the mismatch in $k$: the $k$th Noether current refers to the BF${}^k\!$V theory, but uses $(k+1)$-labeled functionals. This is so that the $0$th Noether current returns the standard (BV) Noether current.}
    \[
    \mathbb{J}^{(k+1)}={L}^{(k+1)} - \iota_{Q^{(k+1)}}{\theta}^{(k+1)}.
    \]
\end{definition}

\section{BV formalism in pAQFT}\label{sec:BVnoboundary}
In pAQFT one starts with quantising the free theory (a quadratic Lagrangian), and introduces the interaction perturbatively. We begin with a brief review of the basic structures needed for quantising a quadratic Lagrangian that contains both fields and antifields and solves the classical master equation. The validity of CME is what differentiates this section from the rest of the paper, where we shall drop that assumption, and instead focus on corrections coming from boundary terms.

\subsection{Free-interacting splitting}
The generalised BV-Lagrangian $L$ splits into a (choice of a) quadratic term $L_0$ and the remainder, which we will call $V$ and generally assume to be a polynomial. Using the chosen quadratic part of $L$, we define the linearised BV differential
$Q_0$ (its Hamiltonian vector field). This splitting is not unique, and it is given in the following terms.

\begin{definition}[Free-interacting splitting]
    Let $L$ be a generalised Lagrangian. A \emph{free-interacting splitting} of $L$ is a choice of a quadratic generalised Lagrangian $L_0$, dubbed the \emph{free Lagrangian} such that 
    \[
    L = L_0 + V
    \]
    where the \emph{interaction Lagrangian} $V$ is defined by a polynomial local density $\mathbf{V}\in \oloc^{0,\top}(\Ecal\times M)$ through
    \[
    V=\int_M \boldsymbol{V}f\,.
    \]
\end{definition}

We will be interested in two particularly interesting splittings:
\begin{proposition}\label{prop:quad+minimalsplit}
    Let $L$ be a polynomial generalised Lagrangian that satisfies the classical master equation. 
\begin{enumerate}[label = (\roman*)]
      \item If $L_0$ is at most linear in non-zero degree fields, $L_0\equiv L_{\mathrm{lin}}\doteq L - L^{\geq 3}$, we talk about the \emph{linear BV splitting}. 
      The Hamiltonian vector field $Q_0\equiv Q_{\mathrm{lin}}$ is cohomological and it is called the linear, free BV operator.
      \item If $L_0$ is the ``degree zero'' part of $L$, i.e.\ $L_0\equiv L_{00}\doteq L_{\mathrm{lin}}\vert_{\mathrm{Body}(\Ecal)}$ we talk about the \emph{minimal splitting}, and the associated cohomological vector field is \emph{the Koszul--Tate differential} $Q_0\equiv Q_{00} \equiv \delta_{KT}$ associated to the critical locus of $L_{00}$ (i.e. the 0-th cohomology characterizes the space of solutions to the equations of motion).
\end{enumerate}

\end{proposition}
\begin{proof}
    (i) follows from Proposition \ref{prop:CMEequivalence}. Indeed $\{L,L\} \sim 0$ implies $\{L_0,L_0\}\sim 0$ for $\{L_0,L_0\}\equiv\{L_{\mathrm{lin}},L_{\mathrm{lin}}\}$ is the lowest polynomial order (constant) in $\{L,L\}$. Hence $L_0$ satisfies the CME, its associated density $\mathbf{L}_0$ satisfies the densitised CME, and its Hamiltonian vector field is cohomological.

    Now for (ii), $L_0\equiv L_{00}$ and thus $L_{00}$ is a purely classical generalised Lagrangian (i.e. no antifields are present) so $\{L_{00},L_{00}\}=0$ trivially and the Hamiltonian vector field $Q_0\equiv Q_{00}$ is again a cohomological vector field by Proposition \ref{prop:CMEequivalence}. $Q_{00}$ contains only derivatives with respect to fields and its coefficients are given by the variations of $L_{00}$, so it is indeed the Koszul--Tate differential associated with the critical locus of $L_{00}$, that is,
    \[
    Q_0\equiv Q_{00} = \iota_{\delta L_{00}} \equiv \delta_{KT}.
    \]
\end{proof}

\begin{remark}\label{rmk:BRSTtype}
    An important class of gauge theories are given by classical local Lagrangians that are invariant under the action of a gauge \emph{group}. We call these theories ``of BRST type'' (cf.\ \cite[Definition 30]{MSW}). In those cases, the BV generalised Lagrangian can be chosen in the particularly simple form 
    \[
    L=L_{00} + \langle Q_{\mathsf{BRST}}\varphi,\varphi^\dag\rangle,
    \]
    where $Q_{\mathsf{BRST}}$ is the Chevalley--Eilenberg differential for a Lie algebra action, seen as a cohomological vector field. Then, the free BV operator further splits into 
    \[
    Q_0=\delta_0+\gamma_0\,
    \]
    where $\delta_0=X_{L_{00}}$ is of antifield degree -1 and, denoting by $L_{01}$ the part of the linear splitting $L_{\mathrm{lin}}$ (linear in nonzero degree fields), we deduce that $\gamma_0=X_{L_{01}}$ is of antifield degree $0$. We recall that $Q_0^2=0$, and expand in antifield degree. This implies that:
    \be\label{e:deltagamma}
    \delta_0\gamma_0+\gamma_0\delta_0=0,\qquad \delta_0^2=0,\qquad \gamma_0^2=0.
    \ee
    
    More generally, gauge transformations may be given in terms of a tangent distribution $\mathfrak{G}\subset T\Ecal$, which may not be involutive, i.e. $[\mathfrak{G},\mathfrak{G}]\not\subset\mathfrak{G}$ (this in particular means that the do not form a Lie subalgebra of vector fields).
    By general arguments, this ``gauge distribution'' is involutive only \emph{on shell} \cite{HenneauxTeitelboim}, and one needs the full power of the BV formalism (see, e.g.\ \cite{Stasheff,mnevlectures}). 
    
\end{remark}

\subsection{Free BV theory on globally hyperbolic manifolds}
We are now going to specialize the discussion to field theory cast on a globally hyperbolic manifold $(M,g)$. The Lorentzian structure of $M$, and its interaction with sections of a vector bundle and functionals thereon, allows us to refine our study of functionals on $\Ecal$, according to the fine structure of their wavefront sets.

\subsubsection{Free classical theory}
From now on, we assume that we are given a generalised Lagrangian density $L$, which defines a differential operator $D:\Ecal\rightarrow \Ecal^{!}$ 
\be\label{e:P}
D_{\gamma\alpha}(x)=\lpar{}{\ph^\gamma}L_{00}[f]\rpar{}{\ph^\alpha}
\ee
where $L_{00}$ is the free part of the minimal free-interacting splitting of $L$, and we will assume that $D$ is Green hyperbolic, with $\Delta^{R/A}$ its retarded/advanced Green's functions (Definition \ref{def:propagators}). (Note that on graded manifolds one can define left and right derivations $d_{L/R}$ according to the order in which one takes the Leibniz rule, namely $d_R(a\otimes b) = a \otimes d_R(b) + (-1)^{|b||d_R|} d_R(a) \otimes b$.)

Moreover, let us define $K:\Ecal\rightarrow \Ecal$ as
\be\label{e:K}
K^{\gamma}_{\ \alpha}(x)=\lpar{}{\ph^\ddagger_\gamma} L_{01}[f] \rpar{}{\ph^\alpha} \,
\ee
where $L_{01}$ is the linear part in antifields of $L_{\mathrm{lin}}$, the free part of $L$ in the quadratic free-interacting splitting, and $f\equiv 1 $ on the support of $x$. Without loss of generality we may assume:
\[
L_{00}[f]=\frac{1}{2}\int \ph^\gamma D_{\gamma\alpha} \ph^\alpha\,\qquad L_{01}[f]=\int \ph^\ddagger_\gamma K^{\gamma}_{\ \alpha} \ph ^\alpha\,,
\]
and the expression of $\delta_0$ and $\gamma_0$ in terms of $D$ and $K$ is as follows:
\be\label{e:koszulandCE}
\delta_0 F = \rpar{L_{00}}{\ph^\alpha} \lpar{}{\ph^\ddagger_\alpha}, \qquad \gamma_0 F =  \rpar{L_{01}}{\ph^\beta} \lpar{}{\ph^\ddagger_\beta}
\ee
From \eqref{e:deltagamma} we can derive the following identity.

\begin{proposition}
The operators $D$ and $K$ as above satisfy the following identity:
\[
\left<\phi,(DK+K^\dagger D) \psi\right> =0\,
\]
where $\phi\in\Ecal$, $\psi\in\Ecal_c$.
\end{proposition}
\begin{proof}
First we find that for any compactly supported functional $F$, we have
\begin{align*}
\delta_0\circ \gamma_0 F&= \rpar{L_{00}}{\ph^\alpha} \lpar{}{\ph^\ddagger_\alpha} L_{01} \rpar{}{\ph^\beta} \lpar{}{\ph^\ddagger_\beta} F+\dots\\
\gamma_0\circ \delta_0 F&= \rpar{L_{01}}{\ph^\ddagger_\alpha} \lpar{}{\ph^\alpha} L_{00} \rpar{}{\ph^\beta} \lpar{}{\ph^\ddagger_\beta} F+\dots\,\\
\end{align*}
where the terms abbreviated by dots will cancel when taking the sum of $\delta_0\circ \gamma_0$ and $\gamma_0\circ \delta_0$, so we get:
\begin{multline*}
(\delta_0\circ \gamma_0+\gamma_0\circ \delta_0) F=\rpar{L_{00}}{\ph^\alpha} \lpar{}{\ph^\ddagger_\alpha} L_{01} \rpar{}{\ph^\beta} \lpar{}{\ph^\ddagger_\beta} F+\rpar{L_{01}}{\ph^\ddagger_\alpha} \lpar{}{\ph^\alpha} L_{00} \rpar{}{\ph^\beta} \lpar{}{\ph^\ddagger_\beta}F \\
=\int  \ph^{\gamma} D_{\gamma\alpha} K^{\alpha}_{\ \beta} \lpar{F}{\ph^\ddagger_\beta}+\int(K^{\alpha}_{\ \delta}\ph^\delta) P_{\alpha\beta} \lpar{F}{\ph^\ddagger_\beta}\\
=\int  \ph^{\gamma} D_{\gamma\alpha} K^{\alpha}_{\ \beta} \lpar{F}{\ph^\ddagger_\beta}+\int \ph^\gamma {K^\dagger}_{\gamma}^{\ \alpha}P_{\alpha \beta}\lpar{F}{\ph^\ddagger_\beta}=\left<\phi,(DK+K^\dagger D) \lpar{F}{\ph^\ddagger}\right>\,
\end{multline*}
where in calculating $\rpar{L_{01}}{\ph_\alpha^\ddagger}$ we used the fact that $L_{01}$ is of ghost degree 1 and in the integrand the total degree  (ghost degree plus the form degree) of $\ph^\ddagger_\gamma$ and $K^{\gamma}_{\ \alpha} \ph ^\alpha$ have to differ by one modulo 2 (the integrand has to have top form degree). In the penultimate step we used the fact that
\[
\left<K\ph, D\psi\right>=\left<\ph,K^\dagger D\psi\right>=\int \ph^\gamma {K^\dagger}_\gamma^{\ \beta} D_{\beta\alpha} \psi^\alpha\,
\]
for $\ph\in\Ecal(M)$, $\psi\in\Ecal_c(M)\equiv \Gamma_c(E\rightarrow M)$, where $\Gamma_c$ denotes compactly supported sections. Note that in the end we can replace $\lpar{F}{\ph^\ddagger}$ with an arbitrary compactly supported section 
$\psi$. This concludes the proof.
\end{proof}

\begin{proposition}
Let $D$ be defined as in Equation \eqref{e:P}. Assume that $D$ is Green-hyperbolic operator, and let $\Delta^R$ be its retarded Green's function, as above. We have:
\[
\left<f,K \Delta^R g\right>+\left<f,\Delta^R K^\dagger g\right>\,
\]
where $f,g\in \Ecal^!_c$ and the pairing the is the natural duality paring between $ \Ecal^!_c$ and $\Ecal$, given by integration. \footnote{Note that $\Delta^{R}$ is the (formal) adjoint of $\Delta^A$ with respect to that pairing.}
\end{proposition}

From the properties of $\Delta^{R/A}$ follows that $\Delta^{\rm D}$ is symmetric and $\Delta^{PJ}$ is antisymmetric, and we obtain:
\begin{proposition}\label{prop:Kidentity}
The operator $K\Delta^{\rm D}$ is antisymmetric, i.e.:
\[
K \Delta^{\rm D}+\Delta^{\rm D} K^\dagger=0.
\]
\end{proposition}

\begin{definition}[Peierls bracket]
    The Peierls bracket associated to the Green Hyperbolic operator $D$ is the antisymmetric bilinear map
    \[
    \Pei{\cdot}{\cdot}\colon \Fcal_{\ml}\times \Fcal_{\ml} \to \Fcal
    \]
    defined by
    \[
    \Pei{F}{G}\doteq \langle F, \Delta G\rangle.
    \]
\end{definition}

\begin{remark}[Equicausal Functionals]\label{rmk:equicausal}
It is natural to ask on what subalgebra of all functional $\Fcal$ the Peierls bracket closes, since it is not obvious that $\Pei{F}{G}$ is going to be a multilocal functional even when both $F$ and $G$ are, due to the nontrivial singularities of the bidistribution $\Delta$. It turns out that one can look at \emph{equicausal funtionals} (as introduced in \cite{HRV}), which form a (large enough) subalgebra of $\Fcal$ (as a commutative, associative algebra) that makes the Peierls bracket a Poisson bracket.  We denote equicausal functionals by $\Fcal_{\rm ec}\subset\Fcal$ and defer their definition to \cite{HRV}. Note that all multilocal functionals are equicausal: $\Fcal_{\rm ml}\subset\Fcal_{\rm ec}$, however we generally only have $[\Fcal_\ml,\Fcal_\ml]\subset \Fcal_{\rm ec}$. 
\end{remark}

\subsubsection{Free quantum theory}\label{sec:freequantum}
The pAQFT strategy aims to quantize the theory by deforming the Poisson algebra defined by the Peierls bracket on equicausal functionals into a star product. More abstractly, we are interested in knowing the resulting algebra $\fA$ (with the operator product denoted by $ \bullet$) up to isomorphism and concrete realizations of this algebra in terms of spaces of functionals equipped with star products are obtained by pulling back via different quantisation maps (see e.g. \cite{HR} for review in the context of pAQFT). 

Here are the important properties required of $\mathfrak{A}$ (following closely \cite{HR}):
\begin{itemize}
\item $\mathfrak{A}$ is a free module of $\mathbb{C}[[\hbar]]$.
\item There is a surjective homomorphism $\mathcal{P} : \mathfrak{A} \to \mathcal{F}_{\rm ec}$
(evaluation at the classical limit).
\item $\ker \mathcal{P} = \hbar \mathfrak{A}$.
\item $\mathcal{P}\!\left(\frac{1}{i\hbar}[A,B]_{\bullet}\right)
= \Pei{\mathcal{P}(F)}{\mathcal{P}(G)}$
\end{itemize}
\begin{definition}
    A \emph{quantisation map} is a linear map 
$\mathfrak{Q} : \mathcal{F}_{\rm ec} \to \mathfrak{A}$ such that
\begin{itemize}
\item $\mathcal{P} \circ \mathfrak{Q} = \mathrm{id} : \mathcal{F}_{\rm ec} \to \mathcal{F}_{\rm ec}$, and
\item the image of the $\mathbb{C}[[\hbar]]$-linear extension of $\mathfrak{Q}$ is closed under the noncommutative product in $\mathfrak{A}$.
\end{itemize}
\end{definition}
We will use the notation $\fA_{\reg}$ for the case where we restrict ourselves to deforming the Poison algebra of regular functionals $(\Fcal_{\reg},\Pei{\cdot}{\cdot})$ and the corresponding star products are defined on $\Fcal_{\reg}$.

A quantisation map
$\mathfrak{Q} : \mathcal{F}_{\rm ec} \to \mathfrak{A}$
induces a star product on $\mathcal{F}_{\rm ec}[[\hbar]]$ by the condition
\be
\mathfrak{Q}(F \star G)
= \mathfrak{Q}(F) \bullet \mathfrak{Q}(G).
\ee
Similarly, when restricted to $\mathcal{F}_{\rm reg}$, we obtain a star product on $\Fcal_{\reg}[[\hbar]]$.

Following \cite{HR} we will now introduce a concrete realization of $\fA$. Note that given a Green hyperbolic operator $D$, different choices of the Hadamard 2-point function $\Delta^+=\frac{i}{2}\Delta+H$ differ by its symmetric part, i.e. $H$. Let ${\rm Had}_D$ denote the set of all such choices for the given Green hyperbolic operator $D$. For each $H\in {\rm Had}_D$ we can equip $\Fcal_{\rm ec}[[\hbar]]$ with a star product, according to the following prescription.
\begin{definition}\label{def:wickstarprod}
    Let $F,G\in\Fcal_{\rm ec}[[\hbar]]$, let $D$ be a Green-hyperbolic operator and $H\in {\rm Had}_D$. The Hadamard star product of $F$ and $G$ is
    \[
    \star_H\colon \Fcal_{\rm ec}[[\hbar]]\times \Fcal_{\rm ec}[[\hbar]] \to \Fcal_{\rm ec}[[\hbar]],\qquad F\star_H G \doteq m \circ e^{\hbar D_+} (F\otimes G)
    \]
    where
    \[
    D_+\doteq \left\langle \Delta^+, \frac{\delta}{\delta \varphi}\otimes \frac{\delta}{\delta \varphi}\right\rangle,
    \]
    The $\star_H$-commutator is 
    \[
    [F,G]_{\star_H} = F\star_H G - G\star_H F.
    \]
\end{definition}
For two such $H$ and $H'$ there is a map $\alpha_{H'-H}:\Fcal_{\rm ec}[[\hbar]]\rightarrow \Fcal_{\rm ec}[[\hbar]]$ defined by
\[
\alpha_{H'-H}=e^{\frac{\hbar}{2} \mathcal{D}_{H'-H} }\,,
\]
where $\Dcal_K$ for an integral kernel $K$ is given by
\[
\mathcal{D}_{K} \doteq \left<K,\frac{\delta^2}{\delta\ph^2}\right>\,.
\]
The map defined above intertwines star products for different $H\in {\rm Had}_D$.
\[
F\star_{H'} G= \alpha_{H'-H}(\alpha_{H'-H}^{-1} F\star_H \alpha_{H'-H}^{-1} G)\,.
\]
Following \cite{HR}, we now define $\fA$.
\begin{definition}
Denote by $A_H\doteq A(H)$ for $A: \mathrm{Had}_D \to \mathcal{F}_{\rm ec}[[\hbar]]$.
    \[
\mathfrak{A}
:= \left\{
A : \mathrm{Had}_D \to \mathcal{F}_{\rm ec}[[\hbar]]
\;\middle|\;
\forall H, H' \in \mathrm{Had}_D,\;
A_{H'} = \alpha_{H' - H} A_H
\right\}\,.
\]
We define the product by

\[
(A \bullet B)_H := A_H \star_H B_H\,,
\]
and the involution by 

\[
(A^*)_H := (A_H)^* .
\]
\end{definition}
Maps $\alpha_{H}$ preserve support of functionals, so one can define
\[
\supp A\doteq \supp A_H\,,
\]
for any $H\in {\rm Had}$.

\begin{remark}[Presentations]\label{rmk:presentations}
One can see that the evaluation of $A$ at $H$ is equivalent to the composition with the inverse quantisation map: $A_H=\mathfrak{Q}_H(A)$. Then, changing the Hadamard state gives rise to $\mathfrak{Q}_H: \Fcal_{\rm ec}\rightarrow \fA$, given by 
 \[
(\mathfrak{Q}_H(F))_{H'}=\alpha_{H'-H}F, \qquad \alpha_{H'-H} = \mathfrak{Q}_{H'}^{-1} \circ\mathfrak{Q}_H\,.
 \]
Quantisation by means of a Hadamard state is often called \emph{Wick quantisation}. On the other hand, \emph{Moyal--Weyl quantisation} corresponds to the choice of star product given by $\frac{i}{2}\Delta$. The relationship between the two is given by
 \[
(\mathfrak{Q}_{MW}(F))_{H}=\alpha_{H}F\,,
 \]
and is related to $\mathfrak{Q}_H$ by
\[
\mathfrak{Q}_{MW}=\mathfrak{Q}_H\circ \alpha_{H}\,.
\]
Different star products on functionals can be obtained by pulling back $\bullet$ along different quantisation maps, and they are intertwined by the maps $\alpha$. Moreover, one can realise other abstract algebraic structures defined on $\mathfrak{A}$ by pulling back to functionals via the various quantisation maps. These will be called \emph{presentations}.
\end{remark}

\subsubsection{Regular time-ordering}
In addition to the star product, one wishes to endow functionals with another binary operation: the time-ordered product, needed to introduce interactions using an algebraic version of the path integral approach to quantisation.
We define it for regular functionals first.

\begin{definition}[Regular time ordered product]\label{def:regTimeprod}
    Let $F,G\in \Fcal_{\reg}[[\hbar]]$ be two regular functionals. Their time ordered product is
    \[
    F\cdot_\TT G \doteq m\circ e^{\hbar D_F} F\otimes G \equiv \sum_{k=1}^\infty \frac{\hbar^k}{k!}\langle F^{(n)}, (\Delta^{\rm F})^{\otimes n}, F^{(n)}\rangle,
    \]
    and the time ordering operator $\TT\colon \left(\Fcal_{\reg}\llbracket\hbar\rrbracket,\cdot\right) \to \left(\Fcal_{\reg}[[\hbar]],\cdot_\TT\right)$ is given by
    \[
    \TT\doteq \alpha_{\Delta^{\rm F}},
    \]
    where the Feynmann propagator\footnote{For a choice of Hadamard state $H$, kept implicit in the notation.} is
    \[
    \Delta^F = \frac{i}{2}\left(\Delta^R + \Delta^A\right) + H
    \]
    The time-ordering operator induces the quantisation map $\mathfrak{Q}_{\TT}: \Fcal_{\reg}\rightarrow \fA_{\reg}$ defined by\footnote{Note that the $\TT$ here below is called $\TT_H$ in \cite[4.4.2]{HR}, to reflect the dependence on $H$. We will leave this dependence implicit.}
    \[
    \mathfrak{Q}_{\TT} \doteq \mathfrak{Q}_H \circ \TT,
    \]
    so that
    \[
    (\mathfrak{Q}_\TT(F))_H = \TT F.
    \]
\end{definition}

\begin{remark}[$\TT$ presentation]
Since $\mathfrak{Q}_\TT$ is a quantisation map, it gives rise to another presentation, i.e. a realisation of the algebraic structures on $\mathfrak{A}$ onto functionals (see Remark \ref{rmk:presentations}). We call this the $\TT$-presentation.
One important advantage of the $\mathfrak{Q}_{\TT}$ presentation is that the induced time ordered product in this presentation is, by definition, just $\cdot$, i.e the pointwise product of functionals.  Note that here we differ in our conventions from \cite{HR}, where the $\TT$ presentation is associated to $\alpha_{i\Delta^{\rm D}}$, rather than $\alpha_{\Delta^{\rm F}}$ as it is here.
\end{remark}
We can use $\mathfrak{Q}_\TT$ to introduce the time-ordered product on $\fA_{\reg}$.
\[
(\mathfrak{Q}_{\TT}F)\circT (\mathfrak{Q}_{\TT} G)\doteq  \mathfrak{Q}_{\TT}( F\cdot G)\,.
\]
It was shown in \cite{HR} that this product has the desired property
\[
A \circT B =
\begin{cases}
A \bullet B & \text{if } \supp(B) \prec \supp(A), \\[4pt]
B \bullet A & \text{if } \supp(A) \prec \supp(B)\,,
\end{cases}
\]
meaning that it is indeed the time-ordering of $\bullet$.

 In contrast, the induced time-ordered product in the $\mathfrak{Q}_{H}$ presentation is more complicated, namely
\begin{align*}
\mathfrak{Q}_{H}^{-1}((\mathfrak{Q}_{H}F)\circT (\mathfrak{Q}_{H} G))&=\TT\circ \mathfrak{Q}_{\TT}^{-1}((\mathfrak{Q}_{\TT}\circ \TT^{-1}F)\circT (\mathfrak{Q}_{\TT}\circ \TT^{-1} G))\\
    &=\TT(\TT^{-1}F\cdot \TT^{-1} G)\\
    &=F\T G\,.
\end{align*}
The induced star product in the $\mathfrak{Q}_{\TT}$ presentation has been computed in \cite{HR} and is given by
\[
F\star_{\TT} G=  m \circ e^{-i\hbar D_A} (F\otimes G)\,,
\]
    where
\[
    D_A\doteq \left\langle \Delta^A, \frac{\delta}{\delta \varphi}\otimes \frac{\delta}{\delta \varphi}\right\rangle\,.
\]

\subsubsection{Physical intuition behind the time-ordered product}\label{sssec:physicsl intuition}
The time-ordering operator $\TT$ is closely connected to the path integral formulation of quantum field theory. Here, following \cite{Keller} we present the heuristic reasoning, which can be made precise in special cases in the Euclidean setting, where the path integral measure can be rigorously constructed.
We start with the path integral convolution with the ``Gaussian measure'' defined by the free generalised Lagrangian $L_0(\psi)=\frac{1}{2}\left<\psi,D\psi\right>$.
\[
\int F(\phi) e^{\frac{i}{\hbar}L_0(\phi-\ph)} \mathcal{D}\phi \,.
\]
Even though ``$\mathcal{D}\phi $'' doesn't exist on its own right as a translation-invariant measure, the whole expression $d\mu(\psi)\equiv e^{\frac{i}{\hbar} L_0(\psi)} \mathcal{D}\psi$ makes sense in the Euclidean setting (with the factor $i$ replaced by $-1$), where it becomes a Gaussian measure. Taking the Fourier transform, we obtain
\[
\int d\mu(\phi-\varphi)\,F(\phi)
=  \int \mathcal{D}\phi \int \mathcal{D}J \,
   e^{i\langle \phi-\ph,J\rangle}
   e^{\frac{\hbar}{2}\langle J,(iD)^{-1}J\rangle} F(\phi)\,, 
\]
where the crucial subtlety is the choice of the inverse of $iD$. Taking it to be $D_F$, we get (up to factors of $2\pi$)
\[
\begin{aligned}
\int d\mu(\phi-\ph)\,F(\phi)
&=  \int \mathcal{D}J\, e^{-i\langle \varphi,J\rangle}
   \int \mathcal{D}\phi \, e^{i\langle \phi,J\rangle}
   e^{\frac{\hbar}{2}\langle J,D^FJ\rangle} F(\phi) \\
&= \int \mathcal{D}J\, e^{-i\langle \varphi,J\rangle}
   \int \mathcal{D}\phi \, e^{i\langle \phi,J\rangle}
   e^{\frac{\hbar}{2}\left\langle \frac{\delta}{\delta\phi},
      D^F\,\frac{\delta}{\delta\phi}\right\rangle} F(\phi) \\
&= \bigl[\mathfrak{F}^{-1}\mathfrak{F}
   \bigl(e^{\frac{\hbar}{2}     \mathfrak{D}_F}F\bigr)\bigr](\ph)
 = (\Tcal F)(\ph)\,,
\end{aligned}
\]
where $\mathfrak{F}$ denotes the Fourier transform and $\mathfrak{F}^{-1}$ the inverse Fourier transform. Using this heuristic idea, we can think of $(\Tcal F)(\ph)$ as the algebraic version of the convolution with the path integral Gaussian measure, without the necessity to directly work with the measure itself. Evaluation at $\ph$ is then understood as a coherent state.

\subsubsection{Renormalised time-ordering}
Definition \ref{def:regTimeprod} is valid for \emph{regular} functionals. This is obviously unsatisfactory because the relevant interactions in QFT are local and not just regular,\footnote{The only local and regular functionals are the constant and linear ones.} so one looks for an extension to multilocal functionals. The issue one encounters is that the distributions defining star and time-ordered product have nontrivial wavefront sets. While the star product extends without issue to a larger class of functionals\footnote{The domain of the star product was usually taken to be the space of microcausal functionals (see \cite{Book} and references therein), but it was later realised in \cite{HRV} that the correct space to consider is rather that of equicausal functionals.} called \emph{equicausal functionals $\Fcal_{\rm ec}$}, as defined in \cite{HRV}, the time ordering operator and the time ordered product require additional care when the inputs are no longer regular.

Owing to the Feynman propagator being highly singular along the diagonal, one is faced with the problem of renormalisation, which is embodied by the problem of extension of distributions defined away from the diagonal, to distributions defined everywhere. We deal with it using the Epstein--Glaser formalism, and the output of this procedure is a renormalised time ordering operator \cite{EpsteinGlaser,FredenhagenRejzner2015}.
\begin{definition}[Renormalised time ordered product]
    Let 
    \be\label{e:ImTT}
        \Fcal_{\TT}(\Ecal)\equiv \Fcal_{\TT} \doteq \im(\TT). 
    \ee
    for 
    \[
    \TT\colon \Fcal_{\ml}[[\hbar]] \to \Fcal_{\rm ec}[[\hbar]]
    \]
    a renormalised time ordering operator (\cite{EpsteinGlaser,FredenhagenRejzner2015}), i.e. such that
    \begin{enumerate}
        \item $\TT=\mathrm{id} \mod \hbar$,
        \item $\supp(\TT(F))$ contained in the causal completion of the support of $F$,
        \item $\TT\vert_{\Fcal_\reg}$ is a differential operator,
        \item $\TT F = F$ for $F$ linear.
    \end{enumerate}
    
    The renormalised time ordered product is the commutative, associative product defined as 
    \[
    \cdot_\TT \colon \Fcal_{\TT}\times \Fcal_{\TT} \to \Fcal_{\TT}, \qquad F\cdot_\TT G \doteq \TT(\TT^{-1}F \cdot \TT^{-1}G).
    \]
    Similarly, one can define $\fA_{\TT}\subset \fA$ and introduce the normalized time-ordered product $\circT$ in the abstract algebra on this domain, by requiring that there be a quantisation map $\mathfrak{Q}_\TT\colon \Fcal_* \to \mathfrak{A}$ such that \cite{HR}
    \[
    \mathfrak{Q}_\TT F \circT \mathfrak{Q}_\TT G = \mathfrak{Q}_\TT(F\cdot G)
    \]
    for every $F,G\in \Fcal_*$ where $F_*$ denotes either regular or equicausal functionals.
\end{definition}

\subsubsection{Free quantum BV operator}
We introduce the free quantum BV operator on $\fA$ using the following definition
\[
(\mathcal{Q}_0 A)_H\doteq Q_0 A_H\,,
\]
where $Q_0$ is the free classical BV operator. This definition makes sense, since $H\in {\rm Had}$ is a solution for $D$ and hence $Q_0$ commutes with $\alpha_{H'-H}$. To show that $\mathcal{Q}_0$ commutes with the quantisation map $\mathfrak{Q}_H$ we compute
\[
(\mathcal{Q}_0(\mathfrak{Q}_H(F)))_{H'}=Q_0((\mathfrak{Q}_H(F))_{H'})=Q_0(\alpha_{H'-H}F)=\alpha_{H'-H}(Q_0F)=(\mathfrak{Q}_H(Q_0F))_{H'}
\]
Hence 
\[
\mathcal{Q}_0\circ \mathfrak{Q}_H=\mathfrak{Q}_H\circ Q_0\,,
\]
so in the $\mathfrak{Q}_{H}$ presentation the induced quantum BV operator on $\Fcal_{\rm ec}[[\hbar]]$ is given by
\[
Q_{0,H}\doteq \mathfrak{Q}^{-1}_H\circ \mathcal{Q}_0\circ \mathfrak{Q}_H=Q_0\,.
\]
The same is not true for the quantisation map $\mathfrak{Q}_{\TT}$, since $\Delta^{\rm F}$ is a Green function rather than solution for $D$ and for $F\in \Fcal_{\reg}$ we have
\begin{align*}
\mathcal{Q}_0\circ \mathfrak{Q}_{\TT} F & = \mathcal{Q}_0\circ \mathfrak{Q}_H \circ \TT\, F= \mathfrak{Q}_H \circ Q_0\circ \TT F\\
&= \mathfrak{Q}_H \circ \TT\circ (Q_0 - i\hbar\Lap)F= \mathfrak{Q}_{\TT} \circ (Q_0 - i\hbar\Lap)F\,,
\end{align*}
where the penultimate step uses:
\begin{lemma}\label{lem:Q0hatexpression}
    On functionals in $\Fcal_{\rm reg}[[\lambda]]((\hbar))$ 
    we have
    \[
    Q_0\circ \TT F = \TT\circ (Q_0 - i\hbar\Lap)F.
    \]
    The \emph{bare} (or unrenormalised) BV Laplacian on regular functionals is
    \[
    \Lap\doteq \int \frac{\delta^2}{\delta \ph^\alpha \delta\ph^{\ddagger}_\alpha }\,.
    \]
\end{lemma}
\begin{proof}
    A proof is given, for example, in \cite[Section 3.3]{rej22}.
\end{proof}
We conclude that in the $\mathfrak{Q}_{\TT}$ presentation the induced quantum BV operator on $\Fcal_{\reg}[[\hbar]]$ is given by
\[
Q_{0,\TT}\doteq \mathfrak{Q}_{\TT}^{-1}\circ \mathcal{Q}_0\circ  \mathfrak{Q}_{\TT}= \TT^{-1}\circ Q_0 \circ \TT= Q_0 - i\hbar\Lap\,.
\]

We will be interested in twisting also the cutoff version of the BV operator, by looking at $Q_0[f]$, the Hamiltonian vector field of $L_0[f]$. Hence we implicitly define also 
\[
{Q}_{0,\TT}[f] \doteq \mathfrak{Q}_{\TT}^{-1} \circ \mathcal{Q}_0[f] \circ \mathfrak{Q}_{\TT}.
\]
This will be particularly relevant in Section \ref{sec:BVBFVPAQFT}, where we will explore boundary contributions to the classical master equation.

The above expression for the time-ordered BV operator $Q_{0,\TT}$ (and the ensuing BV Laplacian) holds only for regular functionals. To extend this reasoning to a larger domain let us introduce now the concept of Anomalous Master Ward Identity (AMWI), due to Brennecke and D\"utsch. Let $\Fcal_{\rm ml}((\hbar))[[\lambda]]$ be the space of formal power series in $\lambda$ with coefficients in Laurent series in $\hbar$. $Q_0$ naturally extends to this space.

\begin{theorem}[\cite{BrenneckeDutsch}]\label{thm:AMWI}
    Let $F\in \Fcal_{\ml}[[\hbar]]$ and $\TT$ be a renormalised time ordered product. Then
    \[
    Q_0 \TT e^{\frac{i\lambda }{\hbar}F} = \frac{i}{\hbar}\TT\left(\left( Q_0 \lambda F + \frac12\{\lambda F,\lambda F\} - i\hbar A(\lambda F)\right) \cdot e^{\frac{i\lambda }{\hbar}F}\right)
    \]
    for some local functional $A(\lambda F)\in \Fcal_{\loc}\llbracket \hbar,\lambda \rrbracket$ called the \emph{anomaly term}.
\end{theorem}

\begin{remark}
    Note that Theorem~\ref{thm:AMWI} can be formulated more abstractly as
    \[
    \mathcal{Q}_0\circ  \mathfrak{Q}_{\TT} e^{\frac{i\lambda }{\hbar}F} = \frac{i}{\hbar} \mathfrak{Q}_{\TT}\left(\left( Q_0 \lambda F + \frac12\{\lambda F,\lambda F\} - i\hbar A(\lambda F)\right) \cdot e^{\frac{i\lambda }{\hbar}F}\right)\,,
    \]
    which allows us to define $Q_{0,\TT}\doteq \mathfrak{Q}_{\TT}^{-1}\circ\mathcal{Q}_0\circ  \mathfrak{Q}_{\TT}$ on exponentials of local functionals, i.e. for $F\in \Fcal_{\loc}$, we have
    \[
    Q_{0,\TT} e^{\frac{i\lambda }{\hbar}F} = \frac{i}{\hbar} \left( Q_0 \lambda F + \frac12\{\lambda F,\lambda F\} - i\hbar A(\lambda F)\right) \cdot e^{\frac{i\lambda }{\hbar}F}
    \]
\end{remark}
Using the remark above, we can now extend the definition of $Q_{0,\TT}$ to multilocal functionals.
Observe first that any homogeneous multilocal functional $F= F_1\cdot \dots \cdot F_n$ can be written as 
\[
F = \left(\frac{\hbar}{i}\right)^n\frac{d^n}{d\lambda_1\dots d\lambda_n} e^{\frac{i}{\hbar}F_\lambda}\Big\vert_{\lambda=0}\,,
\]
where $F_\lambda = \sum_i \lambda_i F_i$ and $\lambda=(\lambda_1,\dots,\lambda_n)\in \mathbb{R}^n$. Then

{
\begin{proposition}\label{prop:freeQBVop}
For any $F\in\Fcal_\loc$, denote by
\be\label{e:QME}
\mathsf{QME}(F)\doteq Q_0 F + \frac12\{F,F\}^\std - i\hbar A(F).
\ee
The free quantum BV operator in the $\TT$ presentation on multilocal functionals is 
\begin{align*}
Q_{0,\TT}F &= {\left(\frac{\hbar}{i}\right)^{n-1}} \frac{d^n}{d\lambda_1\dots d\lambda_n} \left[\left( \mathsf{QME}(F_\lambda)\right) \cdot e^{\frac{i }{\hbar}F_\lambda} \right]\Big\vert_{\lambda_1=\dots=\lambda_n=0}\\
&={\sum_i(Q_0F_i) F_1 \dots F_{\hat{i}} \dots F_n + \frac{\hbar}{i}\sum_{i,j}\{F_i,F_j\}^\std F_1\dots F_{\hat{i}}\dots F_{\hat{j}} \dots F_n}\\
&{+\left(\frac{\hbar}{i}\right)^n\langle A^{(n)}(0),F_1\otimes \dots\otimes F_n\rangle}
\end{align*}
and the renormalised free BV Laplacian reads
\begin{align*}
\Lap_0^{\ren}F&\doteq \frac{i}{\hbar}(Q_{0,\TT} - Q_0)F \\
&= { \frac{\hbar}{i}\sum_{i,j}\{F_i,F_j\}^\std F_1\dots F_{\hat{i}}\dots F_{\hat{j}} \dots F_n}{+\left(\frac{\hbar}{i}\right)^n\langle A^{(n)}(0),F_1\otimes \dots\otimes F_n\rangle}.
\end{align*}
Moreover, for $F\in\Fcal_\loc$ we have
    \[
    -i\hbar\Lap_0^\ren(F\eihbV) = \left(\{V,F\}^\std - i\hbar\langle A'(V),F\rangle\right) \eihbV  +\frac{i}{\hbar}\left(\mathsf{QME}(V) - Q_0(V)\right)F\eihbV.
    \]
\end{proposition}
\begin{proof}
    Directly, from Theorem \ref{thm:AMWI} 
    \begin{align*}
    Q_0 \TT F &=  \left(\frac{\hbar}{i}\right)^n\frac{d^n}{d\lambda_1\dots d\lambda_n} Q_0 \TT e^{\frac{i}{\hbar}F_\lambda}\Big\vert_{\lambda_1=\dots=\lambda_n=0} \\
    &= \left(\frac{\hbar}{i}\right)^{n-1}  \TT \frac{d^n}{d\lambda_1\dots d\lambda_n}{\left[\left( \underbrace{Q_0 F_\lambda + \frac12\{F_\lambda,F_\lambda\} - i\hbar A(F_\lambda)}_{\mathsf{QME}(F_\lambda)}\right) \cdot e^{\frac{i}{\hbar}F_\lambda} \right]}\Big\vert_{\lambda_1=\dots=\lambda_n=0}\\
    &=\left(\frac{\hbar}{i}\right)^{n-1}\TT\Big(\left(\frac{i}{\hbar}\right)^{n-1}\sum_i(Q_0F_i)F_1\dots F_{\hat{i}}\dots F_n\\
    & + \left(\frac{i}{\hbar}\right)^{n-2}\sum_{i,j}\{F_i,F_j\}^\std F_1\dots F_{\hat{i}}\dots F_{\hat{j}} \dots F_n + \langle A^{(n)}(0),F_1\otimes\dots \otimes F_n\rangle\Big)
    \end{align*}
so that:
    \[
    \TT^{-1}Q_0\TT(F) = Q_0 F + \frac{\hbar}{i}\sum_{i,j}\{F_i,F_j\}^\std F_1\dots F_{\hat{i}}\dots F_{\hat{j}} \dots F_n + \left(\frac{i}{\hbar}\right)\langle A^{(n)}(0),F_1\otimes\dots \otimes F_n\rangle.
    \]

    To obtain the last statement we apply the same technique, writing, for $F$ local
    \begin{align*}
    \Lap_0^\ren(F\eihbV) &= \left(\frac{\hbar}{i}\right)\frac{d}{d\lambda}\Big\vert_{\lambda=0}\Lap_0^\ren(e^{\frac{i}{\hbar}(V + \lambda F)}) \\
    &= \frac{i}{\hbar}\frac{d}{d\lambda}\Big\vert_{\lambda=0}\left(\mathsf{QME}(V+\lambda F)e^{\frac{i}{\hbar}(V + \lambda F)}\right) - \frac{i}{\hbar}Q_0(F\eihbV)\\
    &=\frac{i}{\hbar}\frac{d}{d\lambda}\Big\vert_{\lambda=0}\left(\frac12\left(\{V+\lambda F, V+\lambda F\}^\std - i\hbar A(V+\lambda F)\right)e^{\frac{i}{\hbar}(V + \lambda F)} \right)\\
    &=\frac{i}{\hbar}\left(\{V,F\} - i\hbar\langle A'(V),F\rangle \right)\eihbV + \left(\frac{i}{\hbar}\right)^2\left(\frac12\{V,V\}^\std - i\hbar A(V)\right)F\eihbV
    \end{align*}
    and we conclude by identifying the last term in brackets with $\mathsf{QME}(V) - Q_0(V)$. 
\end{proof}}

The term $\mathsf{QME}(F)$ will later lead to a ``Quantum Master Equation'', but for now it is just a symbol. Note that, by linearity, $\mathsf{QME}(0) = 0 $.

We can again compare how different structures behave in the $\mathfrak{Q}_H$ and $\mathfrak{Q}_{\TT}$ presentations. The induced BV operator is simpler in the $\mathfrak{Q}_H$ presentation, since it coincides with the classical BV operator $Q_0$, while in the $\mathfrak{Q}_{\TT}$ presentations it is deformed.
\begin{table}[!h]
\label{table:presentations}
\caption{Algebraic translations between presentations}
\begin{tabular}{c|c|c}
 $\fA$    & $\mathfrak{Q}_H$ presentation  & $\mathfrak{Q}_\TT$ presentation \\
 \hline
 $\bullet$ & $\star_H$ & $\star_\Tcal$\\
 \hline
 $\circT$ &$\T$  & $\cdot $\\
 \hline
 $\mathcal{Q}_0$ & $Q_0$ & $Q_0-i\hbar \Lap_0^{\ren}$ 
\end{tabular}

\end{table}

Note that, compared to the bare BV Laplacian $\Lap$, $\Lap_0^{\ren}$ is not necessarily nilpotent, and neither $\Lap$ nor $\Lap_0^\ren$ are derivations of the associative commutative product, although they both annihilate the unit.  Moreover, while $\Lap$ is generally interpreted as the generator of the Gerstenhaber (i.e.\ BV) algebra $(\Fcal_\reg,\cdot,\{\cdot,\cdot\}^\std)$, the renormalised free BV Laplacian can be related to a ``higher'' structure, as discussed in \cite{Froeb,BDFR_infty}. Indeed, one defines higher brackets by
    \begin{subequations}\label{e:higherbrackets}\begin{align}
    &[-]_0^V=\mathsf{QME}(V)\label{e:[]0}\\
    &[G]_1^V=\langle \mathsf{QME}'(V),G\rangle\label{e:[]1}\\
    &[G_1,G_2]_2^V=\langle \mathsf{QME}''(V),G_1\otimes G_2\rangle=\{G_1,G_2\}+\langle A''(V),G_1\otimes G_2\rangle\label{e:[]2}\\
    &[G_1,\ldots,G_n]_n^V=\langle \mathsf{QME}^{(n)}(V),G_1\otimes\cdots\otimes G_n\rangle =\langle A^{(n)}(V),G_1\otimes\cdots\otimes G_n\rangle,
    \end{align}\end{subequations}
where the last equation holds for all $n\geq 3$, and one can show 
\begin{theorem}[\cite{BDFR_infty} Proposition 6.1]\label{thm:AMWI-Linfty}
    The brackets $[\cdot,\dots,\cdot]_k^V$ define an $L_\infty$ algebra when $\mathsf{QME}(V)=0$.
\end{theorem}

This structure will be relevant also when studying the interacting theory, since in that case the ensuing $L_\infty$ algebra will be determined by the interaction term $V$. Furthermore, in the bounded case we will see that the $L_\infty$ structure will be curved.

\begin{remark}
    In this context it would be interesting to flesh-out the link with a BV${}_\infty$ structure, see \cite{BashkirovVoronov} and therein, after \cite{BeringLada}. We plan to address this in a future work.
\end{remark}

\subsection{Interacting BV theory on globally hyperbolic manifolds}\label{sec:interactingBV}
In this section we will make full use of the free-interacting splitting introduced above. Our starting point is (the quantisation of) a free theory in the BV setting, to which we apply certain deformation techniques to get (a quantisation of) the full BV theory. 

\subsubsection{The formal S-matrix and M{\o}ller maps}
Let us begin with a few definitions.
\begin{definition}[$S$ matrix]\label{def:Smatrix}
The \emph{formal S-matrix} is defined as
    \[
    \mathcal{S}\colon  \lambda \fA_{\TT}[[\lambda]] \to \fA_{\TT}[[\lambda,\lambda/\hbar]],\qquad  \mathcal{S}(A)=e_{\circT}^{iA/\hbar}= \mathfrak{Q}_{\TT}(e^{\frac{i}{\hbar} \mathfrak{Q}_{\TT}^{-1} A }) \,.
    \]
\end{definition}
We are interested in situations where the argument of $  \mathcal{S}$ arises from a polynomial Lagrangian in a free-interacting splitting $L=L_0 + V$ where $V[f] = \lambda\int_M \boldsymbol{V} f$ and we can set $A=\mathfrak{Q}_{\TT}(V)$ in the defintion above.  

\begin{remark}[$\TT$ and Hybrid-$H$ presentations]
    The idea, going back to \cite{BDF} (see also \cite{BucholzFredenhagenC*} for the formulation in terms of C*-algebras), is that formal S-matrices are unitaries in $\fA[[\lambda,\lambda/\hbar]]$ labeled by local functionals.

    As we have seen when discussing presentation, in order to identify functionals with elements of $\fA_{\TT}[[\lambda,\lambda/\hbar]]$ one can use any quantizaton map. We will henceforth use the quantisation map $\mathfrak{Q}_{\TT}$.

    The S-matrix on $\mathfrak{A}$ can then be pulled back to a map valued in $\Fcal((\hbar))[[\lambda]]$ using the \emph{same} quantisation map $\mathfrak{Q}_{\TT}$, leading to the $\TT$ presentation, or a different one. 
    
    This `mismatch' gives rise to a \emph{hybrid presentation} of the S matrix which will be particularly useful, and is motivated by the relation to path integral that will be discussed later.

    As with the $S$-matrix, this mismatched presentation can be applied to M\o ller maps, to induce corresponding maps on functionals. We again choose $\mathfrak{Q}_\TT$ to map functionals to elements of $\fA_{\TT}[[\lambda]]$ and we have two choices for pulling things back.
\end{remark}

\begin{definition}[$S$ matrix on functionals]\label{def:Smatrixfunctionals}
The formal S-matrix as a map from functionals to $\fA_{\TT}[[\lambda,\lambda/\hbar]]$ is defined by
\[ \Scal_{\TT} \colon  \lambda \Fcal_{\ml}[[\hbar]]\to  
\fA_{\TT}[[\lambda,\lambda/\hbar]],\qquad \Scal_{\TT}(V)= \Scal\circ \mathfrak{Q}_\TT (V)
\]
In the $\mathfrak{Q}_{\TT}$ presentation we then define (observe the font change)
     \[
     S_{\TT}\colon  \lambda \Fcal_{\ml}[[\hbar]]\to  \Fcal((\hbar))[[\lambda]],\qquad  S_{\TT}(V)= \mathfrak{Q}_{\TT}^{-1}\circ \Scal_{\TT}(V) =e^{\frac{i}{\hbar}  V}  \,.
     \]
 Similarly, we introduce the \emph{hybrid H-presentation} and applying $\mathfrak{Q}_{H}$ we define
     \[
     S_H\colon  \lambda \Fcal_{\ml}[[\hbar]]\to \Fcal((\hbar))[[\lambda]],\qquad  S_H(V)= \mathfrak{Q}_{H}^{-1}\circ \Scal_{\TT}(V)= \TT e^{\frac{i}{\hbar}  V}  \,.
     \]
\end{definition}

\begin{definition}\label{def:moller}
    With definitions as above, we define the abstract retarded M\o ller operator as 
    \[
    \Rcal_A\colon \fA_{\TT}[[\lambda]]  \to \fA[[\lambda]], \qquad \Rcal_A(B)\doteq \mathcal{S}(A)^{-1}\bullet ( \mathcal{S}(A)\circT B))\,.
    \]
    Similarly, the abstract advanced M\o ller operator is
    \[
    \Acal_A\colon \fA_{\TT}[[\lambda]]  \to \fA[[\lambda]], \qquad \Rcal_A(B)\doteq ( \mathcal{S}(A)\circT B))\bullet \mathcal{S}(A)^{-1}\,.
    \]
\end{definition}

\begin{definition}\label{RV on functionals H}
    The retarded M\o ller operator in the \emph{hybrid H-presentation} is
    \[
    R_{V,H} \colon \Fcal_{\ml}[[\hbar]] \to \Fcal_{\rm ec}[[\hbar,\lambda]], \quad R_{V,H}(F)\doteq \mathfrak{Q}_H^{-1} \Rcal_{\mathfrak{Q}_\TT V}(\mathfrak{Q}_\TT F)
    =    (\TT e^{\frac{i}{\hbar}V} )^{-1}\star_H (\TT e^{\frac{i}{\hbar}  V} F )\,.
    \]
    Similarly, the advanced M\o ller operator in the \emph{hybrid H-presentation} is
    \[
    A_{V,H} \colon \Fcal_{\ml}[[\hbar]] \to \Fcal_{\rm ec}[[\hbar,\lambda]], \quad A_{V,H}(F)\doteq \mathfrak{Q}_H^{-1} \Acal_{\mathfrak{Q}_\TT V}(\mathfrak{Q}_\TT F)=    
    (\TT e^{\frac{i}{\hbar}  V} F )\star_H (\TT e^{\frac{i}{\hbar}V} )^{-1}\,.
    \]
\end{definition}

For the $\mathfrak{Q}_{\TT}$ presentation we need to restrict ourselves to regular functionals, due to issues with normalizing $\star_\TT$ (see \cite{HR} for a detailed discussion).
\begin{definition}\label{RV on functionals T}
    The retarded M\o ller operator in the $\TT$-presentation is
    \[
    R_{V,\TT} \colon \Fcal_{\reg}[[\hbar]] \to \Fcal_{\reg}[[\hbar]] \, \qquad R_{V,\TT}(F)\doteq \mathfrak{Q}_\TT^{-1} \Rcal_{\mathfrak{Q}_\TT V}(\mathfrak{Q}_\TT F)=    
    (e^{\frac{i}{\hbar}V} )^{-1}\star_\TT (e^{\frac{i}{\hbar}  V} F )\,.
    \]
    Similarly, the advanced M\o ller operator in the $\TT$-presentation is
    \[
    A_{V,\TT} \colon \Fcal_{\reg}[[\hbar]] \to \Fcal_{\reg}[[\hbar]] \, \qquad A_{V,\TT}(F)\doteq \mathfrak{Q}_\TT^{-1} \Acal_{\mathfrak{Q}_\TT V}(\mathfrak{Q}_\TT F)= (e^{\frac{i}{\hbar}  V} F )\star_\TT (e^{\frac{i}{\hbar}V} )^{-1}\,.
    \]
\end{definition}
On regular functionals, the two presentations are related by
\[
R_{V,H}=\TT\circ R_{V,\TT} \,,\qquad A_{V,H}=\TT\circ A_{V,\TT} 
\]


\begin{remark}[H vs. Hybrid-H presentation]
    Note that, in principle, we could also consider mapping functionals to elements of $\mathfrak{A}$ using $\mathfrak{Q}_H$ and then return using either $\mathfrak{Q}_H$ or $\mathfrak{Q}_\TT$. This would give us two more definitions for M\o ller operators, as $\mathfrak{Q}_\TT^{-1}\circ \mathcal{R}_{\mathfrak{Q}_H(V)}\circ\mathfrak{Q}_H$, which would also be ``hybrid'', and $\mathfrak{Q}_H^{-1}\circ \mathcal{R}_{\mathfrak{Q}_H(V)}\circ\mathfrak{Q}_H$, which would instead be fully within the $H$-presentation. We shall not pursue this here. It is important to note that, from now on, the subscript $H$ will denote the hybrid-H presentation.
\end{remark}

Following \cite{HR} we introduce the abstract interacting product.
\begin{definition}
The (regular) interacting star product in the abstract algebra 
\[
\bullet_A\colon \fA_{\reg}\llbracket\lambda\rrbracket\times \fA_{\reg}\llbracket\lambda\rrbracket \to \fA_{\reg}\llbracket\lambda\rrbracket
\]
is defined as
\[
B\bullet_A C\doteq \Rcal_A^{-1}(\Rcal_A(B)\bullet \Rcal_A(C))\,.
\]
\end{definition}
We can now use quantisation maps to pull this back to the level of functionals.
\begin{definition}\label{def:intstarprod}
The (regular) interacting star product in the $\mathfrak{Q}_*$ presentation (with $*=H,\TT$) is a map
\[
\star_{V,*}\colon \mathcal{F}_{\reg}\llbracket\hbar,\lambda\rrbracket\times \mathcal{F}_{\reg}\llbracket\hbar,\lambda\rrbracket \to \mathcal{F}_{\reg}\llbracket\hbar,\lambda\rrbracket
\]
defined as
\[
F\star_{V,*} G\doteq R_{V,*}^{-1}(R_{V,*}(F)\star_* R_{V,*}(G))\,.
\]
\end{definition}
{\begin{lemma}\label{lem:Vproduct}
The interacting star product is independent of the choice of presentation and is given by
\[
F\star_V G= e^{-\frac{i}{\hbar} V} \left( (\eihbV F)\star_{\TT} (\eihbV)^{-1} \star_{\TT} (\eihbV G) \right)
\]
Moreover
\[
F\star_V G=  A_{V,*}^{-1}(A_{V,*}(F)\star_* A_{V,*}(G))\,.
\]  
\end{lemma}
\begin{proof}
    Recall that
    \[
    R_{V,H}= \TT \circ  R_{V,\TT}\,,
    \]
    Hence
\begin{align*}
F\star_{V,H} G&=R_{V,\TT}^{-1}\circ \TT^{-1}(\TT\circ R_{V,\TT}(F)\star_H \TT\circ R_{V,\TT}(G))\\
&=R_{V,\TT}(R_{V,\TT}(F)\star_\TT R_{V,\TT}(G))\\
&= F\star_{V,\TT} G\,,
\end{align*}
where in the second step we used the following rule for commuting $\TT$ past the pointwise multiplication operation $m$:
\[
\TT\circ m= e^{\frac{\hbar}{2}\left\langle \Delta^F, \frac{\delta^2}{\delta \varphi^2}\right\rangle}\circ m =m\circ e^{\hbar\left\langle \Delta^F, \frac{\delta}{\delta \varphi}\otimes \frac{\delta}{\delta \varphi}\right\rangle} \circ (\TT\otimes \mathrm{id})\circ (\mathrm{id}\otimes \TT)\,.
\]
To obtain the explicit formula for $\star_V$, note that
\[
R_{V,\TT}^{-1}F= e^{-\frac{i}{\hbar}V}\cdot (\eihbV\star_{\TT} F)
\]
Hence
\begin{align*}
F\star_V G&=e^{-\frac{i}{\hbar}V}\cdot (\eihbV\star_{\TT} ((\eihbV)^{-1}\star_{\TT} (\eihbV F) \star_\TT (\eihbV)^{-1}\star_{\TT} (\eihbV G) ))\\
&=e^{-\frac{i}{\hbar}V}\cdot ((\eihbV F) \star_\TT (\eihbV)^{-1}\star_{\TT} (\eihbV G) )
\end{align*}
The final statement follows from the fact that 
\[
A_{V,\TT}^{-1}F= e^{-\frac{i}{\hbar}V}\cdot (F\star_{\TT} \eihbV)
\]
and we can perform the same manipulation as for the retarded version, but one factor of $\eihbV$ cancels on the right instead of the left.
\end{proof}}

\begin{remark}\label{rem:starVwelldef}
The interacting star product does not have a (known) direct generalisation to the case where $F$ and $G$ are local but not regular. Instead, one works with:
\[
R_{V,H}(F\star_{V} G)=R_{V,H}(F)\star_H R_{V,H}(G)\,,
\]
since the right-hand side can be renormalized, and thus $\star_{V}$ can be extended---albeit indirectly---to multilocal functionals. Analogous construction does not work for the $\TT$-presentation, since the product  $\star_{\TT}$ relevant in this case is constructed using the advanced propagator (see \cite{HR}) rather than the 2-point function $\Delta^+$ and wavefront set arguments used in the construction of $\star_H$ do not apply.
\end{remark}

\subsubsection{Interacting quantum BV operators}
From the quantisation of a free field theory and the M\o ller operators and maps introduced above, we can construct an ineracting quantum theory, given the interaction satisfies certain conditions.
\begin{definition}\label{df: BVint}
    We say that $\mathcal{V}\in \fA$ satisfies the Quantum Master Equation (henceforth QME) iff
    \be
       \mathsf{(QME)} \qquad  \mathcal{Q}_0 \Scal(\mathcal{V})=0.\label{eq:QME}
    \ee  
    We say that a generalised Lagrangian in a free-interacting splitting $L=L_0+V$ satisfies the QME iff $\mathcal{V}=\mathfrak{Q}_*(V)$ does.
    
\end{definition}

\begin{remark}[QME in $\TT$ and $H$ presentations]\label{rem:QME presentations}
    
As before, it induces the following operators on functionals. Consider a free-interacting splitting $L=L_0+V$ for a generalised Lagrangian $L$. 
The QME in the hybrid H-or $\TT$-presentations takes the form
\[
 \mathfrak{Q}_{H/\TT}^{-1}\circ \mathcal{Q}_0\circ  \Scal\circ \mathfrak{Q}_\TT (V) {\equiv \mathfrak{Q}_{H/\TT}^{-1}\circ \mathcal{Q}_0\circ \mathfrak{Q}_\TT \eihbV}=0
\]

In the hybrid $H$-presentation this reduces to
\be\label{e:equivQME-H}
\mathsf{(QME)}_H \qquad Q_{0,H}\TT e^{\frac{i}{\hbar}V}\equiv Q_0 \TT e^{\frac{i}{\hbar}V}=0\,.
\ee\label{e:equivQME-T}

In the $\TT$-presentation we have\footnote{We must use $\Lap_0^\ren$ if $V$ is local and not regular.}
\be
\mathsf{(QME)}_\TT \qquad Q_{0,\TT}e^{\frac{i}{\hbar}V}\equiv(Q_0-i\hbar \Lap_0^{(\ren)}) e^{\frac{i}{\hbar}V}=0\,.
\ee
These are equivalent statements of the QME.

\end{remark}

\begin{definition}\label{df: BVint on functionals}
    Consider a free-interacting splitting $L=L_0+V$ for a generalised Lagrangian $L$. 
    
    \noindent The abstract retarded/advanced interacting quantum BV operators are
    \[
    \mathcal{Q}^R_\mathcal{V} \doteq \Rcal_\mathcal{V}^{-1}\circ \mathcal{Q}_0 \circ \Rcal_\mathcal{V}\,\qquad \mathcal{Q}^A_\mathcal{V} \doteq \Acal_\mathcal{V}^{-1}\circ \mathcal{Q}_0 \circ \Acal_\mathcal{V}\,.
    \]
    The retarded/advanced interacting quantum BV operators on functionals are
    \[
    Q^R_{V} \doteq \mathfrak{Q}_\TT^{-1} \circ  \mathcal{Q}^R_{\mathfrak{Q}_\TT V} \circ \mathfrak{Q}_\TT \,\qquad Q^A_{V} \doteq  \mathfrak{Q}_\TT^{-1} \circ  \mathcal{Q}^A_{\mathfrak{Q}_\TT V} \circ \mathfrak{Q}_\TT\,.
    \]
    {Moreover, the auxiliary interacting BV operator (on functionals) is
    \[
    \widehat{Q}_V(F) \doteq e^{-\frac{i}{\hbar}V} Q_{0,\TT} (F \eihbV),
    \]
    }
\end{definition}
\begin{remark}
	The retarded interacting quantum BV operator $Q^R_{V}$ can be written in two ways, using the two possible presentations for $\mathcal{Q}_0$:
	\[
	Q^R_{V} =R_{V,H}^{-1}\circ Q_{0,H} \circ R_{V,H} =R_{V,\TT}^{-1}\circ Q_{0,\TT} \circ R_{V,\TT}\,.
	\]
    Similarly for $Q_V^A$.
\end{remark}

{
\begin{proposition}\label{prop:auxiliaryBV}
    Let $Q_V\doteq Q_0 + \{V,\cdot\}$. For any local functional $F\in \Fcal_\loc$ we have
    \[
    \widehat{Q}_V(F) = Q_V F  -i\hbar\langle A'(V),F\rangle + \frac{i}{\hbar}\mathsf{QME}(V)F.
    \]
    Moreover, we have
    \[
    \widehat{Q}_V(1) = \frac{i}{\hbar}\mathsf{QME}(V).
    \]
\end{proposition}
\begin{proof}
    Let us write
    \begin{align*}
    \widehat{Q}_V(F) &= e^{-\frac{i}{\hbar} V} Q_{0,\TT}(F\eihbV) = e^{-\frac{i}{\hbar} V} \left(Q_0 - i\hbar\Lap_0^\ren\right)(F\eihbV) \\
    &= \left(Q_0 F + \{V,F\}^\std - i\hbar \langle A'(V),F\rangle\right) + \frac{i}{\hbar}\left(\mathsf{QME}(V)\right)F,
    \end{align*}
    and we conclude by recognizing $Q_V = Q_0 + \{V,\cdot\}^\std$.

    The second statement follows from the observation that $1\in\Fcal_\reg$ so that $\langle A'(V),1\rangle = \Lap(1)=0$ (or simply reworking the above calculation with $F=1$).
\end{proof}

\begin{definition}
    The renormalised $V$-dependent BV Laplacian is
    \[
    \Lap^{\ren}_V F\doteq \langle A'(V),F\rangle = \frac{i}{\hbar}(\widehat{Q}_V - Q_V)(F) + \frac{1}{\hbar^2}\mathsf{QME}(V)F.
    \]
\end{definition}
We can compare the renormalised free and BV Laplacians on multilocal functionals as follows
\begin{proposition}\label{prop:Lapdiff}
    Let $\{F_i\}_{i=1}^n$ be local functionals. Then
    \[
    (\Lap_V^\ren - \Lap_0^\ren)(F_1\cdot \dots \cdot F_n) = \left(\frac{\hbar}{i}\right)^{n}\left(\langle A^{(n)}(V),F_1\otimes \dots F_n\rangle - \langle A^{(n)}(0),F_1\otimes \dots F_n\rangle\right)
    \]
\end{proposition}
\begin{proof}
    This follows from Proposition \ref{prop:freeQBVop}
    \begin{align*}
    -i\hbar\Lap_V^\ren(F_1\dots F_n) &= (\widehat{Q}_V - Q_V - \frac{i}{\hbar} \mathsf{QME}(V))(F_1\dots F_n) \\
    &= e^{-\frac{i}{\hbar}V}Q_{0,\TT}(\eihbV F_1\dots F_n) - (Q_V + \frac{i}{\hbar} \mathsf{QME}(V))(F_1\dots F_n)\\
    &=\left(\frac{\hbar}{i}\right)^{n-1}\frac{d^n}{d\lambda_1\dots d\lambda_n}\left(\mathsf{QME}\left(V+\sum_i\lambda_iF_i\right)e^{\frac{i}{\hbar}\sum_i\lambda_iF_i}\right)\\
    &- (Q_V + \frac{i}{\hbar} \mathsf{QME}(V))(F_1\dots F_n)
    \end{align*}
    and, similarly, we have
    \begin{align*}
    -i\hbar\Lap_0^\ren(F_1\dots F_n) &= \left(\frac{\hbar}{i}\right)^{n-1}\frac{d^n}{d\lambda_1\dots d\lambda_n}\left(\mathsf{QME}\left(\sum_i\lambda_iF_i\right)e^{\frac{i}{\hbar}\sum_i\lambda_iF_i}\right)\\
    &- Q_0(F_1\dots F_n).
    \end{align*}
    With a straighforward calculation one checks that 
    \[
    (\Lap_V^\ren - \Lap_0^\ren)(F_1\cdot \dots \cdot F_n) = \left(\frac{\hbar}{i}\right)^{n} \left(\langle A^{(n)}(V),F_1\otimes \dots F_n\rangle - \langle A^{(n)}(0),F_1\otimes \dots F_n\rangle\right)
    \]
    observing that 
    \begin{multline*}
    \left(\frac{\hbar}{i}\right)^{n-1}\frac{d^n}{d\lambda_1\dots d\lambda_n}\left(\mathsf{QME}\left(V+\sum_i\lambda_iF_i\right)e^{\frac{i}{\hbar}(V+\sum_i\lambda_iF_i)}\right) \\
    = \left(\frac{\hbar}{i}\right)^{n-1}\left(Q_0V + \frac12\{V,V\} - i\hbar A(V)\right)\frac{d^n}{d\lambda_1\dots d\lambda_n}e^{\frac{i}{\hbar}\sum_i\lambda_iF_i} + \dots\\
    = \frac{i}{\hbar}\mathsf{QME}(V)F_1\dots F_n + \dots
    \end{multline*}
    which cancels with the opposite term in the same equation.
\end{proof}}

{

We are going to summarise and slightly reformulate now a series of results that can be found in \cite{FR3,BDFR_infty}. This reformulation will be useful in what follows
\begin{theorem}\label{thm:quantumBVopnoboundary:ren}
    Let $L_0$ be a quadratic generalised Lagrangian and $V$ be either a regular functional $V\in \Fcal_{\reg}$ or a local functional $V\in\Fcal_\loc$ (Definition \ref{def:regfun}), then we have that 
    \be
    \widehat{Q}_{V}(F) = \begin{cases}
        Q^R_V(F) + \frac{i}{\hbar}\mathsf{QME}(V)\star_{V}F\\
        Q^A_V(F) + \frac{i}{\hbar}F\star_{V}\mathsf{QME}(V)
    \end{cases}
    \ee
{(where $\star_V$ for local $V$ has to be understood as explained in Remark~\ref{rem:starVwelldef}.)}
    
    Moreover, if $L_0$ and $V$ satisfy the QME, i.e.
    \[
    Q_{0,\TT} \eihbV =0 \iff Q_0 \TT e^{\frac{i}{\hbar}V}=0 \iff \mathsf{QME}(V)=0
    \]
    then, for $V,F$ respectively regular or local, the retarded/advanced BV operators both act as
    \begin{equation}\label{e: BVop}
    \begin{cases}
        Q^{R/A}_{V} F=Q_0 F+\{V,F\}-i\hbar \Lap F=Q_{0,\TT} + \{V,F\} \\
        Q^{R/A}_{V} F=Q_0 F+\{V,F\}-i\hbar \Lap^{\ren}_V F = Q_{0,\TT}F + \{V,F\}-i\hbar (\Lap^{\ren}_V -\Lap^{\ren}_0)F 
    \end{cases}
    \end{equation}
    where\footnote{Note that sometimes $A'(0)$ is renormalised to zero. We do not need to do this here.}
    \[
    \Lap^{\ren}_\bullet F= \frac{d}{d t} A(\bullet+tF)\big|_{t=0} = \langle A'(\bullet),F\rangle\,.
    \]
    In particular:
    \be\label{e:fund-adjnoboundary}
    Q_{0,\TT}(F\eihbV) = Q^{R/A}_V(F)\eihbV + \frac{i}{\hbar}\left(\mathsf{QME}(V)\star_{V}F\right) \eihbV.
    \ee
\end{theorem}
\begin{proof}
From the definition of $R_{V,H}(F) = (\TT\eihbV)^{-1}\star_H(\TT(\eihbV F))$, using that $Q_0$ is a derivation of the star product we get:

\begin{align*}
Q^R_{V} F
    &=R_{V,H}^{-1}\circ Q_0\left( (\TT \eihbV)^{-1}\star_H   (\TT (\eihbV F)) \right)\\
    &=R_{V,H}^{-1}\Big(  -(\TT \eihbV)^{-1}\star_H \underbrace{(Q_0\circ \TT \eihbV)}_{\TT \rm \frac{i}{\hbar}QME(V)\eihbV}\star_H \underbrace{(\TT \eihbV)^{-1}\star_H (\TT (\eihbV F))}_{R_{V,H}(F)}\\
        &+ \underbrace{(\TT \eihbV)^{-1}\star_H (\TT\circ \TT^{-1}\circ Q_0(\TT (\eihbV F)))}_{R_{V,H}e^{-\frac{i}{\hbar}V}Q_{0,\TT}(\eihbV F)} \Big)\\
    &=-\frac{i}{\hbar}R_{V,H}^{-1}\left((\TT \eihbV)^{-1} \star_H(\TT(\mathsf{QME}(V)\eihbV))\star_H R_{V,H}(F)\right)\\
    &+R_{V,H}^{-1}\circ R_{V,H}(e^{-\frac{i}{\hbar}V} Q_{0,\TT} (\eihbV F) \\
    &=-\frac{i}{\hbar}R_{V,H}^{-1}\left(R_{V,H}(\mathsf{QME}(V)) \star_H R_{V,H}(F)\right) + e^{-\frac{i}{\hbar}V} Q_{0,\TT} (\eihbV F)\\
    &=-\frac{i}{\hbar}(\mathsf{QME}(V)) \star_{V} (F) + \widehat{Q}_V(F).
\end{align*}
whence we conclude the first claim.

Using Lemma \ref{lem:Q0hatexpression} to write $Q_{0,\TT}$ acting on regular functionals, we get 
\begin{align*}
    Q^R_VF &=-\frac{i}{\hbar}(\mathsf{QME}(V)) \star_V (F) + e^{-\frac{i}{\hbar}V} (Q_0-i\hbar \Lap)  (\eihbV F)\\
    &= -\frac{i}{\hbar}(\mathsf{QME}(V)) {\star_V} (F) + \frac{i}{\hbar}(\underbrace{Q_0V+\tfrac{1}{2}\{V,V\} -i\hbar \Lap V}_{\rm QME(V)})  F \\
    &+ Q_0F + \{V,F\}-i\hbar \Lap F\\
    &=Q_0 F+\{V,F\}-i\hbar \Lap F \qquad \text{QME}\\
    &= Q_{0,\TT}F +\{V,F\} \,,
\end{align*}
and similarly for its advanced version, proving the first one of Equation \ref{e: BVop}.

In the local case, we can use directly Proposition \ref{prop:auxiliaryBV} to show
\begin{align*}
Q^R_V F &= \widehat{Q}_V - \frac{i}{\hbar}(\mathsf{QME}(V)) \star_{V} (F)\\
&= Q_V F  -i\hbar\langle A'(V),F\rangle + \frac{i} {\hbar}\mathsf{QME}(V)F - \frac{i}{\hbar}(\mathsf{QME}(V)) \star_{V} (F).
\end{align*}

Recalling that $\Lap_0^\ren F= \langle A'(0),F\rangle$ and using Proposition \ref{prop:Lapdiff} with $n=1$, the retarded/advanced BV operators can be alternatively characterised as
\[
Q^{R/A}_V F = Q_{0,\TT}F + \{V,F\}-i\hbar (\Lap^{\ren}_V -\Lap^{\ren}_0)F.
\]
\end{proof}
}

\subsubsection{QME, Maurer--Cartan elements and curved anomalous master ward identity}\label{sec:curvedAMWI}
    The QME (for regular functionals) can be thought of as the Maurer--Cartan equation for the element $V\in\Fcal_{\reg}$ in the (1-shifted) dgLa given by the BV differential deformed by the bare BV Laplacian, and the BV bracket:
    \[
    \left(\Fcal_\reg, Q_0 - i\hbar \Lap, \{\cdot,\cdot\}\right), \qquad \mathsf{QME}(V) = (Q_0 - i\hbar\Lap) V + \frac12\{V,V\} \stackrel{!}{=} 0.
    \]
    Given a MC element (such as $V$) one defines a new dgLa structure (on the same graded vector space) by twisting the differential as (higher brackets are unchanged because we started from a dgLa)
    \[
    Q_0 - i\hbar\Lap \leadsto Q_0 - i\hbar\Lap + \{V,\cdot\}
    \]
    
    For the renormalised version of this statement, take any $V,F\in \Fcal_{\TT}\equiv \im(\TT)$ (cf.\ Equation \eqref{e:ImTT}) and use Theorem \ref{thm:quantumBVopnoboundary:ren} to write
    \[
    Q^{R}_{V} F=Q_0 F+\{V,F\}-i\hbar \Lap_V^{\ren} F + \frac{i}{\hbar}\mathsf{QME}(V)\cdot F - \frac{i}{\hbar}\mathsf{QMF}(V){\star_V}F \,
    \]
    where $\Lap_V^{\ren}$ is the \emph{renormalized}, $V$-dependent BV Laplacian, acting on $\Fcal_{\TT}$, {which can be computed in terms of the Anomaly $A(V)$ as (for $F$ local)
    \[
    \Lap_V^\ren F = i \langle A'(V),F\rangle.    
    \]
    Then, assuming $\mathsf{QME}(V)=0$, going from $Q_{0,\TT}$ to $Q_{V,\TT}^{R/A}$ can be seen as the twist
    \[
    Q_0 - i\hbar\Lap_0^\ren \leadsto Q_0 - i\hbar\Lap_0^\ren + \{V,\cdot\} - i\hbar(\Lap_V^\ren - \Lap_0^\ren)
    \]
    and ($F$ local)
    \[
    \Lap_V^\ren - \Lap_0^\ren = \langle A'(V) - A'(0),\cdot\rangle.
    \]
    More generally, one can compute this difference for $F=F_1 \dots F_n$ multilocal and obtain
    \[
    (\Lap_V^{\ren} - \Lap_0^\ren) (F) = \langle A^{(n)}(V) - A^{(n)}(0),F\rangle
    \]

    }

    Note that in the previous works \eqref{e:[]0} is assumed to vanish identically, i.e. the QME is assumed to hold on the nose. When $V$ is not a Maurer--Cartan element, i.e. when the QME is spoiled (for example, as we will see, by boundary terms) we have that
    \begin{enumerate}
        \item from abstract algebraic considerations one can still build an $L_\infty$ algebra via twist, but this will be \emph{curved} in general,\footnote{We thank J.\ Schnitzer for suggesting this to us.}
        \item the proof that $[\dots ]^F_n$ is an $L_\infty$ structure provided in \cite{BDFR_infty} does not use the fact that the curvature vanishes, i.e. $[-]_0^F\equiv0$, so it can be used also in the general case presented here.
    \end{enumerate}
     Indeed, as we will see, using the definition of QME given here, we will have generally a statement of the form $\mathsf{QME}(V)=J^\hbar$ for some $J^\hbar$, potentially a power series as a result of renormalisation (see Definition \ref{e:ABSmQME} and subsequent realisations). The difference then is that the general scenario we are working with here gives rise to a curved $L_\infty$ algebra, because $V$ is not Maurer--Cartan. Then the $0$-ary bracket obtained by renormalisation---called curvature---does not vanish and is instead controlled by a (potentially renormalised) Noether current.

    It is important to notice that, again following \cite{BDFR_infty}, one can rephrase the QME as
    \[
    \mathsf{QME}(V) = 0 \iff \widehat{Q}_V(1) = 0 \iff Q_{0,\TT}(e^{\frac{i}{\hbar}V}) = 0.
    \]
    This means that the interacting quantum BV operator is compatible with the unital associative algebraic structure on functionals if and only if the QME is satisfied. When this is not the case, the failure is indeed encoded by a nonzero curvature. It is worth mentioning, in closing, that $\widehat{Q}_V(1) = 0$ is equivalent to 
    \[
    \Lap_V^\ren (1) = 0
    \]
    since $\Lap_V^\ren = i(\widehat{Q}_V - Q_V)$ and $Q_V(1)=0$ as it is a derivation.

\section{BV-BFV in pAQFT}\label{sec:BVBFVPAQFT}

An important feature of pAQFT is that one can introduce an interaction term to perturb a free theory by means of a smearing with a compactly supported function $f$, a.k.a.\ a cutoff. 

This framework can also be used to approximate a sharp cutoff---e.g. taking a sequence of test functions $f_n$ that converges to the characteristic function of some closed compact subset $C\subset M$ (possibly with boundary). We will present an approach that instead of sharpening cutoffs, will smoothen submanifolds, with the idea of recovering the sharp-submanifold case as a limit.

\subsection{Smoothened BV-BFV data}
In this section, we will specify relative BV data (a.k.a.\ BV-BFV data) on compact regions. To make this compatible with the pAQFT philosophy, and in particular with restriction of BV data on $M$, we introduce the following:

\begin{definition}[Smoothened boundaries]\label{def:smoothenedbdry}
    Let $C\subset M$ be a submanifold with closed boundary $\partial C$.\footnote{We would like $C$ to model a closed, compact region, which need not be connected. However, since locality allows us to consider each connected subregion separately at this stage, we can also assume $C$ connected.} $C\subset M$.
    Let $f\in\Ci_c(M)$ such that $C\subset \supp f$ and which is equal to 1 on some compact set $C'\subsetneq C$, treated as ``the bulk.'' Let $N$ denote the open set $(\supp f)^\circ\setminus C'$. We think of $(C,N,f)$ as a \emph{submanifold $C$ with smoothened boundary}.
\end{definition}

The setup is shown on Figure~\ref{fig:supports}.

\begin{figure}[ht!]
\begin{tikzpicture}[scale=0.9, line cap=round, line join=round]

\draw[black, very thick, smooth cycle] plot coordinates {
  (-4.6,-0.2) (-4.7,0.4) (-4.3,1.0) (-3.6,1.6)
  (-2.8,2.0) (-2.0,2.2) (-1.1,2.15) (-0.4,2.05)
  (0.3,2.0) (1.0,2.1) (1.8,2.35) (2.7,2.25)
  (3.5,1.8) (4.0,1.1) (4.2,0.3) (4.05,-0.4)
  (3.6,-1.1) (2.8,-1.8) (1.9,-2.1) (0.8,-2.25)
  (-0.1,-2.1) (-0.8,-1.5) (-1.4,-0.8) (-2.1,-0.5)
  (-3.0,-0.55) (-3.8,-0.8) (-4.4,-0.7)
};

\path[pattern=north east lines, pattern color=lightgray, even odd rule]
  plot[smooth cycle] coordinates {
    (-1.7,0.4) (-1.4,1.0) (-0.8,1.45)
    (0.2,1.45) (1.0,1.55) (1.8,1.7) (2.6,1.7)
    (3.2,1.25) (3.7,0.7) (3.8,0.0) (3.55,-0.7)
    (2.9,-1.3) (2.0,-1.75) (1.0,-1.85) (0.0,-1.55)
    (-0.7,-1.1) (-1.1,-0.45) (-1.5,-0.15)
  }
  plot[smooth cycle] coordinates {
    (-0.15,0.05) (0.25,0.45) (0.85,0.35) (1.45,0.5)
    (2.05,0.35) (2.35,-0.15) (2.0,-0.55) (1.45,-0.9)
    (0.8,-0.9) (0.25,-0.65) (-0.2,-0.35)
  };

\node at (-3.2,0.2) {$M$};
\node[cyan!80!blue] at (1.0,1.3) {$\supp f$};
\node[blue!70!purple] at (1.0,0.65) {$C$};
\node[red] at (1.0,-0.2) {$C'$};

\draw[cyan!80!blue, very thick, smooth cycle] plot coordinates {
  (-1.7,0.4) (-1.4,1.0) (-0.8,1.45)
  (0.2,1.45) (1.0,1.55) (1.8,1.7) (2.6,1.7)
  (3.2,1.25) (3.7,0.7) (3.8,0.0) (3.55,-0.7)
  (2.9,-1.3) (2.0,-1.75) (1.0,-1.85) (0.0,-1.55)
  (-0.7,-1.1) (-1.1,-0.45) (-1.5,-0.15)
};

\draw[blue!70!purple, very thick, smooth cycle] plot coordinates {
  (-0.9,0.15) (-0.5,0.75) (0.15,1.0) (0.95,1.0)
  (1.65,1.15) (2.45,1.15) (3.0,0.75) (3.15,0.15)
  (2.9,-0.55) (2.3,-1.1) (1.55,-1.45) (0.75,-1.45)
  (0.1,-1.1) (-0.55,-0.9) (-0.9,-0.45)
};

\draw[red, very thick, smooth cycle] plot coordinates {
  (-0.15,0.05) (0.25,0.45) (0.85,0.35) (1.45,0.5)
  (2.05,0.35) (2.35,-0.15) (2.0,-0.55) (1.45,-0.9)
  (0.8,-0.9) (0.25,-0.65) (-0.2,-0.35)
};

\end{tikzpicture}
\caption{Geometrical setup from Definition~\ref{def:smoothenedbdry}. The hatched area is the smoothened boundary of the region $C$. The sharp boundary of $C$ is the navy blue contour.\label{fig:supports}}
\end{figure}
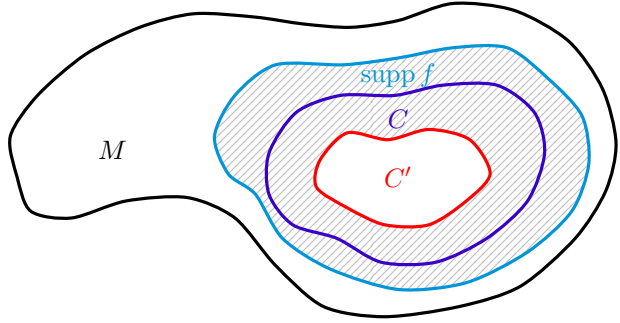

\begin{definition}[Smoothened BV-BFV theory]\label{def:BV-BFV}
A relative BV theory on a compact region with smoothened boundary $(C,N,f)$ is the data of 
\begin{itemize}
    \item a BV theory $(\Ecal,\Omega,L)$ over $C$;
    \item a BFV theory $(\Ecal^\de,\Omega^\de,L^\de)$ over $\de C$;
    \item a surjective submersion $\pi\colon \Ecal \to \Ecal^\de$,
\end{itemize}
    such that, 
    \begin{enumerate}
        
    \item denoting by $J^\std\in\oloc^{0,(1)}(\Ecal)$ we have the classical master equation for $L$ 
        \[
        \frac12\{L[f],L[f]\}^\std = - J^\std[fdf]\,
        \]
    \item there exists a generalised $1$-form 1-current $\theta\colon \Omega^1_c(M) \to \Omega^1_\loc(\Ecal)$ such that
    \be\label{e:smaredBVBFVeqt}
    \iota_{Q[f]}\Omega = \delta L[f] - \theta[df], \quad \theta[df]=\int_M  df {\bth}
    \ee
    for some ${\bth}\in \oloc^{1,\top -1}(\Ecal \times M)$;
    \item \label{l:BVBFV3} the following holds:
    \[
    \iota_{Q[f]}\iota_{Q[f]}\delta \theta [df] = 0.
    \]
    
    \end{enumerate}
    
    The smoothened BV-BFV theory is said to be \emph{convergent} (to a sharp BV-BFV theory) if, upon picking a sequence $(f_n)_{n\in\NN}$ such that $df_n$ and $f_ndf_n$ converge to delta distributions supported on $\partial C$, we have that\footnote{This defines a generalised Lagrangian on functions on $\de C$ by $(\lim_{n\to\infty}J[f_ndf_n])[g] = L^\de[\iota_{\de C}^*g]$ or, in other words, by $f_ndf_n$ limiting to a distribution supported on $\de C$.}
        \[
        \lim_{n\rightarrow \infty} J[f_ndf_n]= \pi^* L^{\de}\qquad \lim_{n\to \infty} \delta\theta[df_n] = \pi^*\Omega^\partial.
        \]
    Finally, a convergent smoothened BV-BFV theory is said to be exact if $\Omega^\de=\delta \theta^\partial$ is exact and $\lim_{n\to \infty} \theta[df_n] = \pi^*\theta^\partial$.
\end{definition}

Recall that, by Proposition \ref{prop:CMEequivalence}, $L[f]$ satisfies the standard classical master equation $\{L[f],L[f]\}^\std \sim 0$ iff $\mathbf{L}$ satisfies the densitised classical master equation, i.e. there exists a $(0,\top -1)$ local form $\mathbf{J}^\src\in\oloc^{0,\top-1}$ such that\footnote{Note that both the factor of 2 and the minus sign come from the integration by parts $\int d\mathbf{J}^\src f^2 = -2 \int \mathbf{J}^\src fdf$}
\be
\frac12\iota_Q\iota_Q\bom = d\mathbf{J}^\src, \quad \text{i.e.} \quad \frac12\{L[f],L[f]\}^\src = - 2 J^\src[fdf],
\ee
where $J^\src\in\iloc^{0,(1)}(\Ecal)$ is the associated $1$-current. Indeed, given the BV generalised Lagrangian $L[f] = \int \mathbf{L} f$ we can feed it into two types of brackets, namely $\{\cdot,\cdot\}^\std$ and $\{\cdot,\cdot\}^\src$ and the classical master equation tells us that the result is equivalent to zero in both cases. However, the result does not generally coincide. The general relation between these two brackets is outlined in Proposition \ref{prop:fullbracketvsaction}, but in the specific case of a solution of the classical master equation we have the following:

\begin{proposition}\label{prop:descent}
    Consider a smoothened BV-BFV theory as above.  Define the \emph{standard generalised Noether current} as
    \[
    J^\std[fdf] \doteq 2 J^\src[fdf] - (\iota_Q\theta)[fdf].
    \]
    Then
    \[
    J^\std[fdf] = - \frac12\{L[f],L[f]\}^\std.
    \]
    Moreover, denoting by $\alpha$ a primitive\footnote{This always exists since $\Ecal$ is a cotangent bundle.} of the BV form $\Omega= \delta \alpha$, and by $\mathbb{L}$ the $0$th generalised Noether current (Definition \ref{def:relBVtheory})
    \[
    \mathbb{L}[f] \doteq L[f] - \iota_{Q[f]}\alpha
    \]
    we have the descent equation
    \[
    Q[f]\left(\mathbb{L}[f]\right) = J^\std[fdf].
    \]
    Finally, we have that
    \[
    Q[f](J^\std[fdf]) = 0.
    \]
\end{proposition}
\begin{proof}
    From Corollary \ref{cor:srcnormalbracket} we have that 
    \[
    \{L[f],L[f]\}^\std = \{L[f],L[f]\}^\src + 2 \iota_{X_L}\theta_L[fdf] = -4 J^\src[fdf]+ 2 \iota_{X_L}\theta_L[fdf]
    \]
    By definition we have that $J^\std[fdf] \doteq 2 J^\src[fdf] - (\iota_Q\theta)[fdf]$, whence the claim upon renaming $X_L\equiv Q$ and $\theta_L\equiv \theta$:
    \[
    \frac12 \{L[f],L[f]\}^\std = -J^\std[fdf]. 
    \]
    
    Now, let us compute $Q[f]\left(\mathbb{L}[f]\right)$ using \eqref{e:smaredBVBFVeqt}
    \begin{align*}
        Q[f]\left(\mathbb{L}[f]\right) &= Q[f]\left(L[f] - \iota_{Q[f]}\alpha\right) = \iota_{Q[f]}\iota_{Q[f]}\Omega + \iota_{Q[f]}(\theta[df]) - L_{Q[f]}\iota_{Q[f]}\alpha \\
        &=2dJ^\src[f^2] + (\iota_{Q}\theta)[fdf] - \frac12 \iota_{Q[f]}\iota_{Q[f]}\delta \alpha \\
        &= dJ^\src[f^2] + dJ^\src[f^2] + (\iota_{Q}\theta)[fdf] - dJ^\src[f^2]\\
        &=-2J^\src[fdf] + (\iota_{Q}\theta)[fdf] = J[fdf].
    \end{align*}
    where we used Cartan's formula to write $L_{Q[f]}\iota_{Q[f]}\alpha = \frac12 \iota_{Q[f]}\iota_{Q[f]}\delta \alpha = -J^\src[fdf]$ together with $[Q[f],Q[f]]=0$ and the definition of the generalised current $J$. The last statement follows from $Q^2=0$.
\end{proof}

In order to prepare  a smoothened BV-BFV theory on $(C,N,f)$ for perturbative quantisation, we discuss its free-interacting splitting.

\begin{definition}[Noether currents for free/interacting splittings]
    Let $L=L_0 + V$ be a free-interacting splitting for a generalised Lagrangian. The \emph{standard and source free generalised Noether currents} $J_0^{\std},J_0^{\src}\in\iloc^{0,(1)}(\Ecal)$ are, respectively
    \[
    \frac12\{L_0[f],L_0[f]\}^\std = -J_0^{\std}[fdf]\,\qquad \frac12\{L_0[f],L_0[f]\}^{\src} = -J_0^\src[fdf].
    \]
    Additionally, the standard and source interacting Noether currents are 
    \[
    J_I^{\std/\src} \doteq J^{\std/\src} - J_0^{\std/\src}.
    \]
    
    \noindent We call \emph{linear Noether currents} $J_{\mathrm{lin}}^{\std},J_{\mathrm{lin}}^{\src}\in\iloc^{0,(1)}(\Ecal)$ the free generalised Nother currents associated to the linear BV splitting (Definition \ref{prop:quad+minimalsplit}):
    \[
    \frac12\{L_{\mathrm{lin}}[f],L_{\mathrm{lin}}[f]\}^\std = -J_{\mathrm{lin}}^{\std}[fdf],\qquad \frac12\{L_{\mathrm{lin}}[f],L_{\mathrm{lin}}[f]\}^{\src} = -J_{\mathrm{lin}}^\src[fdf],
    \]
    where $L_{\mathrm{lin}}$ is the quadratic factor in the linear BV splitting. 

    \noindent We call ``minimal Noether currents'' $J_{00}^{\std},J_{00}^{\src}\in\iloc^{0,\top-1}(\Ecal)$ the  free generalised Noether current associated to the minimal splitting:
    \[
    \frac12\{L_{00}[f],L_{00}[f]\}^\std = -J_{00}^{\std}[fdf],\qquad \frac12\{L_{00}[f],L_{00}[f]\}^{\src} = -J_{00}^\src[fdf].
    \]
\end{definition}

\begin{lemma}\label{lem:vanishingminimalcurrent}
    The standard and source minimal Noether currents $J_{00}^{\std},J_{00}^{\src}$ vanish. Hence, 
    \[
    J_I^{\std/\src}\equiv J^{\std/\src} - J_{00}^{\std/\src} = J^{\std/\src}
    \] 
    for the minimal splitting.
\end{lemma}
\begin{proof}
    This is immediate from the fact that $\{L_{00},L_{00}\}^\std \equiv \{L_{00},L_{00}\}^{\src}\equiv 0$ since $L_{00}$ only depends on degree $0$ fields.
\end{proof}

\begin{remark}\label{rmk:etasmoothening}
    {Note that our smoothened-boundary expressions will depend on a combination of $f$ and $df$, so that in principle we should allow for general compactly supported $1$ forms $d\eta(f)$ determined by $f$ to localise to $\supp(df)$. The most relevant ones for our purposes will be $\eta_1(f)=\frac12f^2$, leading to $d\eta_1(f)=fdf$, and $\eta_2(f)=f^2$ leading to $d\eta_2(f)=2 fdf$. The reason for this is that, depending on whether one wants to use $J^\std$ or $J^\src$ or some other modifications thereof, one has a slightly different version of the modified quantum master equation. Indeed, one could write Proposition \ref{prop:descent} as
    \[
    J^\std[d\eta_1(f)] = J^\src[d\eta_2(f)] - (\iota_Q\theta)[d\eta_1(f)].
    \]
    This definition comprises all possibilities. More generally, we can think of a smoothened region $N\subset M$ to be detected by the one form $d\eta(f)$ with $f$ compactly supported. Then, a current supported on $N$ will be written as $K[d\eta(f)]$.}
\end{remark}

\subsection{Modified Quantum Master Equation}

The classical master equation is generically violated in the presence of boundaries owing to Proposition \ref{prop:CMEequivalence}. We cannot therefore expect the quantum master equation to be satisfied as is, but we will need to correct it taking into account the boundary ($d$-exact) contributions. 

Henceforth we will consider $(C,N,f)$ be a submanifold with smoothened boundary associated to the the sharp submanifold with boundary $C\hookrightarrow M$ as in Definition \ref{def:smoothenedbdry}. In principle, however, the following definition applies to any open region $N\subset M$. {We will also use the compactly supported 1-forms $d\eta(f)$ to compactly write the currents' dependence on both $f$ and $df$ (cf. Remark \ref{rmk:etasmoothening}).}

\begin{definition}[Modified Quantum Master Equation and Generalised Quantum Noether Current]\label{def:mQME}
     We say that $A\in \fA$ satisfies the \emph{modified quantum master equation (mQME)} relative to the region $N\subset M$ if there exists $\mathcal{J} \in \mathfrak{A}(N)$ such that
    \be\label{e:ABSmQME}
    \mathsf{(mQME)} \qquad (\mathcal{Q}_0 + \frac{i}{\hbar} (\cdot)\circT \mathcal{J}) \Scal(A)=0.
    \ee
    When $L=L_0 + V$ is a generalised Lagrangian in a free-interacting splitting, we say that $L$ satisfies the mQME iff $\mathfrak{Q}_\TT(V)\in\mathfrak{A}$ does, and if $\mathcal{J} = \mathfrak{Q}_\TT(J)$ for some generalised current $J$. In that case we call $J$ the generalised quantum Noether current associated to $L$.
\end{definition}
Recall (Remark~\ref{rem:QME presentations}) that in a free-interacting splitting $L=L_0+V$ the quantum master equation \eqref{eq:QME} can be written in a given presentation by postcomposing with  $\mathfrak{Q}_H^{-1}$ or $\mathfrak{Q}_\TT^{-1}$. In particular, the definition above requires that
\[
(\mathfrak{Q}_\TT^{-1}\circ(\mathcal{J}\circT\mathcal{S})\circ \mathfrak{Q}_\TT(V))[f] = J[d\eta(f)] \eihbV
\]
for some compactly supported function $\eta\in \Omega^0_c(M)$, and where $\mathcal{J} = \mathfrak{Q}_\TT(J)$. 
Hence we have that in the $\TT$-presentation, the mQME reads (note that $\cdot$ is commutative) 
\be\label{e:MQMETT}
\mathsf{(mQME)}_\TT \qquad (Q_0[f]-i\hbar \Lap_0^{\ren}[f] + \frac{i}{\hbar}J[d\eta(f)]\cdot) e^{iV[f]/\hbar}=0.
\ee
The renormalised BV Laplacian arises in virtue of the Anomalous Master Ward Identity. Indeed one writes the above mQME using Theorem \ref{thm:AMWI} as
\[
Q_{0,\TT}[f] e^{\frac{i}{\hbar} V[f]} = \frac{i}{\hbar}\mathsf{QME}(V[f])e^{\frac{i}{\hbar} V[f]} = - \frac{i}{\hbar} J[d\eta(f)]e^{\frac{i}{\hbar} V[f]}
\]
The mQME in the hybrid-$H$ presentation, instead, reads 
\[
\mathsf{(mQME)}_H \qquad Q_0[f] \TT\eihbV + \frac{i}{\hbar}\TT(\eihbV J[d\eta(f)])  = 0.
\]

As a consequence, we immediately get
\begin{lemma}\label{lem:reftofreecase}
    A generalised Lagrangian in a free-interacting splitting satisfies the modified quantum master equation relative to the region $N\subset M$ with generalised current $J$ whenever
    \[
    \mathsf{QME}(V[f]) = - J[d\eta(f)]
    \]
    where $d\eta(f)$ is the characteristic smoothening of $N\subset M$ (see Remark \ref{rmk:etasmoothening}).
\end{lemma}
\begin{proof}
    From the AMWI of Theorem \ref{thm:AMWI} we have
    \[
    Q_0\TT \eihbV = \frac{i}{\hbar}\TT(\mathsf{QME}(V)\eihbV).
    \]
    On the other hand, the modified quantum master equations is written as
    \[
    Q_0\TT \eihbV = -\frac{i}{\hbar}\TT(J \eihbV).
    \]
\end{proof}

\begin{proposition}\label{prop:mQMEconditions}
Let $L_0$, $L$ be generalized Lagrangians such that $L_0$ is quadratic, {$V=L-L_0$ is regular}, and the following conditions hold, with $\Lap \doteq Q_{0,\TT} - Q_0$ the bare BV Laplacian:
    \begin{subequations}\label{e:mQMEconditions}
    \be\label{e:freeCMRmQME}
    \frac12\{L_0[f],L_0[f]\}^{\std}-i\hbar \Lap (L_0[f]))+ J_0^{\std}[fdf]=0
    \ee
    and 
    \be\label{e:totCMRmQME}
    \frac12\{L[f],L[f]\}^{\std}-i\hbar \Lap (L[f]))+J^{\std}[fdf]=0\,.
    \ee
Then the mQME in the $\TT$ presentation reads
    \be\label{e:ourmQME}
    \left(Q_{0,\TT}[f] + \frac{i}{\hbar}J_V^{\std}[fdf]\right)e^{\frac{i}{\hbar}V}=0\,,
    \ee\end{subequations}
where 
\be\label{e:adjustedcurrent}
J_V^{\std} \doteq J_I^{\std} - B_0(V) \equiv J^{\std} - J_0^{\std} - B_0(V), \qquad B_0(V) \doteq Q_0(V) - \{L_0,V\}^{\std}
\ee
is called the adjusted interacting Noether current, and it satisfies
\begin{equation}\label{e:adjustedinteractingnoethercurrent}
Q_0[f](V[f]) + \frac12 \{V[f],V[f]\}^\std +J_V^{\std}[fdf] = 0. 
\end{equation}
\end{proposition}

\begin{proof}
    In the calculation below we suppress the test functions (in $Q_0[f]$ and $V[f]$) for simplicity (recall that by assumption $\{L_{(0)}[f],L_{(0)}[f]\}^\std=-J_{(0)}^\std[fdf]$ at lowest order in $\hbar$)
    \begin{align*}
    Q_{0,\TT} e^{\frac{i}{\hbar}V} 
        &= (Q_0 -i\hbar \Lap) e^{\frac{i}{\hbar}V} \\
        &= \frac{i}{\hbar}(Q_0(V) + \frac12\Lap(V^2) - i\hbar \Lap V)e^{\frac{i}{\hbar}V}\\
        &= \frac{i}{\hbar}(\{L_0,V\}^{\std} + B_0(V) + \frac12\{V,V\}^{\std} - i\hbar \Lap V)e^{\frac{i}{\hbar}V}\\
        &=\frac{i}{\hbar}\left(\frac12 \{L,L\}^{\std} - i\hbar  \Lap L - \frac12\{L_0,L_0\}^{\std} -i\hbar  \Lap L_0 + B_0(V) \right)e^{\frac{i}{\hbar}V}\\
        &=  \frac{i}{\hbar}(-J^{\std} + J_0^{\std} + B_0(V))e^{\frac{i}{\hbar}V}.\\
        &=-\frac{i}{\hbar}J_V^{\std}e^{\frac{i}{\hbar}V},
    \end{align*}
    where we have defined
    \[
        B_0(V)\doteq Q_0(V) - \{L_0,V\}^\std,
    \]
    which, owing to Proposition \ref{prop:fullbracketvsaction}, is equivalent to $0$, i.e.\ $B_0(V)\sim 0$.
    Then, we compute 
    \begin{align*}
        -J_V^{\std} 
            &= -J^{\std} + J_0^{\std} + B_0(V) \\
            &= \frac12\{L,L\}^{\std} - \frac12\{L_0,L_0\}^\std + Q_0(V) - \{L_0,V\}^{\std} \\
            &= Q_0(V) + \frac12 \{V,V\}^\std.
    \end{align*}
\end{proof}

    The notion of modified quantum master equation was introduced and formalised by Cattaneo, Mnev and Reshetikhin in \cite{CMR2}, as the quantum version of the modified classical master equation \cite{CMR1}. We will see that our version of the mQME recovers the salient parts of CMR's construction in Section \ref{sec:CMR}.    In the language of \cite{CMR2}, Equations \eqref{e:freeCMRmQME} and \eqref{e:totCMRmQME} are the modified quantum master equation for the free and full Lagrangian respectively. Equation \eqref{e:ourmQME} is our version of the above (see Definition \ref{def:mQME}).

    Note that the quantity that we called `adjusted interacting Noether current' in \eqref{e:adjustedcurrent} features also in \cite{GwilliamRejzner} (up to a global sign), where it is proven to vanish under certain conditions (orthogonal to ours). Indeed, in the cited paper the setup is such that boundary contributions are discarded, but we see that they would contribute exactly by the term $J_V^\std$.

{In the same way we introduced the interacting data in the no-boundary case, we define M\o ller maps and interacting BV operators just as we did in Definition \ref{df: BVint on functionals}.

\begin{definition}\label{def:intBVopsBoundary}
    The retarded/advanced, quantum interacting BV operators $Q^{R/A}_V$ and the auxiliary interacting BV operator $\widehat{Q}_V$ are as in Definition \ref{df: BVint on functionals}.
\end{definition}

The auxiliary BV operator allows us to characterise the mQME in a clean way.

\begin{proposition}\label{prop:auxopBVBFV}
    Let $L=L_0+V$ a generalised Lagrangian in a free-interacting splitting. Then, it satisfies the modified Quantum Master Equation with current $J$ iff
    \be\label{e:AuxiliaryOperatorCurvature}
        \widehat{Q}_V(1) = - \frac{i}{\hbar}J.
        \ee
\end{proposition}
\begin{proof}
    This is a direct consequence of Proposition \ref{prop:auxiliaryBV}, that tells us $\widehat{Q}_V(1) = \frac{i}{\hbar}\mathsf{QME}(V)$ and Lemma \ref{lem:reftofreecase}, which states $\mathsf{QME}(V)=-J$.
\end{proof}

In order to discuss the relation between $Q^{R/A}_V$ and the auxiliary operator $\widehat{Q}_V$, we need to observe the following.
}

\begin{definition}
    The \emph{retarded/advanced quantum interacting Noether currents} for the interaction $\mathcal{V}=\mathfrak{Q}_\TT(V)$ and the generalised Noether current $\mathcal{J}=\frac{i}{\hbar}\mathfrak{Q}_\TT(J)$ are
    \[
    \mathcal{J}^\mathcal{R}(A) = A \bullet \mathcal{R}_{\mathcal{V}}(\mathcal{J}),
        \qquad \mathcal{J}^\mathcal{A}(A) =  \mathcal{A}_{\mathcal{V}}(\mathcal{J}) \bullet A
    \]
    for $A\in\mathfrak{A}$. In the $\TT$ presentation and the hybrid $H$ presentation, they read
    \[
    \begin{cases}
    {J^R_{\TT}}(F) = \frac{i}{\hbar}F\star_\TT R_{V,\TT}(J[d\eta]) \\
    {J^A_{\TT}}(F) = \frac{i}{\hbar}A_{V,\TT}(J[d\eta]) \star_\TT F
    \end{cases}
    \quad
    \begin{cases}
    {J^R_{H}}(F) = \frac{i}{\hbar}F\star_H R_{V,H}(J[d\eta]) \\
    {J^A_{H}}(F) = \frac{i}{\hbar}A_{V,H}(J[d\eta]) \star_H F
    \end{cases}
    \]
    Moreover, the \emph{retarded/advanced modified quantum BV operators} in the $*$ presentation are
    \be\label{e:RAmodqBV}
        \widetilde{Q}^{R/A}_{J,*} = Q_{0,*} + {J^{R/A}_{*}}
    \ee
\end{definition}

\begin{lemma}\label{lem:quantJcurrent}
     Let $A,B,V\in \mathfrak{A}$ and denote by $\mathcal{R}_V$ (resp. $\mathcal{A}_V$) the associated M\o ller map(s). Then
     \begin{align*}
     \mathcal{J}^R(\mathcal{S}(A)\circT B ) &= \mathcal{S}(A) \circT\left( B\bullet_{\rm int} \mathcal{J}\right),\\
     \mathcal{J}^A(B\circT \mathcal{S}(A) ) &= \left( \mathcal{J}\bullet_{\rm int} B\right)\circT \mathcal{S}(A).
     \end{align*}
     In the $\TT$-presentation this implies
     \begin{align*}
     {J^R_\TT}(\eihbV F) & = \frac{i}{\hbar}\eihbV(F \star_{V} J[d\eta])&\\
     \end{align*}
     while in the $H$-presentation we have
     \[
     {J^R_H}(\TT \eihbV  F) = {\frac{i}{\hbar}\TT(\eihbV(F \star_{V} J[d\eta]))}
     \]
     and similalrly for their advanced versions.
\end{lemma}
\begin{proof}
    Using again the definition $\mathcal{S}(A)\bullet\mathcal{R}_A(B)=\mathcal{S}(A)\circT B$ twice, we compute
    \begin{align*}
        (\mathcal{S}(A)\circT B)\bullet \mathcal{R}_{A}(\mathcal{J}) &= (\mathcal{S}(A)\bullet \mathcal{R}_{A}(B))\bullet \mathcal{R}_{A}(\mathcal{J}) \\
        &= \mathcal{S}(A)\bullet (\mathcal{R}_{A}(B)\bullet \mathcal{R}_{A}(\mathcal{J})) \\
        &= \mathcal{S}(A) \bullet \mathcal{R}_{A}(B\bullet_{\rm int} \mathcal{J})\\
        &=\mathcal{S}(A)\circT(B\bullet_{\rm int}\mathcal{J})
    \end{align*}
    The rest is obtained by choosing $A=\mathfrak{Q}_\TT(V)$ and $B=\mathfrak{Q}_\TT(F)$, using table \ref{table:presentations} to translate between presentations, and recalling that the S matrix in the $\TT$-presentation is given by the functional $S_{\TT}(V) = \eihbV$; in the hybrid $H$ presentation
  it is given by $S_H(V)=\TT e^{\frac{i}{\hbar}V}$.
\end{proof}

\begin{theorem}\label{thm:mQMEStates}
    Let $(C,N,f)$ be a smoothening of the submanifold with boundary $C\hookrightarrow M$ as in Definition \ref{def:smoothenedbdry}. Then, a generalised Lagrangian in a free-interacting splitting $L=L_0 + V$ (for V regular) satisfies the modified quantum master equation (mQME) in the region $N$ with associated quantum generalised Noether current $\mathcal{J}$ (Equation \eqref{e:ABSmQME} if and only if
    \[
    \mathcal{Q}_0 \mathcal{S}(\mathfrak{Q}_\TT(V)) = 
        \begin{cases}
            -\mathcal{S}(\mathfrak{Q}_\TT(V))\bullet \mathcal{R}_{\mathfrak{Q}_\TT(V)}(\mathcal{J})\doteq -\mathcal{J}^{\mathcal{R}}(\mathcal{S}(\mathfrak{Q}_\TT(V)))\\
            -\mathcal{A}_{\mathfrak{Q}_\TT(V)}(\mathcal{J})\bullet\mathcal{S}(\mathfrak{Q}_\TT(V)) \doteq -\mathcal{J}^{\mathcal{A}}(\mathcal{S}(\mathfrak{Q}_\TT(V)))
        \end{cases}
    \]
    where $\mathcal{R},\mathcal{A}$ are the M\o ller operators of Definition \ref{def:moller}. In the $\TT$-presentation this reads
    \[
    \widetilde{Q}^{R/A}_{J,\TT}\eihbV = \left(Q_0 - i\hbar \Lap + \frac{i}{\hbar}(\cdot) \star_\TT R_{V,\TT}(J[d\eta])\right)\eihbV = 0
    \]
    while in the hybrid-$H$ presentation we have
    \[
    \widetilde{Q}^{R/A}_{J,H}\TT \eihbV=\left(Q_0 + \frac{i}{\hbar}(\cdot) \star_H R_{V,H}(J[d\eta])\right) \TT\eihbV  = 0
    \]
\end{theorem}

\begin{remark}
    {Before moving to the proof we must observe that the statement is given for $V$ \emph{regular}. This is because the $\star_\TT$ star product is only well-defined for regular functionals. However, the version of the statement in the hybrid $H$ presentation carries over to local $V$'s in the form
    \[
        \widetilde{Q}^{R/A}_{J,H}\TT \eihbV=\left(Q_0 + \frac{i}{\hbar}(\cdot) \star_H R_{V,H}(J[d\eta])\right) \TT\eihbV  = 0
    \]  
    We leave the generalisation of the proof to the reader.}
\end{remark}

\begin{proof}
    From the definition of $\mathcal{R}_A(B)$ (Definition \ref{def:moller}) we have
    \[
    \mathcal{S}(A)\bullet\mathcal{R}_A(B)=\mathcal{S}(A)\circT B.
    \]
    Then, applying this to $A=\mathfrak{Q}_\TT(V)$ and $B=\mathcal{J}=\frac{i}{\hbar}\mathfrak{Q}_\TT(J[d\eta])$ the modified quantum master equation reads
    \begin{align*}
    \mathcal{Q}_0\mathcal{S}(\mathfrak{Q}_\TT(V)) &= - \mathcal{S}(\mathfrak{Q}_\TT(V)) \circT  \mathcal{J} \qquad \qquad \qquad[\mathsf{mQME}]\\
    &= -\mathcal{S}(\mathfrak{Q}_\TT(V))\bullet \mathcal{R}_{\mathfrak{Q}_\TT(V)}(\mathcal{J}) \\
    &= -\frac{i}{\hbar}\mathcal{S}(\mathfrak{Q}_\TT(V))\bullet \mathcal{R}_{\mathfrak{Q}_\TT(V)}(\mathfrak{Q}_\TT(J[d\eta]))
    \end{align*}
    and similarly for $\mathcal{A}$. The claim is then obtained by translating the general statement into the $\TT$ and $H$ presentations, namely:
\[
Q_{0,H}S_H(V)= - \frac{i}{\hbar}S_H(V)\star_H R_{V,H}(J[d\eta])
\]
in the hybrid $H$-presentation and 
\[
Q_{0,\TT}S_\TT(V)= - \frac{i}{\hbar}S_\TT(V)\star_\TT R_{V,\TT}(J[d\eta])
\]
in the $\TT$-presentation. On the other hand, we have
\[
Q_{0,H}S_H(V)=Q_0\circ \TT(e^{\frac{i}{\hbar}V})=\TT(Q_0-i\hbar \Lap)e^{\frac{i}{\hbar}V}
\]
and
\[
Q_{0,\TT}S_\TT(V)=(Q_0-i\hbar \Lap)e^{\frac{i}{\hbar}V}.
\]
\end{proof}

\begin{theorem}\label{thm:R/AqBV}
    Assume that $L=L_0 + V$ satisfies the mQME  for some generalised current ${J}[d\eta]$ (Definition \eqref{def:mQME} and Lemma \ref{lem:reftofreecase}). We have
    \[
    \begin{cases}
        Q^R_V(F) \eihbV = Q_{0,\TT}(\eihbV F) + \frac{i}{\hbar}(J[d\eta]{\star_V} F)\eihbV\\
        Q^A_V(F) \eihbV = Q_{0,\TT}(\eihbV F) + \frac{i}{\hbar}(F{\star_V} J[d\eta])\eihbV
    \end{cases}
    \]
    Then the retarded/advanced, quantum interacting BV operators, for any local functional $F\in \Fcal_\loc$, read
    \begin{subequations}\label{e:mQMEBVops}\begin{align}
    Q^R_V (F) &\stackrel{\text{mQME}}{=}  \left(Q_0 + \{V,\cdot\} -i\hbar \Lap_V^\ren + \frac{i}{\hbar}\overrightarrow{\Omega}_J \right) F, &\quad \\
    Q^A_V (F) &\stackrel{\text{mQME}}{=}  \left(Q_0 + \{V,\cdot\} -i\hbar \Lap_V^\ren + \frac{i}{\hbar}\overleftarrow{\Omega}_J \right) F, &\quad \, \
    \end{align}\end{subequations}
    where 
    \begin{align*}
    \overrightarrow{\Omega}_J &
        \doteq  {J}[d\eta] {{{\star_V}}}(\cdot)  - {J}[d\eta]\, \cdot\ ;\\
    \overleftarrow{\Omega}_J &
        \doteq  (\cdot){{{\star_V}}} {J}[d\eta] - {J}[d\eta]\cdot\,\ ; 
    \end{align*}
    are called \emph{retarded and advanced quantum BFV corrections}.\footnote{Cf.\ Section \ref{sec:CMR} for an explanation of the notation and terminology.} The operators $Q^{R/A}_V$ in Equations \eqref{e:mQMEBVops} are nilpotent, and we have
    \begin{align*}
    Q^R_V(F) - Q^A_V(F) &= \frac{i}{\hbar}\left(\overrightarrow{\Omega}_J(F) - \overleftarrow{\Omega}_J(F)\right) 
    = \frac{i}{\hbar}[{J}[d\eta],F]_{{{\star_V}}}.
    \end{align*}
    
    Finally, if $L$ and $L_0$ satisfy conditions \eqref{e:mQMEconditions}, we have that 
    \[
    {J}[d\eta(f)] = J_V^{\std}[fdf].
    \]
    where $J_V^{\std}$ is the adjusted interacting Noether current (Equation \eqref{e:adjustedcurrent}).
\end{theorem}
\begin{proof}
This is a consequence of Theorem \ref{thm:quantumBVopnoboundary:ren}, according to which
\[
Q^R_V(F) = e^{-\frac{i}{\hbar}V}Q_{0,\TT}(F\eihbV) - \frac{i}{\hbar}\mathsf{QME}(V){\star_V} F
\]
(similarly for the advanced version) and Lemma \ref{lem:reftofreecase} which states that mQME is equivalent to $\mathsf{QME}(V) = -J[d\eta]$. Looking at the proof of Theorem \ref{thm:quantumBVopnoboundary:ren} one clearly sees that
\begin{align*}
    Q^R_V (F) &=  \left(Q_0 + \{V,\cdot\} -i\hbar \Lap_V^\ren + \frac{i}{\hbar}\left(\mathsf{QME}\cdot(\cdot) - \mathsf{QME}{{{\star_V}}}(\cdot) \right)\right) F, &\quad \\
    Q^A_V (F) &=  \left(Q_0 + \{V,\cdot\} -i\hbar \Lap_V^\ren + \frac{i}{\hbar}\left(\mathsf{QME}\cdot(\cdot) - (\cdot){{{\star_V}}}\mathsf{QME} \right)\right) F, &\quad \,.
    \end{align*}
    Then, we get the result by means of Lemma \ref{lem:reftofreecase}. These operators are nilpotent, given mQME, because $Q^{R/A}_V$ are nilpotent:
    \[
    (Q^R_V)^2(F) = R_V^{-1}\circ Q_0^2 \circ R_V(F) \equiv 0,
    \]
    and similarly for $Q^A_V$.
    
    It is now easy to gather that
    \[
    (Q^R_V-Q^A_V)(F) = \frac{i}{\hbar} {J}[d\eta] {{{\star_V}}} (F) - \frac{i}{\hbar}(F) {{{\star_V}}} {J}[d\eta] = \frac{i}{\hbar}[{J}[d\eta],F]_{{{\star_V}}},
    \]
    and the last claim is a consequence of Proposition \ref{prop:mQMEconditions}.
\end{proof}

    Observe that the definition of (retarded/advanced) quantum interacting BV operators is the same we gave in Definition \ref{df: BVint}. However, in the presence of boundaries the QME is modified. Hence, the quantum BV operators also change as we have seen above. Notice that, when the theory is cast on a bounded region, there is a natural distinction between $Q^R_V$ and its advanced counterpart. This explains the emergence of a nonzero difference of retarded and advanced BV operators. 
    
    Note, indeed, that Theorem \ref{thm:R/AqBV} is the BV-BFV version of Theorem \ref{thm:quantumBVopnoboundary:ren}: the latter is obtained from the former by simply setting $J_V[d\eta]=0$, which is true whenever we are not interested in bounded regions, and/or when the interacting Noether current vanishes (e.g. for the scalar field).

\begin{remark}\label{rmk:RVMultilocExtension2}
    The operators $Q^{R/A}_V$ are defined only on regular functionals. In order to extend their action to multilocal functionals we note that
    \[
    Q_0 \circ R_V(F) = \frac{i}{\hbar} R_V(J[d\eta]) \star_H R_V(F) + R_V\left(Q_0F+\{V,F\}-i\hbar \Lap_V^{\mathrm{ren}} F - \frac{i}{\hbar} {J}[d\eta] \cdot F\,\right)
    \]
    gives us access to an implicit definition of $Q^R_V$ via $R_V\circ Q^R_V\doteq Q_0 \circ R_V(F)$. Here we need to use the full $L_\infty$ structure emerging from the Anomalous Master Ward Identities (Theorem \ref{thm:AMWI} and subsequent). Compare this with Section \ref{sec:curvedAMWI}.
\end{remark}

The comparison between the various BV operators is given by the following result, generalising Equation \eqref{e:fund-adjnoboundary} in Theorem \ref{thm:R/AqBV}:

\begin{proposition}\label{prop:BVopstatevsobs}
We have the following, for every $F\in \Fcal_\loc$:
    \[
    \widetilde{Q}^{R/A}_{J,H}\TT(\eihbV F) = \TT(Q^{A/R}_V(F) \eihbV).
    \]
\end{proposition}

\begin{proof}
    Note that, using the definitions of the M\o ller maps:
    \begin{align*}
    (\TT (\eihbV F))\star_H R_{V,H}(J) &= (\TT \eihbV)\star_H R_{V,H}(F) \star_H R_{V,H}(J) \\
    &= (\TT\eihbV)\star_H R_{V,H}(F{\star_V}J)\\
    &=\TT(\eihbV(F{\star_V} J))
\end{align*}
    Then, we have
\begin{multline*}
        \left(Q_0 + \frac{i}{\hbar}(\cdot) \star_H R_{V,H}(J)\right)\TT(\eihbV F)  \\
        =\TT\left[\left( Q_0F + \{V,F\} - i\hbar \Lap_V^\ren F + \frac{i}{\hbar}J\cdot F\right) \eihbV\right] + \frac{i}{\hbar}(\TT (\eihbV F))\star_H R_{V,H}(J) \\
        =\TT\left[\left( Q_0F + \{V,F\} - i\hbar \Lap_V^\ren F + \frac{i}{\hbar}J\cdot F + \frac{i}{\hbar}F{\star_V} J\right) \eihbV\right] 
        =\TT(Q^A_V \eihbV).
\end{multline*}

\end{proof}

\subsection{Modified Anomalous Master Ward Identity}\label{sec:modAMWI}
{Let us observe that the Modified Quantum Master Equation for a (regular) functional $V$ can be phrased as 
\[
\mathsf{QME}(V) = (Q_0 - i\hbar \Lap)V + \frac12\{V,V\} \not=0,
\]
and in usual situations, i.e.\ when Equations \eqref{e:mQMEconditions} hold, we get (Proposition \ref{prop:mQMEconditions})
\[
\mathsf{QME}(V) = -\frac{i}{\hbar}J_V^\std.
\]
More generally we can make two adjustments to the formulas above.

First we correct the QME by the anomaly term
\[
\mathsf{QME}(V) = Q_0(V) + \frac12\{V,V\} -i\hbar A(V)
\]
where $A$ is the anomaly coming from the Master Ward Identity (Theorem \ref{thm:AMWI}). In good cases this satisfies a modified version of the extended Wess--Zumino consistency condition (\cite[Proposition 6.2]{BDFR_infty}):
\begin{proposition}
    Let $L=L_{00} + V$ be a generalised Lagrangian in the minimal splitting. Assume that\footnote{Observe that here we are implicitly assuming the QME for $L_0$, i.e. $\frac12\{L_0,L_0\}^\std - i\hbar\Lap_0^\ren L_0=0$. Recall also that in the minimal splitting $J=J_I$ (Lemma \ref{lem:vanishingminimalcurrent}).} 
    \[
    \mathsf{QME}(V)=\frac12\{L_{00} + V,L_{00} + V\}^\std - i\hbar A(V) = J
    \]
    for some generalised current $J$. Then we have 
    \[
    \Lap_V^\ren J =0, \qquad \{L_{00}+V , A(V)\}^\std = 0.
    \]
\end{proposition}

\begin{proof}
    To to this, we follow the steps of \cite[Proposition 6.2]{BDFR_infty}, which---owing to the (higher) Jacobi identity of the higher brackets \eqref{e:higherbrackets}---lead us to
    \[
    0=\{L_{00}+V , A(V)\}^\std + \langle A'(V), \{L_{00}+V, L_{00}+V\}^\std - i\hbar A(V)\rangle.
    \]
    The result is proven iff the second term vanishes, which by assumption is equivalent to the condition
    \[
    \langle A'(V), J\rangle \equiv \Lap_V^\ren J=0. 
    \]
    Combining \cite[Equation 6.36 and Equation 6.31 (for $\chi=0$)]{BDFR_infty} with $QJ=0$ (Proposition \ref{prop:descent}) we get
    \[
    \langle A'(V), J\rangle = -J \widehat{Q}_V(1) + \widehat{Q}_V(J) = - JJ + e^{-\frac{i}{\hbar} V} \TT^{-1} Q_{00} \TT(\eihbV J).
    \]
    The first term vanishes for degree reasons ($J$ is odd). Now we use a result of Hollands \cite[Section 4.6]{Hollands}, according to which $\TT(\eihbV J) = 0$ modulo the equations of motion, which we rewrite as $\TT(\eihbV J) = Q_{00}(X)$ for some $X$. The result follows by nilpotency of $Q_{00}$, which is the Koszul differential (Proposition \ref{prop:quad+minimalsplit}).
\end{proof}

To obtain examples where the conditions above are satisfied, one typically  starts with the minimal quadratic generalised Lagrangian $L_{00}$ that satisfies the modified classical master equation and $A(V)$ is the anomaly term that arises through renormalisation \cite{FR3} (see also \cite{BDFR_infty}) for a given---typically local---interaction $V$. This anomaly term satisfies the modified extended Wess-Zumino consistency condition above.

When the cohomology of $Q$ in degree 1 is trivial, ${A(V)}$ is also $Q$-exact (up to boundary terms). One can then add a term of order $\hbar$ to $L=L_0+V$ and potentially a term to $J$ and require that $\mathsf{QME}(V^\hbar) = \frac{i}{\hbar}J^\hbar$ with $J^\hbar \sim 0$ a power series with values in currents. One proceeds inductively to obtain a new $L$ and a new $J$, where the anomaly is absent. This is, implicitly, also the reason why all of our results are stated \emph{w.r.t.\ some generalised Noether current} for which $L$ satisfies the modified quantum master equation. 

Alternatively, one can also absorb the modification of $L$ and $J$ into the modification of the time-ordering prescription $\TT$, as discussed in \cite{FR3}. The concrete example where such procedure is performed for Yang-Mills theory can be found in \cite[Sections 4.6, 4.7]{Hollands}. In this case one talks about renormalised Noether current.  

We can interpret the term $\mathsf{QME}(V)$ as a curvature term, as it vanishes when the auxiliary interacting BV operator annihilates the unit: $\widehat{Q}_{V}(1)=0$ (Remark \ref{prop:auxopBVBFV}). Indeed, from the AMWI tower of multilinear operations in Equations \ref{e:higherbrackets} we get:
\begin{theorem}\label{thm:AMWI_curved}
Let $V$ satisfy the modified Quantum Master Equation with generalised Noether current $J^\hbar$. Then, the following data define a curved $L_\infty$ algebra
    \begin{subequations}\label{e:MODhigherbrackets}\begin{align}
    &[-]_0^V=J\label{e:Mod[]0}\\
    &[G]_1^V=Q_0(F) + \{V,F\}^\std + \langle A',G\rangle\label{e:Mod[]1}\\
    &[G_1,G_2]_2^V=\{G_1,G_2\}^\std + \langle A''(V),G_1\otimes G_2\rangle\label{e:Mod[]2}\\
    &[G_1,\ldots,G_n]_n^V=\langle A^{(n)}(F),G_1\otimes\cdots\otimes G_n\rangle,
    \end{align}\end{subequations}
\end{theorem}
\begin{proof}
    This is a direct consequence of \ref{thm:AMWI-Linfty}, when $\mathsf{QME}(V)=J$.
\end{proof}
}

\section{BV-BFV on causal cylinders}\label{sec: Cauchy burger}

Let $M$ be a globally hyperbolic spacetime. In this section we consider a ``slice'' of $M$ contained between two Cauchy surfaces $\Sigma_-$ and $\Sigma_+$, such that $\Sigma_-$ is in the past of $\Sigma_+$. In order to apply our version of the BV-BFV framework to this slice spacetime we need to ``fatten'' these Cauchy surfaces by embedding them in open neighborhoods $N_-$ and $N_+$, such that the closure of $N_-$ does not intersect the closure of $N_+$. We think of bulk observables as functionals localised in the open region $\Ocal$ between the closure of $N_-$ and $N_+$. The setup is shown on figure \ref{fig:causalcylinder}. From a more fundamental point of view, one wishes to access a description of field theory that is functorial, from an appropriate category of (co)bordisms. The reader should compare this with \cite{BmMS}, where a the categorical nature of ``Lorenzian bordisms'' is studied. We plan on tying the cited approach to ours in a subsequent work, where gluing is studied. 

\subsection{Definitions of relevant operators}\label{sec:burgerdefs}

\begin{figure}[ht]
\centering
\ifpdf
  \setlength{\unitlength}{1bp}%
  \begin{picture}(200, 238.69)(0,0)
  \put(0,0){\includegraphics{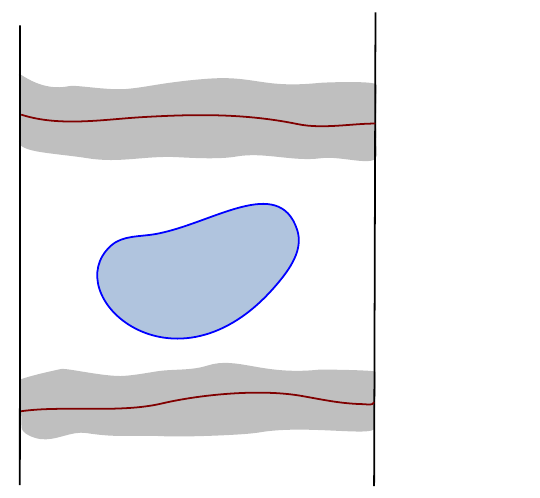}}
  \put(80.69,107.53){\fontsize{14.23}{17.07}\selectfont \textcolor[rgb]{0, 0, 1}{$\mathcal{O}$}}
  \put(18.78,178.07){\fontsize{14.23}{17.07}\selectfont \textcolor[rgb]{0, 0, 0}{$N_+$}}
  \put(22.47,39.45){\fontsize{14.23}{17.07}\selectfont \textcolor[rgb]{0, 0, 0}{$N_-$}}
  \put(183.57,174.68){\fontsize{14.23}{17.07}\selectfont \textcolor[rgb]{0.50196, 0, 0}{$\Sigma_+$}}
  \put(182.03,42.53){\fontsize{14.23}{17.07}\selectfont \textcolor[rgb]{0.50196, 0, 0}{$\Sigma_-$}}
  \end{picture}%
\else
  \setlength{\unitlength}{1bp}%
  \begin{picture}(200, 238.69)(0,0)
  \put(0,0){\includegraphics{burger}}
  \put(80.69,107.53){\fontsize{14.23}{17.07}\selectfont \textcolor[rgb]{0, 0, 1}{$\mathcal{O}$}}
  \put(18.78,178.07){\fontsize{14.23}{17.07}\selectfont \textcolor[rgb]{0, 0, 0}{$N_+$}}
  \put(22.47,39.45){\fontsize{14.23}{17.07}\selectfont \textcolor[rgb]{0, 0, 0}{$N_-$}}
  \put(183.57,174.68){\fontsize{14.23}{17.07}\selectfont \textcolor[rgb]{0.50196, 0, 0}{$\Sigma_+$}}
  \put(182.03,42.53){\fontsize{14.23}{17.07}\selectfont \textcolor[rgb]{0.50196, 0, 0}{$\Sigma_-$}}
  \end{picture}%
\fi
   \caption{Causal cylinder}
    \label{fig:causalcylinder}
\end{figure}

In the special setting we are considering here one can split ${J}[d\eta]$ into 
\[
{J}[d\eta]={J}[d\eta]^++{J}[d\eta]^-\,,
\]
where ${J}[d\eta]^+$ is supported in $N_+$ and ${J}[d\eta]^-$ is supported in $N_-$. This implies that the retarded/advanced quantum interacting BV operators simplify. 
\begin{lemma}\label{lem:ret/advBVopCylinder}
    With the notation of the previous section, assuming the mQME with current $J$, for $\supp(F)\subset \Ocal$ (see figure \ref{fig:causalcylinder}) we have that
    \[
    \overrightarrow{\Omega}_J F = [{J}[d\eta]^-,F]_{{\star_V}}\qquad \overleftarrow{\Omega}_J F = [F,{J}[d\eta]^+]_{{\star_V}}\,.
    \]
    Hence, the retarded/advanced adjoint quantum BV operators (cf. Definition \ref{df: BVint on functionals} and Equation \eqref{e:mQMEBVops}) read
\be\label{e:QVhat burger}
  \begin{cases}
      Q^R_V(F) =  \left(\frac{i}{\hbar} [{J}[d\eta]^-,\cdot ]_{{\star_V}} + Q_0 + \{V,\cdot\} -i\hbar \Lap_V^\ren\right) F,\\
      Q^A_V(F) =  \left(\frac{i}{\hbar} [\cdot,{J}[d\eta]^+]_{{\star_V}} + Q_0 + \{V,\cdot\} -i\hbar \Lap_V^\ren\right) F.
  \end{cases}
\ee
\end{lemma}
\begin{proof}
    We have $\overrightarrow{\Omega}_J F = {J}[d\eta] {\star_V} F - {J}[d\eta] \cdot F$ and split
    \begin{subequations}\begin{align}
    {J}[d\eta] {\star_V} F&=({J}[d\eta]^++{J}[d\eta]^-){\star_V} F\,,\label{e:starVsplit}\\
    {J}[d\eta] \cdot F&=({J}[d\eta]^++{J}[d\eta]^-)\cdot F\,. \label{e:dotsplit}
    \end{align}\end{subequations}
    Now we use the fact that if $\supp F$ does not intersect $\mathcal{J}^-(\supp G)$ then
    \[
    F{\star_V} G = F\cdot G\,. 
    \]
    This fact is proven, for example in Lemma~2.1 of \cite{DHP}, but with slightly different convention in the definition of M{\o}ller maps, i.e. they use $R_V\circ \TT^{-1}$ instead of $R_V$. This means that they prove the relation in the $H$-presentation and thus the relation in \cite{DHP} is between ${\star_V}$ and $\T$ rather than ${\star_V}$ and $\cdot$, as in our case.  Similarly, if $\supp F$ does not intersect $\mathcal{J}^+(\supp G)$, we have
    \[
    G{\star_V} F = F\cdot G\,. 
    \]
    Using these in \eqref{e:dotsplit}, we obtain
    \[
    {J}[d\eta] \cdot F= {J}[d\eta]^+{\star_V} F +F{\star_V} {J}[d\eta]^-\,.
    \]
    Inserting this together with \eqref{e:starVsplit} into the formula for the quantum BFV corrections $\overrightarrow{\Omega}_J$ gives
    \[
    \overrightarrow{\Omega}_J F=({J}[d\eta]^++{J}[d\eta]^-){\star_V} F-({J}[d\eta]^+{\star_V} F +F{\star_V} {J}[d\eta]^-)=[{J}[d\eta]^-,F]_{{\star_V}}\,.
    \]
    A completely analogous calculation allows to show us the result for $\overleftarrow{\Omega}_J$.
\end{proof}

\subsection{Comparision to local BV-BFV approach}\label{sec:CMR}

In a series of papers, Cattaneo, Mnev and Reshetikhin (henceforth CMR) used that the master equations that are center stage in the theory of quantisation of gauge theories in the BV formalism are naturally modified by boundary terms (see \cite{CMR1,CMR2} and subsequent descendants), to formulate a prescription for quantising gauge theories on manifolds with boundary. In the finite-dimensional BV formalism one can conveniently think of the quantum master equation also formulated as
\[
\Lap e^{\frac{i}{\hbar} L} = 0, \quad \iff \quad \left(\frac12\{L,L\} - i\hbar \Lap L\right)e^{\frac{i}{\hbar}L} = 0.
\]
From this point of view, this is fundamental condition in order for gauge fixing to make sense, where $L$ is a generalised Lagrangian on a (compact) manifold, possibly with boundary, and $\Lap$ denotes the BV Laplacian.\footnote{The BV theorem in finite dimensions asserts that the Quantum Master Equation is needed for the gauge-fixed integral to be independent on a continuous family of gauge fixing conditions, which are choices of Lagrangian submanifolds of the BV space of fields.}

Note that, in the equation above, $\Lap$ is the BV Laplacian on densities and $e^{\frac{i}{\hbar}L}$ should be thought of as multiplied by a reference density $\mu_0$, so truly we mean $\Lap (e^{\frac{i}{\hbar} L}\mu_0)$ in the first equation. This is related to the BV Laplacian on functions, appearing as $\Lap L$ in the second equation, but for simplicity we will conflate the two notions and use the same symbol for both (see \cite{CattaneoMnevSchiavina} for further details).

When $L\equiv L^\hbar$ is a formal power series in $\hbar$ one can attempt to find a solution of QME starting from the zeroth-order condition $\{L,L\}=0$, the strong version of the classical master equation. 

However, in field theory, without a specific choice of falloff or boundary conditions, the classical master equation is obstructed by a boundary term: $\frac12\{L,L\}=\pi^*L^\partial$, where $\pi\colon \Ecal \to \Ecal^\partial$ maps the space of bulk fields $\Ecal$ to some appropriate space of boundary fields $\Ecal^\partial$, and $L^\de$ is a local functional (generalised lagrangian) of degree $1$ thereon. Thus, the Quantum Master Equation must take into account the ``boundary failure'' of CME.

In \cite{CMR2}, this is addressed by requiring that the boundary term $\pi^*L^\partial e^{\frac{i}{\hbar}L}$ appearing in the QME be compensated by the action of some operator, let us call it $\widehat{\Omega}_{CMR}$, in the sense that the quantum master equation gets modified as 
\[
(\hbar^2\Lap  + \widehat{\Omega}_{CMR})e^{\frac{i}{\hbar}L} = 0.
\]
On ``observables'', one defines a ``modified'' quantum BV operator via
\be\label{e:modqBVopCMR}
\widehat{Q}(F) \doteq (Q - i\hbar \Lap + \widehat{\Omega}_{CMR})(F).
\ee

In order to find $\widehat{\Omega}_{CMR}$, the ansatz proposed by CMR is that it should be given by some appropriate quantisation of $L^\partial$, for example by looking at the (geometric) quantisation of the space of boundary fields $\Ecal^\partial$.\footnote{Note, in passing, that this requires a choice of polarisation.} This is a version of BFV quantisation for constrained Hamiltonian systems (see \cite{FradkinLinetsky}), and is in principle aimed at constructing a space of physical states as the cohomology of $\widehat{\Omega}_{CMR}$. (For our version of the definition of a space of states see Section \ref{sec:Hollands}.)

Let us compare the definition of the modified quantum BV operator considered in CMR (equation \eqref{e:modqBVopCMR}) with the retarded/advanced, quantum BV operator in the $\TT$-presentation characterised in Theorem \ref{thm:R/AqBV} for a generalised Lagrangian in a free-interacting splitting $L=L_0 + V$ that satisfies the quantum master equation with generalised current ${J}$:
\[
Q^{R/A}_V (F) =  \left(Q_0 + \{V,\cdot\} -i\hbar \Lap + \frac{i}{\hbar}\overleftrightarrow{\Omega}_J \right) F.
\]
We can see that the operators $\overleftrightarrow{\Omega}_J$, which we called retarded/advanced quantum BFV corrections, play the role of the modification to the standard quantum BV operator ($Q_{0,\TT} = Q_0 -i \hbar\Lap$ in the $\TT$-presentation), up to the fact that $Q = Q_0 + \{V,\cdot\}^\std + B_V(\cdot)$, where $B_F(G) = X_F(G) - \{F,G\}^\std$.

Our version of the modified quantum master equation $(Q_{0,\TT} +{J}) \eihbV = 0$ can be rewritten equivalently, using the definition of modified retarded/advanced BV operators in the $\TT$-presentation $\widetilde{Q}^{R/A}_{J,\TT}\doteq Q_0(F) + \frac{i}{\hbar} {\mathcal{J}}^{R/A}_{\TT}(F)$ (Equations \ref{e:RAmodqBV}), as
\be\label{e:mQMEreprise}
\widetilde{Q}^{R/A}_{J,\TT} \eihbV =0 
\ee

We can link the R/A quantum BV operators $Q^{R/A}_V$ and the modified R/A BV operators in the $\TT$-presentation $\widetilde{Q}_{J,J}^{R/A}$ by means of Proposition \ref{prop:BVopstatevsobs}, which states, together with Theorem \ref{thm:R/AqBV}, that
\begin{align}\notag
(\widetilde{Q}^{R/A}_{J,\TT})\TT(F \eihbV) 
    &= \TT(Q^{A/R}_{V,\TT} (F) \eihbV) \\\label{e:CMRcomp}
    &=  \TT\left(\left(Q_0 + \{V,\cdot\} -i\hbar \Lap + \frac{i}{\hbar}\overleftrightarrow{\Omega}_J \right) (F) \eihbV \right).
\end{align}
In order to compare this result with CMR we note that, for any observable $F$:
\begin{equation*}
    (\hbar^2\Lap  + \widehat{\Omega}_{CMR})(Fe^{\frac{i}{\hbar}L}) = (Q F - i\hbar \Lap F) e^{\frac{i}{\hbar}L} + \widehat{\Omega}_{CMR}(Fe^{\frac{i}{\hbar}L}) - F \widehat{\Omega}_{CMR}e^{\frac{i}{\hbar}L}
\end{equation*}
where we used the CMR version of the mQME. Now, one gets that the operator $(Q - i\hbar \Lap) F$ is corrected by $\widehat{\Omega}_{CMR}$ whenever the latter is a derivation of the dot product. This is not always the case (recall that $\Lap$ itself is not a derivation), and the issue is resolved in CMR, for example, by assuming that $L^\partial$ is linear in fibre coordinates.\footnote{Fibres of the map $\Ecal\to\mathcal{B}$ where $\mathcal{B}$ is the space of leaves of the chosen polarisation of the space of boundary fields $\Ecal^\partial$.} In that case, however, one has
\[
(\hbar^2 \Lap + \widehat{\Omega}_{CMR})(Fe^{\frac{i}{\hbar}L}) = (Q - i\hbar \Lap + \widehat{\Omega}_{CMR})(F) e^{\frac{i}{\hbar}L}
\]
which would suggest, comparing with Equation \eqref{e:CMRcomp} that the following two operators should be related
\[
\overleftrightarrow{\Omega}_J \longleftrightarrow \widehat{\Omega}_{CMR}.
\]
From this point of view, our answer is more general.

If we specialise to the example of a causal cylinder\footnote{For more general situations one needs to adapt our results to field theory on Riemannian manifolds, which is something we plan to do in subsequent work.} (see Section \ref{sec:burgerdefs}) we have that
\begin{gather*}
    \overrightarrow{\Omega}_J F = [J_V[d\eta]^-,F]_{{\star_V}}\qquad \overleftarrow{\Omega}_J F = [F,J_V[d\eta]^+]_{{\star_V}}\,\\
    Q^{R-A}_{V,\TT}\doteq Q^R_V - Q^A_V = [J_V[d\eta],F]_{{\star_V}}.
\end{gather*}

To explicitly work out the relation to the CMR approach, we analyse the star product in more detail. Introducing a parameter $\lambda$ to play the role of the coupling constant (recall that in the previous sections we have assumed that $V(f)= \lambda\int \mathbf{V}f$), we note that, at the lowest order in $\lambda$, we have
\[
  Q^{R-A}_{\lambda V,\TT}F= [J_{\lambda V}[d\eta],F]_{{\star_V}}+\Ocal(\lambda^2)\,.
\]

In constructing a star product, often one begins by building the \textit{Wick-ordered star product} (see Definition \ref{def:wickstarprod}), which is given by $\Delta^+=\frac{i}{2}\Delta + H$, the 2-point function of some Hadamard state (cf.\ Definition \ref{def:propagators}). 
Geometrically, the building blocks of $\Delta^+$ produce a K{\"a}hler structure\footnote{When there are no gauge symmetries, or when we perform gauge-fixing,\ i.e. on gauge-fixed solutions of the EL equations. See \cite[Remark 4.11]{RielloSchiavinaPS} for a discussion of the presymplectic/Dirac structure on $EL(M)$.} on $EL(M)$, the space of solutions to field equations, namely $H$ is a symmetric non-degenerate  bilinear form on $EL(M)$, $\Delta$ is a symplectic structure and $\Delta^+$ is hermitean on $EL(M)^{\mathbb{C}}$ (complexification of $EL(M)$) and if the Hadamard state used for the construction of $\Delta^+$ is pure, then we can define an anti-involution $\mathfrak{J}$ such that $G^C\circ \mathfrak{J}= 2H$ \cite[Section~5.3]{Book}. The star product $\star_H$ is then readily interpreted as the one where one argument is differentiated in the holomorphic direction and the other argument only in the anti-holomorphic direction. This choice of star product corresponds to the holomorphic (K\"ahler) polarisation. 

It is possible to choose a star product that reflects the `position' or `momentum' polarisations (of phase space) instead. Formally, the corresponding products are related by an automorphism. (In the deformation quantisation literature this is referred to as a ``gauge transformation'', which is admittedly confusing terminology and not to be mistaken with the local Lie algebra action on fields.) Restricting our attention to regular functionals, we define the map 
\[
\alpha_H\doteq e^{\left<H,\frac{\delta^2}{\delta\ph^2}\right>}\,.
\]
For $F,G\in \Fcal_{\reg}$, we have (cf. this with Definition \ref{def:wickstarprod} and the subsequent discussion)
\[
F\star_H G= \alpha_H(\alpha_H^{-1}F\star_0 \alpha_H^{-1} G)\,,
\]
where $\star_0$ is the \textit{Moyal--Weyl-ordered} product defined by
\[
F\star_0 G= m\circ e^{\left<\frac{i\hbar}{2}\Delta,\frac{\delta}{\delta \ph}\otimes \frac{\delta}{\delta \ph}\right>} F\otimes G
\]
This product is related to the quantisation of the reduced phase space of the theory, as we will now outline.

On a globally hyperbolic manifold, the phase space can be described in two ways: either as the space of solutions of the Euler--Lagrange equations of motion (for fields in $M$) modulo gauge transformations, or in terms of Cauchy data, i.e.\ gauge-equivalence classes of field configurations that solve certain `constraints' on a Cauchy surface.\footnote{If the Cauchy surface is closed without boundary, this is the coisotropic reduction by its characteristic distribution, and the result is symplectic. If the Cauchy surface has boundary, the procedure is more involved and it outputs a Poisson manifold in a large class of examples. See \cite{RielloSchiavina,CanepaCattaneoPinfty,RielloSchiavinaPS}.} 

For the free scalar field Cauchy data and initial value data coincide, since no constraints (and gauge degeneracies) are present, and are given by $\mathcal{C}(\Sigma)\equiv \Ecal^\partial=C^{\infty}_0(\Sigma)\times C^{\infty}_0(\Sigma)$. We can now build a solution map $\beta:\mathcal{C}(\Sigma)\rightarrow \mathcal{EL}(M)$ from the Cauchy/initial data $(\varphi,\psi)$ to solutions of the field equations in $M$. For the free scalar field this map is explicitly given in \cite{Book}:
\[
\phi(z)= \beta(\varphi,\psi)(z) = \int_{\Sigma_t} \Delta (z; t,\underline{x}) \psi(\underline{x}) - \partial_0^{(2)}\Delta (z; t,\underline{x})\varphi(\underline{x})d\sigma(\underline{x})
\]
where $\partial_0^{(2)}\Delta (z; t,\underline{x})$ denotes the derivative of $\Delta$ along the zero component of the second variable, evaluated at $(t,\underline{x})$.

This result is generalisable to a large class of physically relevant models, including Yang--Mills(--Higgs) theory (see e.g. \cite{ChruscielShatah,Moncrief,ChoquetBruhat}), in the sense that (up to some energy estimates to preserve field regularity), one can establish a gauge-equivariant isomorphism between the subspace of $\mathcal{EL}(M)$ typically given by solutions of the Euler--Lagrange equations on a globally hyperbolic manifold $M$ such that the time gauge is satisfied, and the (coisotropic) submanifold of constrained initial value configurations $\mathcal{C}(\Sigma)$ associated to a Cauchy surface $\Sigma\hookrightarrow M$. On $\mathcal{EL}(M)$ one typically defines a Poisson structure using Peierl's prescription, while on $\mathcal{C}(\Sigma)$ one inherits a pre-symplectic structure using the Noether--Liouville--Kijowski--Tulczyjew prescription (see \cite{RielloSchiavinaPS}).

To fix the ideas, let us stay on the case of scalar field theory for now. Let $\{.,.\}_{\rm can}$ be the canonical bracket on the space of initial/Cauchy data, induced by the symplectic form \cite{CattaneoMnev-wave}:
\[
\omega_{\rm can} = \int_{\Sigma} \delta \varphi \wedge \delta J_\varphi \mathrm{vol}_\Sigma \in \iloc^{2}(\Ecal^\partial)
\]
where $\varphi,J_{\varphi}\in C^\infty(\Sigma)$, and $J_\varphi$ is interpreted as the transverse jet of $\varphi$ in a tubular neighborhood of $\Sigma\hookrightarrow M$, evaluated at $\Sigma$. Then, we have:
    \be\label{e:scalarPeierlsequivalence}
    \{\beta^*F,\beta^*G\}_{\rm can}=\beta^*\Pei{F}{G}
    \ee
(See \cite[Section 5.1.3 and 5.1.4]{FR-axiomatic} and also \cite{Book} for a proof.)

\begin{remark}
    
The relationship between the Peierls bracket defined by $\Delta^+$ and the Poisson structure associated to the reduced symplectic structure on Cauchy data on $\Sigma$ has been investigated by several authors, and it was proven in some examples. In general, said relationship is---at the moment only conjecturally true in generality---that they are Poisson isomorphic. The statement was proven in the case of the scalar field (where there is no gauge) in \cite{HRV}, and a homological proof of such a statement in the case of linear gauge theories was given in \cite{BeniniMusanteSchenkel}. Other compelling general arguments have been put forward in \cite{RomeroForger}, using the covariant phase space as a factor between the Peierls phase space and the reduced phase space of Cauchy data (see also \cite{RielloMartinoli}). Another version of this statement is being completed in \cite{NewHRV} using the methods developed in \cite{Visser}.

Equation \eqref{e:scalarPeierlsequivalence} can be generalised to gauge theories using the approach of \cite{FR}, where introducing extra fields and using the gauge-fixing fermion we reduce the problem to a system of normally hyperbolic equations. This is necessary, since if we want to compare \emph{symplectic} structures, we need to reduce the space of constrained initial values along gauge directions. 

In other words, one can establish a relationship between the space of gauge-fixed fields, endowed with a ``gauge-fixed Peierls'' Poisson bracket (see also \cite{Khavkine_Peierls}), and a gauge-fixed version of the Poisson bracket on the reduced phase space, inside the space of Cauchy data. Technically speaking, a ``gauge-fixing`` in this sense is a section of the quotient $\pi:C\to \underline{C}$, where $C$ denotes the set of constrained Cauchy data, and $\pi$ is coisotropic reduction. This means that, by a choosing a gauge fixing, we are reproducing the Poisson algebra\footnote{This is because there is an injective map $\pi^*\colon C^\infty(\underline{C}) \to C^\infty(C)$. Note that $C^\infty(\underline{C})\simeq N(\mathcal{I}_C)/\mathcal{I}_C$, where $\mathcal{I}_C$ denotes the vanishing ideal of $C$ and $N(\mathcal{I}_C)$ denotes its normaliser, is a Poisson algebra even if $\underline{C}$ is singular.} $\left(C^\infty(\underline{C}),\{\cdot,\cdot\}_{\rm red}\right)$ within $C^\infty(C)$. 

Since the resulting symplectic structure is hard to describe explicitly (with some notable exceptions \cite{RielloSchiavina,RielloSchiavinanull}), one can present the statement using a cohomological resolution. Indeed, one can replace the gauge-fixed Peierls bracket by a generalisation defined by means of the technology of Green's witnesses, as introduced in \cite{BeniniMusanteSchenkel}, a.k.a. gauge-fixing operators within the BV formalism (see e.g. \cite{costello2011renormalization} as well as \cite{SchiavinaStucker}).

\end{remark}

On regular functionals, we can quantize this bracket using the Moyal prescription to obtain the \emph{Moyal star product},
\[
f\star_{M}g=m\circ e^{\frac{i}{2\hbar}P-P^*}f\otimes g\,,
\]
for $f,g\in \Ci(\Ecal^\partial)$ with $P=\int_{\Sigma} \frac{\delta}{\delta \ph(\underline{x})}\otimes \frac{\delta}{\delta \psi(\underline{x})}$ and $(\ph,\psi)$ are the variables in $\Ecal^\partial$ (analogous to position and momentum). We can then  prove the quantum analog of Equation \eqref{e:scalarPeierlsequivalence}:
\begin{proposition}
For the free scalar field we have, for any $F,G\in\Fcal_{\rm reg}$:
    \[
    \beta^*F\star_{M}\beta^*G=\beta^* (F\star_0G)
    \]
    where $\star_{M}$ denotes the star-quantisation of $\{\cdot,\cdot\}_{\rm can}$.
\end{proposition}
\begin{proof}
This is a direct analog of \cite[Section 5.2]{FR-axiomatic}.
\end{proof}
This product can be easily related to the one used by CMR to quantize the boundary action, which we call \textit{standard-ordered}, following the deformation quantisation literature. We have
\[
f \star_{\rm Std} g=m\circ e^{-\frac{i}{\hbar}P^*}f\otimes g
\]

\begin{proposition}
    Let
    \[
    N\doteq e^{-i\frac{\hbar}{2}\int_{\Sigma}\frac{\delta^2}{\delta p\delta q}}
    \]
    then:
    \[
    N(N^{-1} f \star_{\rm Std}N^{-1} g)=f\star_{M} g\,.
    \]
\end{proposition}
\begin{proof}
  See \cite{Waldmann}.
\end{proof}

So, if $f=\beta^*F$ and $g=\beta^*G$ we have, for $F,G$ regular
\[
N(N^{-1} \beta^*F \star_{\rm Std}N^{-1} \beta^*G)=\beta^*F\star_{MW} \beta^*G = \beta^*(F\star_0 G) = \beta^*\alpha_H^{-1}(\alpha_H F \star_H\alpha_H G)
\]

We conclude that the star product $\star_H$ used in this paper is indeed equivalent to $\star_{\rm Std}$, on the space of regular polynomials on Cauchy data (i.e. initial data modulo constraints and modulo gauge). This illustrates the agreement of our approach with CMR at zero order in the coupling. We will carry out an explicit comparison for abelian Yang--Mills theory in Section \ref{sec:AbelianYMComp}. For higher orders, note that the interacting star product ${\star_V}$ is constructed perturbatively from $\star_H$ and the formula for ${\star_V}$ has a natural expression in terms of Feynman diagrams \cite{HR}. We conjecture that this diagrammatic expansion should agree with the higher $\lambda$ corrections of CMR. 

At this time, however, a more general (and direct) comparison is not immediately possible, owing to the fact that our constructions naturally live in a Lorentzian setting, while the expressions of \cite[Lemma 4.20]{CMR2} are given for topological $BF$ theory, which does not have a natural (or known) formulation on Lorentzian manifolds.

We have however shown that, by unpacking the expression for the ${\star_V}$ product one sees that the quantum BV operator \emph{begins} with some quantisation of ${J}[d\eta]$ at lowest order in $\hbar$, where ${J}$ is the failure of the quantum master equation. (Note that ${J}[d\eta]$ here is thought of as a function on $\Ecal$, and not just of boundary fields, and it is to be interpreted as the ``interaction part'' of the boundary action, smoothened by means of a smearing into the bulk.) Higher corrections will be given by the full star product expansion. 

Note also that choosing a particular presentation of the star product is tantamount to a choice of polarisation, which is needed for the (geometric) quantisation of $\Ecal^\partial$. Finally, we observe that one can take a \emph{sharp-boundary} limit so that, when convergent (Definition \ref{def:BV-BFV})
\[
\lim_{n\to \infty} J^\std[f_ndf_n] = \pi^*L^\partial\,,
\]
where $L^\partial$ is the generalised Lagrangian on $\Ecal^\partial$ corresponding to the boundary action in CMR.\footnote{We argue that this limiting procedure is not going to significantly affect the $\hbar$ expansion proposed earlier, owing to the fact that we are restricting to the (reduced) phase space.} This limiting procedure is implicitly assumed in \cite{CMR2} as well, where the space of fields $\Ecal$ is considered to be a trivial bundle on $\Ecal^\partial$, i.e. $\Ecal = \Ecal^\partial \times \mathcal{Y}$ for some space $\mathcal{Y}$ interpreted as ``bulk'' fields.

In our language, working with smoothened boundaries is more natural, but we expect the results to be equivalent in the appropriate cohomology, also due to the `up-to-homotopy gluing' results of \cite{CattaneoMnevGluing}.

\subsection{Example of Abelian Yang--Mills (YM) theory}\label{sec:AbelianYMComp}
We compute the BV-BFV structure for abelian YM theory for a trivial $U(1)$ bundle $P\to M$, over Minkowski spacetime. This has mostly appeared elsewhere (see \cite{CMR1,MSW,RejznerSchiavina_Asymptotic,IrasoMnev}) but we need the implementation of the nonminimal sector, and to differentiate carefully between standard and source brackets of the BV Lagrangian. We will then proceed with the explicit geometric quantisation of the boundary action, which to our knowledge is new in the 4d case (see \cite[Section 2.4]{IrasoMnev} for the non-Abelian version in 2d).

We denote by $\star$ here the Hodge operator given by Minkowski metric. Our space of fields is then the space of principal connections, which in this case is identified with one-forms on $M$, together with the BV field extensions. We will implement the Lorentz gauge fixing condition by means of Nakanishi--Lautrupp fields.\footnote{These represent a way to perform gauge fixing in what is often called the nonminimal BV sector, i.e. we add auxiliary fields $(b,b^\ddag,\bar{c},\bar{c}^\ddag)$ to write down the gauge fixing Lagrangian submanifold of the space of BV fields as the graph of an exact form.} The space of nonminimal BV fields is then the vector space (we shorthand $\Omega^k(M)\equiv \Omega^k$)
\[
\Ecal=\underbrace{\Omega^1\times \Omega^1[-1]}_{(A,A^\ddagger)} \times \underbrace{\Omega^0[1]\times \Omega^\top[-2]}_{(c,c^\ddagger)} \times \underbrace{\Omega^0 \times \Omega^\top[-1]}_{(b.b^\ddagger)} \times \underbrace{\Omega^0[-1]\times \Omega^\top[-2]}_{(\overline{c},\overline{c}^\ddagger)}.
\]  

On this is defined the BV symplectic density
\[
{\bom}_{YM}=\delta A \star \delta A^\ddagger + \delta c\delta c^\ddagger + \delta b\wedge \delta b^\ddagger + \delta \overline{c}\delta \overline{c}^\ddagger,
\]
and a Lagrangian density
\[
\boldsymbol{L}_{YM}=\frac12 dA\wedge\star dA + \star A^\ddagger\wedge dc - db \wedge \star A + dc\wedge \star d\overline{c} + \left(\frac12 b\star b + b\overline{c}^\ddagger\right),
\]
whose Hamiltonian cohomological vector field reads\footnote{Recall that, for signature $s$ and acting on $k$ forms, we have $\star^2=(-1)^{k(n-k)}s\, \mathrm{id}$ and $\star^{-1}=(-1)^k s\, \star$ in even dimensions.} 
\begin{align*}
Q_{YM}A &= dc,  & Q_{YM}A^\ddagger &= \star d\star dA + db, \\
Q_{YM}c &= 0, & Q_{YM}c^\ddagger &= {-}d\star A^\ddagger + d \star d\overline{c},\\
Q_{YM}b&=0, & Q_{YM}b^\ddagger &= d\star A + \star b  + \overline{c}^\ddag,\\
Q_{YM}\overline{c} &= b, & Q_{YM}\overline{c}^\ddagger &= d\star dc.\\
\end{align*}

We check that the (densitised) classical master equation is satisfied with
\[
\frac12\iota_{Q_{YM}}\iota_{Q_{YM}}{\bom}_{YM} = dc \star^2 d\star dA +dc \star db + b d\star dc = d\left( c d\star dA + b\star dc\right) = d\mathbf{J}_{YM},
\]
which leads to the generalised Noether current
\[
J_{YM}[d\eta] = \int_M \left(c d\star dA + b\star dc\right) d\eta, \quad \eta\in\Omega^0_c(M).
\]
This satisfies
\[
J^\src_{YM}[fdf] = \frac12\{L_{YM}[f],L_{YM}[f]\}^{src} = J_{YM}[df^2] = 2 J_{YM}[fdf],
\]
while the boundary one form current $\theta_{YM}[df]$ such that Equation \ref{e:smaredBVBFVeqt} is satisfied (note the relative sign) $\iota_{Q_{YM}[f]}{\omega}_{YM}=\delta{L}_{YM}[f]-{\theta}_{YM}[df]$ reads
\[
\theta_{YM}[df] = \int_M \bth_{YM} df =\int_M \left( \delta A  \star dA -  (\star A^\ddagger) \delta c - \delta b (\star A) + \delta c  (\star d\overline{c}) + (\star dc) \delta \overline{c} \right)df,
\]
where 
\[
{\bth}_{YM} = \delta A  \star dA  -  (\star A^\ddagger) \delta c - \delta b (\star A) + \delta c  (\star d\overline{c}) + (\star dc) \delta \overline{c}.
\]
Note that 
\[
\iota_{Q_YM}\bth_{YM}= dc\star dA + b\star dc = \mathbf{J}_{YM}.
\]

\begin{remark}\label{rmk:JfullYM}
    Observe that one can compute the standard Noether current by means of 
    \begin{align*}
    \frac12\{L_{YM}[f],L_{YM}[f]\}^\std
        &= \int_x \int_y\left((d \delta(x,y)\star dA - db \star \delta(x,y)) f\right)\int_z(\delta(x,z) dcf)\\
        &+\int_x \int_y(dc\star d\delta(x,y) f)\int_zb \delta(x,z)f\\
        &=\int_x \left(\star dA dc  + \star dc b\right) f df = J_{YM}^\std\left[fdf\right]
    \end{align*}

    Moreover, the contraction of the BV vector field with $\theta_{YM}$ reads
    \[
    \iota_{Q_{YM}}[f]\theta_{YM}[df] = \int fdf (dc\star dA + b \star dc)
    \]
    
    Which shows that $J^{\std}_{YM}$ and $J^\src_{YM}$ satisfy, as expected
    \[
    2J_{YM}^\src[fdf] - \iota_{Q_{YM}[f]}\theta_{YM}[df] = J^\std_{YM}[fdf].
    \]
\end{remark}

The one-form $\bth_{YM}$ induces the symplectic form on fields on $\Sigma$ (we use the same symbols for restricted fields when nonambiguous)
\begin{align*}
\Omega^\partial 
    &=\delta \int_\Sigma \iota^*_\Sigma\left( \delta A  \star dA -  (\star A^\ddagger) \delta c - \delta b (\star A) + \delta c  (\star d\overline{c}) + (\star dc) \delta \overline{c} \right)df\\
    &= \int_\Sigma \left( \delta A  \delta(\star dA)\vert_\Sigma -  \delta(\star A^\ddagger)\vert_\Sigma \delta c + \delta b \delta(\star A)\vert_\Sigma + \delta c  \delta(\star d\overline{c})\vert_\Sigma + \delta(\star dc)\vert_\Sigma \delta \overline{c} \right)df\\
    &= \int_\Sigma \delta A \delta E + \delta A^\dag \delta c + \delta b\delta A_0 + \delta J_c \delta \overline{c}
\end{align*}
where 
\[
\begin{cases}
E= (\star dA)\vert_\Sigma\\
A_0 = (\star A)\vert_\Sigma\\
A^\dag = {-} (\star A^\ddagger)\vert_\Sigma + (\star d\overline{c})\vert_{\Sigma}\\
J_c = (\star dc)\vert_\Sigma
\end{cases}
\]
Note that this determines a map 
\[
\pi^\partial\colon \Ecal\to \Ecal^\partial=\Omega^0(\Sigma) \times C^\infty(\Sigma)\times \Omega^\top(\Sigma)[-1] \times \Omega^\top(\Sigma)[1])
\]
and the boundary action reads
\[
L^\partial = \int_\Sigma cdE + b J_c.
\]

\begin{proposition}\label{prop:BFVquantisationYM}
    Consider the polarization $\mathcal{P}\doteq \left\{\frac{\delta}{\delta A},\frac{\delta}{\delta A^\dag},A_0\frac{\delta}{\delta J_c} + \overline{c}\frac{\delta}{\delta b}\right\}$, and the choice of primitive
    \[
    \alpha_{\text{ver}}^\partial= \int_\Sigma E\delta A + A^\dagger \delta c + \frac12\left( b \delta A_0 -  A_0 \delta b\right) + \frac12\left(J_c\delta \overline{c} + c\delta J_c\right), \quad \delta \alpha_{\text{ver}}^\partial= \Omega^\partial.
    \]
    The space of polarised sections (for the trivial line bundle) is 
    \[
    C^\infty_\mathcal{P}(\Ecal) = \{F\in C^\infty(\Ecal^\partial)\ |\ F = e^{A_0b+\overline{c}J_c} \underline{F}(\overline{c},A_0,A,c)\}.
    \]
    
    Then, the geometric prequantisation of the BFV boundary action, restricts to an operator on polarised sections 
    \[
    \hat{L}^\partial\vert_{C^\infty_\mathcal{P}(\Ecal)} = -i\hbar\int_{\Sigma}\left(cd\frac{\delta}{\delta A} + J_c\frac{\delta}{\delta A_0} + b\frac{\delta}{\delta \overline{c}}\right).
    \]
\end{proposition}
\begin{proof}
To check that $\alpha_{\text{ver}}^\partial$ is compatible with the polarization, i.e. $\forall X\in \mathcal{P}$ we have $\iota_X\alpha_{\text{ver}}^\partial = 0$ is straighforward. It is also easy to verify that
\[
\nabla_X F = 0,\quad \forall X\in \mathcal{P} \iff \begin{cases}
    \frac{\delta F}{\delta E}=0\\
    \frac{\delta F}{\delta A^\dag}=0\\
    \overline{c}\frac{\delta F}{\delta b} + A_0 \frac{\delta F}{\delta J_c} = 0,
\end{cases}
\]
which leads to the condition for $F$ being a polarised section, as claimed.

We compute the prequantisation of $L^\partial$
\[
    \hat{L}^\partial \doteq L^\partial -i\hbar \nabla_{X_{L^\partial}} = L^\partial + i\hbar X_{L^\partial} - \iota_{X_{L^\partial}}\alpha_{\text{ver}}^\partial = L^\partial + i\hbar X_{L^\partial} - L^\partial = i\hbar X_{L^\partial}
\]
where we used that 
\[
\iota_{X_{L^\partial}}\alpha_{\text{ver}}^\partial = \int_{\Sigma}(dc E + J_cb)
\]
as it can easily be checked, since the Hamiltonian vector field reads
\[
X_{L^\partial} = \int_{\Sigma}\left(dc\frac{\delta}{\delta A} + dE \frac{\delta}{\delta A^\dag} + J_c\frac{\delta}{\delta A_0} + b\frac{\delta}{\delta \overline{c}}\right).
\]
Now we have that polarised sections are given by functions of the variables $(A,c,A_0,\overline{c})$, whence the quantisation of the boundary action restricts to 
\be\label{e:quantumboundaryactionYM}
\hat{L}^\partial\vert_{C^\infty_\mathcal{P}(\Ecal)} = -i\hbar \int_{\Sigma} \left(cd\frac{\delta}{\delta A} + J_c\frac{\delta}{\delta A_0} + b\frac{\delta}{\delta \overline{c}}\right)
\ee
\end{proof}

\begin{remark}
    Observe that polarised sections are parametrised by functions of the variables $(A_0,A,c)$ and $\overline{c}$. This last variable, of degree $-1$ accounts for one trivial-cohomology part of the BFV complex,  which embodies de Rham cohomology on a vector space, and it is given by the following component of the quantum BFV operator
    \[
    \hat{L}_{\text{trivial}} = \int_\Sigma b\frac{\delta}{\delta{\overline{c}}}.
    \]
    The other trivial part of the cohomology is given by the term involving $J_c$ and $A_0$, but this has a different interpretation. Indeed, both $J_c$ and $A_0$ are the $\Sigma$-values of - respectively - $\star dc$ and $\star A$. Hence the operator
    \[
    \hat{L}_{BRST_\Sigma} = \int_\Sigma J_c\frac{\delta}{\delta A_0}
    \]
    is a residual piece of the gauge transformation $A\to A + dc$. Namely the restriction $(\delta_{BRST}\star A)\vert_{\Sigma} = (\star dc)\vert_{\Sigma}$.
\end{remark}

We now compare the result above with pAQFT quantisation. Let us compute
\[
    \overrightarrow{\Omega}_{J_V^\std}F = [J_V^\std[d\eta],F]_{{\star_V}}
\]
in the minimal splitting, where the relevant Lagrangian densities are
\[
\boldsymbol{L}_{0}=dA\wedge\star dA  - db \wedge \star A + i dc\wedge \star d\overline{c} + \frac12 b\star b 
\]
and
\[
\boldsymbol{V}= \star A^\ddagger\wedge dc + b\overline{c}^\ddagger\,,
\]
we find that (as a result of Lemma \ref{lem:vanishingminimalcurrent}, i.e. the vanishing of the free generalised Noether current in the minimal splitting)
\[
J^\std_V[d\eta]=J^\std_{YM}[d\eta]=\int_M (c\wedge d\star dA+ b\star dc)d\eta\,,
\]
The interaction term $V$ is in this case just an external source term since antifields are not dynamical variables, so ${\star_V}$ coincides with $\star_\TT$ as external sources don't affect the retarded and advanced Green functions.

To compare with the geometric quantisation approach, we first calculate the term of order $\hbar$ in $\widehat{\Omega}^R_{J_V^\std}F$, which is given by
\[
\Pei{J^\std_V[d\eta]}{F}
\]
The expressions for $\Delta^{R/A}$ can be found for example in \cite{Hollands} leading to the following formula for $\Delta$
\[
\Delta(x,y)  = (k_{IJ}) \otimes
\left(
\begin{matrix}
\Delta_v(x,y) & -id^\star _y \Delta_v(x,y) & 0 & 0\\
-id^\star_x \Delta_v(x,y) & 0 & 0 & 0\\
0 & 0 & 0 & i \Delta_s(x,y)\\
0 & 0 & -i \Delta_s(x,y) & 0
\end{matrix}
\right)\,, 
\]
where $\Delta_{v/s}$ are Pauli-Jordan functions for the Hodge Lapalcian $\Box=dd^\star+d^\star d$ on 1-forms and 0-forms respectively. They are related by the following identites:
\begin{subequations}\label{e:Deltaidentities}\begin{align}
d_x\Delta_s(x,y)&=\Delta_v(x,y)\circ d_y\\
d^{\star}_x\Delta_v(x,y)&=\Delta_s(x,y)\circ d^{\star}_y
\end{align}\end{subequations}
and subjected to the initial conditions:
\begin{subequations}\label{e:DeltaBC}\begin{align}
\Delta_{v/s}(t,\underline{x},t,\underline{y})&=0\\
\partial_0^{(1)}\Delta_{v/s}(t,\underline{x},t,\underline{y})= - \partial_0^{(2)}\Delta_{v/s}(t,\underline{x},t,\underline{y}) &= \delta(\underline{x}-\underline{y}),\\
\partial_0^{(1)}\partial_0^{(2)}\Delta_{v/s}(t,\underline{x},t,\underline{y}) &=0,
\end{align}\end{subequations}
where $\partial_0^{(i)}$ refers to time derivative on the $i$-th variable in a chart adapted to some initial value surface $\Sigma_t$.

The non-vanishing terms in the bracket with $J_V[d\eta]$ are
\begin{multline*}
    \Pei{J_V[d\eta]}{F}
        =\underbrace{\left<\frac{\delta J^\std_V[d\eta]}{\delta A},\Delta^{AA} \frac{\delta F}{\delta A}\right>}_{(a)}
            +\underbrace{\left<\frac{\delta J^\std_V[d\eta]}{\delta A},\Delta^{Ab} \frac{\delta F}{\delta b}\right>}_{(b)}
            +\underbrace{\left<\frac{\delta J^\std_V[d\eta]}{\delta b},\Delta^{bA} \frac{\delta F}{\delta A}\right>}_{(c)}\\
        +\underbrace{\left<\frac{\delta J^\std_V[d\eta]}{\delta c},\Delta^{c\bar{c}} \frac{\delta F}{\delta \bar{c}}\right>}_{(d)}
            +\underbrace{\left<\frac{\delta J^\std_V[d\eta]}{\delta \bar{c}},\Delta^{\bar{c}c} \frac{\delta F}{\delta c}\right>}_{(e)}\,.
\end{multline*}
We now compute these terms separately.
\begin{subequations}
    \begin{align}
    \underbrace{\left<\frac{\delta J^\std_V[d\eta]}{\delta A},\Delta^{AA} \frac{\delta F}{\delta A}\right>}_{(a)}
        &=\int_{M^2} \left( c(x) d_x\star d_x\left(\Delta^{AA}(x,y) \frac{\delta F}{\delta A(y)}\right)\right)d\eta(x)\notag\\
        &+\int_{M^2} \Big( c(x) d_x\star d_x\Big(\Delta^{AA}(x,y) \underbrace{\partial_\mu\big(\frac{\delta F}{\delta \partial_{\mu}A(y)}+\text{higher jets}\big)}_{\doteq d^*_y\alpha_A}\Big)\Big)d\eta(x)\notag\\
        &=\int_{M^2} \left( c(x) d_x \star d_x\left(\Delta_v(x,y) \frac{\delta F}{\delta A(y)}\right)\right)d\eta(x)\notag\\
        &+\int_{M^2} \Big( c(x) \star d_x^\star d_x\Big(\Delta_v(x,y) d^*_y\alpha_A\Big)\Big)d\eta(x)\notag\\
        &=\int_{M^2} \left( c(x) d_x\star d_x\left(\Delta_v(x,y) \frac{\delta F}{\delta A(y)}\right)\right)d\eta(x)
    \end{align}
    where we used the defining property $(dd^\star + d^\star d) \Delta_v = 0$, so that
    \be\label{e:identity}
    d_x^\star d_x\Delta_v(x,y) d^\star_y\alpha_A = - d_x d_x^\star \Delta_v(x,y) d^\star_y\alpha = - d_x \Delta_s(x,y) d^\star _yd^\star_y \alpha=0\,.
    \ee
Now, choosing $\eta_n$ such that $d\eta_n$ limits to a delta function supported on $\Sigma$ we have
    \[
    \lim_{n\to\infty} \int_{M^2} \left( c(x) d_x\star d_x\left(\Delta_v(x,y) \frac{\delta F}{\delta A(y)}\right)\right)d\eta_n(x) = \int_{M_x\times\Sigma_y}d_xc(x)\delta(x,y)\frac{\delta F}{\delta A(y)}
    \]
    where $\star d\Delta_v\Big\vert_{\Sigma} = \delta(x,y)$ is the initial condition \eqref{e:DeltaBC} of the solution $\Delta(x,y)$, and the restriction to $\Sigma$ is imposed by the limit of $d\eta_n$.

    Note that $\frac{\delta F}{\delta A}$ is a derivative w.r.t. both components of $A=A_0dx^0 + A_i dx^i$ in an adapted chart around $\Sigma$. However, the term with the $A_0$ derivative vanishes due to the following calculation (recall that $\Delta$ is diagonal)
    \[
    \int_{M^2} d\eta c d_x \star d_x \Delta_v^{00} \frac{\delta F}{\delta A_0}= \int_{M^2}d\eta c \epsilon^{\ell0\mu\nu}\partial_i\left(\frac{1}{\sqrt{g}}g_{\mu j}g_{\nu k} \partial_\ell\Delta^{00}_v(x,y)\right)\frac{\delta F}{\delta A_0} dx^idx^jdx^k =0
    \]
    where $i,j,k,\ell\not=0$, and the expression vanishes because $\partial_\ell \Delta_v^{00}\Big\vert_{\Sigma} = 0$ (again from the initial conditions \eqref{e:DeltaBC} for $\Delta$). Hence we conclude
    \[
    \lim_{n\to\infty}\left<\frac{\delta J^\std_V[d\eta_n(f)]}{\delta A},\Delta^{AA} \frac{\delta F}{\delta A}\right> = \int_{\Sigma}dc \frac{\delta F}{\delta A_\partial},
    \]
    where $A_\partial$ denotes the restriction of $A$ to $\Sigma$, and thus it only depends on the spatial component of $A$.
    
    Similarly
    \begin{align}
    \underbrace{\left<\frac{\delta J^\std_V[d\eta]}{\delta A},\Delta^{Ab} \frac{\delta F}{\delta b}\right>}_{(b)}
        &= \int_{M^2} \left( c(x) d_x\star d_x\left(\Delta^{Ab}(x,y) \frac{\delta F}{\delta b(y)}\right)\right)d\eta(x) \notag\\
        &+\int_{M^2} \Big( c(x) d_x\star d_x\Big(\Delta^{Ab}(x,y) d^\star_y\alpha_b \Big)\Big)d\eta(x)\notag\\
        &= -i\int_{M^2}\left(c(x)d_x\star d_x (d^\star_y\Delta_v(x,y))\frac{\delta F}{\delta b(y)}\right)d\eta(x)\notag\\
        &+\int_{M^2} \Big( c(x) d_x\star d_x(d^\star_y\Delta_v(x,y)) d^\star_y\alpha_b\Big)d\eta(x)\notag=0
    \end{align}
where in the second step we used
\[
d_x\star d_x (d^\star_y\Delta_v(x,y))=d_x\star d_x \Delta_v(x,y) d_y=d_x\star d_x d_x \Delta_s(x,y)=0\,,
\]
which relies on the second identity from \eqref{e:Deltaidentities}.

    Now, analogously we expand the term
    \begin{align}
    \underbrace{\left<\frac{\delta J^\std_V[d\eta]}{\delta b},\Delta^{bA} \frac{\delta F}{\delta A}\right>}_{(c)}
        &= \int_{M^2}\left(\star d_xc(x) \Delta^{bA}(x,y)\left(\frac{\delta F}{\delta A(y)}+d^\star_y\alpha_A\right)\right)d\eta(x) \notag\\
        &= \int_{M^2}\left(\star d_xc(x) d^\star_x\Delta_v(x,y)\left(\frac{\delta F}{\delta A(y)}+d^\star_y\alpha_A\right)\right)d\eta(x)\\
        &= \int_{M^2}\left(\star d_xc(x) \Delta_s(x,y) d^\star_y\left(\frac{\delta F}{\delta A(y)}+d^\star_y\alpha_A\right)\right)d\eta(x)\\
        &= \int_{M^2}\star d_xc(x) \Delta_s(x,y) d^\star_y \frac{\delta F}{\delta A(y)} d\eta(x)\,,
    \end{align}
    where in the second step we again used the second identity from \eqref{e:Deltaidentities}. Observe that we will be taking the limit for $d\eta_n$ to a $\Sigma$-supported delta, which then implies that we must require $d\eta$ to be transversal to $\Sigma$, or else the limit vanishes. We can then assume for simplicity that the only contribution here comes from terms proportional to $d\eta = dt$. Hence, the only term in $\star dc$ that is allowed to survive is a three form that is tangent to $\Sigma$, which we already identified with the notation $J_c$.

    Taking the limit of $d^\star_x \Delta_v(x,y)$ amounts to taking the limit of $\partial_0 \Delta_v^{00}$ since $\Delta_v$ is diagonal in form indices and only the normal derivative of $\Delta_v$ gives a non-zero contribution according to the boundary conditions. Hence only the term $\Delta_v^{00}(x,y)(x,y)\frac{\delta F}{\delta A_0(y)}$ contributes and we obtain
    \[
    \lim_{n\to\infty}\left<\frac{\delta J^\std_V[d\eta_n]}{\delta b},\Delta^{bA} \frac{\delta F}{\delta A}\right> = \int_{\Sigma}J_c\frac{\delta F}{\delta A_0}
    \]
 Let us now turn to 
    \begin{align}
    \underbrace{\left<\frac{\delta J^\std_V[d\eta]}{\delta c},\Delta^{c\bar{c}} \frac{\delta F}{\delta \bar{c}}\right>}_{(d)}
        &=\int_{M^2}\left((d_x\star d_xA + b\star d_x)\Delta_s(x,y)\left(\frac{\delta F}{\delta \bar{c}}+ d^*_y\alpha_{\bar{c}}\right)\right)d\eta(x)\notag\\
    &=\int_{M^2}d_x\star d_xA \Delta_s(x,y)\frac{\delta F}{\delta \bar{c}}d\eta(x)\notag\\
    &+ \int_{M^2}d_x\star d_xA \Delta_s(x,y)d^*_y\alpha_{\bar{c}}(y)d\eta(x)\notag\\
    &+\int_{M^2} b\star d_x\Delta_s(x,y)\frac{\delta F}{\delta \bar{c}} d\eta(x)\notag\\
    &+\int_{M^2} b\star d_x\Delta_s(x,y)d^*_y\alpha_{\bar{c}}(y) d\eta(x)\,.
    \end{align}
    The first and second terms will vanish, in the limit, due to the boundary conditions: $\Delta_s(0,\underline{x},y) = 0$. The fourth term, instead, contains two time derivatives when looking at $\star d_x\Delta(x,y)d^*_y$, which is proportional to $\partial_0^{(1)}\partial_0^{(2)}\Delta_s(x,y)$. This has vanishing boundary condition due to identity \eqref{e:DeltaBC} so that
    \[
    \left<\frac{\delta J^\std_V[d\eta(f)]}{\delta c},\Delta^{c\bar{c}} \frac{\delta F}{\delta \bar{c}}\right> = \int_\Sigma b\frac{\delta F}{\delta \bar{c}}
    \]
    \begin{align}
    \underbrace{\left<\frac{\delta J^\std_V[d\eta]}{\delta \bar{c}},\Delta^{\bar{c}c} \frac{\delta F}{\delta c}\right>}_{(e)}
        &=0
\end{align}
\end{subequations}

Collecting all terms we conclude
\[
\lim_{n\to\infty}\Pei{J^\std_V[d\eta_n]}{F} = \int_\Sigma \left(dc \frac{\delta F}{\delta A_\partial} + J_c \frac{\delta F}{\delta A_0} + b\frac{\delta F}{\delta \bar{c}}\right)
\]
which coincides with the $\hbar$-linear term (up to a sign) of the quantisation of the boundary action $\hat{L}^\partial\vert_{C^\infty_\mathcal{P}(\Ecal)}$ (Equation \ref{e:quantumboundaryactionYM}) in the polarization outlined in Proposition \ref{prop:BFVquantisationYM}: $\mathcal{P}\doteq \left\{\frac{\delta}{\delta A},\frac{\delta}{\delta A^\dag},A_0\frac{\delta}{\delta J_c} + \overline{c}\frac{\delta}{\delta b}\right\}$.

Summarising we have proven:
\begin{theorem}\label{thm:PAQFT-CMR}
Consider the convergent, smoothened BV-BFV data for Yang--Mills theory on smoothened boundary $(C,N,f)$.  Consider the operator 
\[
\overrightarrow{\Omega}^n_V(F)\doteq[J_V^\std[d\eta_n],F]_{\star_V}
\]
for $d\eta_n\doteq d\eta(f_n)$ a sequence of one-forms that converges to a delta function supported on $\de C$, with $\lim J_V^\std[d\eta_n]=\lim J_{YM}[d\eta_n] = \pi^*L^\de$ and let $\hat{L}^\de \vert_{C^\infty_\mathcal{P}(\Ecal)}$ be the quantisation of the BFV action $L^\de$ given in Proposition \ref{prop:BFVquantisationYM}, then
    \[
    \lim_{n\to\infty}\lim_{\hbar\to 0} \frac{1}{\hbar} \widehat{\Omega}_V^nF = \frac{1}{\hbar} \hat{L}^\de \vert_{C^\infty_\mathcal{P}(\Ecal)}.
    \]
\end{theorem}
\medskip 

\subsection{Comparision to Hollands}\label{sec:Hollands}
The (retarded) quantum BV operator \eqref{e:QVhat burger}, which we recall reads:
\[
Q^R_V(F) =  \left(\frac{i}{\hbar} {[J[d\eta]^-,\bullet]_{{\star_V}}} + Q_0 + \{V,\cdot\} -i\hbar \Lap_V^{\rm ren}\right) F
\]
on a causal cylinder, can be used to reproduce results of \cite{Hollands,FR3,Rej13} and explain how they relate to each other. Assume for simplicity that $J=J^{\std}-J^{\std}_0 = J^{\src} - J_0^{\src}\equiv J_I$ is the interacting generalised Noether current for which $L=L_0 + V$ satisfies the mQME, and set
\[
q_I\doteq J_I[d\eta(f)]^-\,
\]
for a smoothened boundary $(C,N,f)$ and some exact form $d\eta(f)$ dependent on the smoothing.

In the language of \cite{Hollands}, this is the interacting part of the BRST charge associated to $N_-$, the smoothened, past, Cauchy surface. Similarly, we define $q\doteq J^{\std}[d\eta]^-$ and $q_0\doteq J^{\std}_0[d\eta]^-$. The operator $Q^R_V$ breaks down into two terms: the commutator with the BRST charge, which is associated to the past Cauchy surface; and the bulk term $Q_0+\{V,.\}-i\hbar \Lap_V^{\rm ren}$. 

The simplest example corresponds to the choice of the minimal splitting of $L$, i.e. take $L_0$ to be $L_{00}$. In this case, since $J_0^{\std}= J_0^{\src} = 0$, $J_I=J^{\std}$ and $q_I=q$, whence
\[
Q^R_V(F)=R_{V,}^{-1}\circ Q_{00} \circ R_{V,H} =  \frac{i}{\hbar} {[q, \bullet]_{{\star_V}}} + Q_{00} + \{V,\cdot\} -i\hbar \Lap  .
\]
where we used the results of Theorem \ref{thm:R/AqBV}. Applying $R_{V,H}$ to both sides (cf.\ Remark \ref{rmk:RVMultilocExtension2}), we obtain
\[
Q_{00} \circ R_{V,H}=  \frac{i}{\hbar} [R_{V,H}(q),R_{V,H}(\bullet) ]_{\star_H} + R_{V,H}\circ (Q_0 + \{V,\cdot\} -i\hbar \Lap_V^{\rm ren})\,,
\]
 and recognise $R_{V,H}(q)$ as the \textit{interacting BRST charge} of \cite[Section4.7]{Hollands}.

{The differential on the left-hand side is just the Koszul-Tate differential for the free theory (cf. Proposition \ref{prop:quad+minimalsplit}), but considered on the image of $R_V$. This has the standard interpretation of constructing interacting fields within the free theory using the M{\o}ller map. The free dynamics is implemented by taking the cohomology of $Q_{00}$ and is referred to as \textit{going on-shell}.
\begin{lemma}
    If the interacting BRST charge $R_{V,H}(q)$ squares to zero
    \[
    R_{V,H}(q)\star_H R_{V,H}(q)=0,
    \]
    then $[R_{V,H}(q),\bullet ]_{\star_H}$ is a differential.
\end{lemma}
\begin{proof}
    Using the fact that $q$ is odd and the graded Jacobi identity, we conclude that
    \[
    [R_{V,H}(q),[R_{V,H}(q),\bullet ]_{\star_H}]_{\star_H}=0
    \]
    is equivalent to
    \[
    [R_{V,H}(q),R_{V,H}(q) ]_{\star_H}=0
    \]
    and this in turn is equivalent to $R_{V,H}(q)\star_H R_{V,H}(q)=0$ (cf. \cite{Hollands}, formula (429)).
\end{proof}
Note that we actually don't need the interacting BRST charge to square to zero on the nose, but rather on the cohomology of $Q_{00}$, i.e. on-shell. 
 This was proven by Hollands \emph{ibidem}, but in a framework where equations of motion are imposed non-cohomologically. Nevertheless, we believe that this can be proven more generally. In particular, the zeroth order in $\hbar$ is the statement that the interacting Peierls bracket of $q$ with itself vanishes: $\Pei{q}{q}\stackrel{!}{=}0$.}
{Observe that, if we have access to a result such as Equation \eqref{e:scalarPeierlsequivalence} that links the (possibly gauge fixed) Peierls bracket with a canonical bracket on the space of boundary fields we can write
\[
\beta^*\Pei{q}{q} = \{\beta^*q,\beta^*q\}_{\rm can},
\]
for $\beta$ a solution map. This then can be thought of as the BFV master equation for $\beta^*q \equiv S^\partial$ the BFV boundary action.\footnote{Observe that for the scalar field $q=S^\partial = 0$.} In the BV-BFV axiomatics, this is satisfied again up to boundary terms, which from the point of view of $M$ are terms supported on codimension 2 submanifolds.
}
 
{In any case we have (recall that in our simplified scenario $q=J^\std$):
 \begin{lemma}\label{lem:vanishingPEIdJ}
    Assume that $Q$ is a (degree $1$) derivation of the (graded) Peierls bracket, i.e.
    \[
    Q\Pei{F}{G} = \Pei{QF}{G} + (-1)^{|F|}\Pei{F}{QG}.
    \]
    Then the Peierls bracket of $J^\std[d\eta]$ with itself is $Q$ exact:
    \[
    \Pei{J^\std[d\eta(f)]}{J^\std[d\eta(f)]} = Q[f]\left(\Pei{\mathbb{L}[f]}{J^\std[d\eta(f)]}\right),
    \]
    where $\mathbb{L}$ is the {$0$th generalised Noether current} (Proposition \ref{prop:descent}).
\end{lemma}

\begin{proof}
    Let us denote by
    \[
    \frac12\left\{L[f],L[f]\right\}^\std = J^\std[fdf], 
    \]
    and owing to Proposition \ref{prop:descent} the {$0$th generalised Noether current} $\mathbb{L}[f]$ satisfies
    \[
    Q[f](\mathbb{L}[f]) = J^\std[d\eta(f)].
    \]
    Then
    \begin{align*}
        \Pei{J^\std[d\eta(f)]}{J^\std[d\eta(f)]} &= \Pei{ Q[f]\mathbb{L}[f]}{J^\std[d\eta(f)]} \\ 
        &=Q[f]\left(\Pei{\mathbb{L}[f]}{J^\std[d\eta(f)]}\right) - \Pei{\mathbb{L}[f]}{Q[f](J^\std[d\eta])}
    \end{align*}
    and we conclude by observing that $Q[f](J[d\eta])$ vanishes owing, once again, to Proposition \ref{prop:descent}.
\end{proof}
It is possible to argue that the term $\Pei{\mathbb{L}[f]}{J^\std[d\eta(f)]}$ is expected to be supported on the boundary of $\Sigma$ (i.e.\ a corner term \cite{CMR1,RielloSchiavina,RejznerSchiavina_Asymptotic}). We plan to investigate these aspects in a further work.}

Following our general philosophy, we take the cohomology of $Q^R_V$, which---in the minimal splitting---amounts to going on-shell with respect to the free dynamics determined by $Q_{00}$, on the image of $R_V$. The well-posedness of the initial value problem in our situation implies that the bulk theory is  on-shell-equivalent to boundary theory. This is reflected here in the statement that - on the cohomology of $Q^R_V$, the cohomology of $[q, \bullet]_{{\star_V}}$ coincides with the cohomology of $Q_0 + \{V,\cdot\} -i\hbar \Lap_V^{\rm ren}$. The latter can be clearly identified as the gauge-invariant on-shell observables of the full bulk theory, since\footnote{This requires $B_0(V) = 0$, which was implicitly assumed by $J=J_I$.}
\[
(Q_0 + \{V,\cdot\} -i\hbar \Lap_V^{\rm ren})F=(Q-i\hbar \Lap_V^{\rm ren})F
\]
for $\supp F\subset \Ocal$.

Identifying the cohomology of $Q-i\hbar \Lap_V^{\rm ren}$ with that of $[q,\bullet]_{\star_H}$ agrees with the result of \cite[Section 4.8]{Hollands} and its generalization obtained in \cite{Rej13}. The advantage of our current formulation as compared to either cited reference is that going on-shell is understood here in a clear cohomological way and no further choices need to be made, apart from the splitting of $L$ into $L_0$ and $V$. 

Next, we discuss the cohomology of $[q, \bullet]_{\star_V}$ and its representations in terms of operators on a Hilbert space. This is essentially the \emph{Kugo-Ojima formalism}\cite{KugoOjima}. To avoid issues with renormalisation, we describe it in terms of cohomology of $[ R_{V,H}(q),\bullet]_{\star_H}$ taken in the space of functionals of the form $R_{V,H}(F)$, for some multilocal $F$. This is a subspace of the algebra $\fA$ of equicausal functionals, taken with the star product $\star_H$. Going on-shell corresponds to taking the quotient $\fA_{os}\doteq \fA/\fA_0$ by the ideal $\fA_0$ generated by the free equations of motion.

In order to represent the cohomology of $[R_{V,H}(q),\bullet]_{\star_H}$, we need to first find a representation $\pi$ of $\fA_{os}$ by (possibly unbounded) operators on a Krein space\footnote{Recall that the inner product on a Krein space fails to be positive definite.} $\Kcal$.  Next, we observe that functionals in the cohomology of $[R_{V,H}(q),\bullet]_{\star_H}$ are represented as well-defined operators on the cohomology of 
$\pi(R_{V,H}(q))$ (seen as an operator on $\Kcal$ that squares to zero, since $R_{V,H}(q)\star_H R_{V,H}(q)=0$). We then require the subsidiary conditions on the representation $\pi$, as introduced in \cite{DFqed}, namely that 
\begin{enumerate}
    \item if $\psi\in\ker \pi(R_V(q))$, then $\left<\psi,\psi\right>\geq 0$, where $\left<\bullet,\bullet\right>$ is the indefinite inner product of $\Kcal$.
    \item $\psi \in \ker \pi(R_V(q))$ and $\left<\psi,\psi\right>=0$ implies that $\psi \in \im  \pi(R_V(q))$.
\end{enumerate}
With these assumptions, we obtain a Hilbert space representations of the cohomology of $[R_V(q),\bullet]_{\star_H}$, identified as the space of gauge-invariant on-shell observables.

The remaining problem is to find $\pi$ satisfying the subsidiary conditions above. Fortunately, this can be easily done in the perturbative setting introduced in \cite[Section 4.3]{DFqed}. One starts with the free BRST charge $q_0$ and finds a representation that satisfies the subsidiary conditions for $\pi(q_0)$. An example of such representation for Yang-Mills theory is given in \cite[Section~4.1]{Hollands}. One can then represent the cohomology of $[q_0,\bullet]_{\star_H}$ on the cohomology of $\pi(q_0)$. Next, one writes  $[R_V(q), \bullet]_{\star_H}$ and 
$\pi(R_V(q))$ as formal power series in the coupling constant, i.e.:
\begin{align*}
 [R_{V,H}(q), \bullet]_{\star_H}&=[q_0, \bullet]_{\star_H}+\Ocal(\lambda)\,,\\
\pi(R_{V,H}(q))&=\pi(q_0)+\Ocal(\lambda)
\end{align*}
Using Theorem~4 of \cite{DFqed}, we conclude that the subsidiary conditions hold also for $ [R_{V,H}(q), \bullet]_{\star_H}$ and $\pi(R_{V,H}(q))$.

To see the relation of our current formulation to \cite{FR3}, note that in the latter, one assumes that the smearing with test function is done in such a way that the classical (and also quantum) master equation holds exactly, without boundary terms. In that case
\[
Q_{00} \circ R_{V,H}(F)=   R_{V,H}\circ (Q -i\hbar \Lap) F\,,
\]
and taking cohomology of $Q_{00}$ on the image of $R_{V,H}$ already realizes the space of gauge-invariant on-shell observables. 

\begin{appendices}

\section{More on local densities}\label{app:densities}
In this appendix we spell out the proofs of some facts about local densities stated in Section~\ref{sec:densities currents lagrangians}. We assum to have a local symplectic density ${\bom}$ which defines both source and standard brackets $\{\cdot,\cdot\}^\src$ and $\{\cdot,\cdot\}^\std$ as per Definitions \ref{def:sourcebracket} and \ref{def:stdbracket}.

\medskip 
\begin{lemma}\label{lem:actionvsbracket}
    Let ${\bF},{\bG}\in\oloc^{0,\top}(\Ecal\times M)$ Hamiltonian. The following relation holds:
    \begin{align*}
    \iota_{X_{\bF}}\iota_{X_{\bG}}{\bom} 
        &= X_{\bF}({\bG}) + d\iota_{X_{\bF}}{\bth}_{{\bG}} \\
        &=\frac12\left(X_{\bF}({\bG}) - \sigma_\omega({\bF},{\bG}) X_{\bG}({\bF})\right)\\
        &- \frac12d\left((-1)^{|{\bF}|+|{\bom}|}\iota_{X_{\bF}}{\bth}_{{\bG}} +(-1)^{|{\bG}|+|{\bom}|}\sigma_\omega({\bF},{\bG})\iota_{X_{\bG}}{\bth}_{{\bF}} \right),
    \end{align*}
    where we defined the sign
    \[
    \sigma_\omega({\bF},{\bG})\doteq (-1)^{(|{\bF}|+|{\bom}|)(|{\bG}|+|{\bom}|)}
    \]
    so that
    \[
    \iota_{X_{\bF}}\iota_{X_{\bG}}{\bom} = -\sigma_\omega({\bF},{\bG})\iota_{X_{\bG}}\iota_{X_{\bF}}{\bom}.
    \]
    Moreover, the local density $\iota_{X_{\bF}}\iota_{X_{\bG}}{\bom} $ is Hamiltonian with Hamiltonian vector field $[X_\bF,X_\bG]$.
\end{lemma}
\begin{proof}
A straightforward calculation yields
    \[
        \iota_{X_{\bF}}\iota_{X_{\bG}}{\bom} = \iota_{X_{\bF}}\left(\delta {\bG} + d{\bth}_{{\bG}}\right) = X_{\bF}({\bG}) - (-1)^{|{\bF}|+|{\bom}|}d\iota_{X_{\bF}}{\bth}_{{\bG}}
    \]
    where we used $\iota_{X}d = (-1)^{|X|+1}d\iota_X$ and $|X_F|=|F|\pm|\omega|$, hence
      \[
        \iota_{X_{\bF}}\iota_{X_{\bG}}{\bom}= \sigma_\omega({\bF},{\bG})\iota_{X_{\bG}}\iota_{X_{\bF}}{\bom} =\sigma_\omega({\bF},{\bG}) \left(X_{\bG}({\bF}) - (-1)^{|{\bG}|+|{\bom}|}d\iota_{X_{\bG}}{\bth}_{{\bF}}\right)
    \]  
    and summing the two expressions
    \[
    \iota_{X_{\bF}}\iota_{X_{\bG}}{\bom} 
    = \frac12\left(X_{\bF}({\bG}) + \sigma_\omega({\bF},{\bG}) X_{\bG}({\bF}) - \left((-1)^{|{\bF}|+|{\bom}|}d\iota_{X_{\bF}}{\bth}_{{\bG}} +(-1)^{|{\bG}|+|{\bom}|}\sigma_\omega({\bF},{\bG})d\iota_{X_{\bG}}{\bth}_{{\bF}} \right)\right)
    \]

    Moreover, up to signs depending on the degrees of ${\bF},{\bG}$:
    \begin{align*}
    \delta \iota_{X_{\bF}}\iota_{X_{\bG}}{\bom} 
    &= \iota_{[X_{{\bF}},X_{{\bG}}]}{\bom} \pm \iota_{X_{{\bF}}}\delta (\delta {\bG} + d{\bth}_{{\bG}}) \pm \iota_{X_{{\bG}}}\delta (\delta {\bF} + d{\bth}_{{\bF}})\\
    &=\iota_{[X_{{\bF}},X_{{\bG}}]}{\bom} \pm d(\iota_{X_{{\bF}}}\delta {\bth}_{{\bG}} \pm \iota_{X_{{\bG}}}\delta {\bth}_{{\bF}})
    \end{align*}
    which means that the density $\iota_{X_{\bF}}\iota_{X_{\bG}}{\bom}$ is Hamiltonian, with Hamiltonian vector field $[X_F,X_G]$ and associated $d$-exact term\footnote{Note that this is defined only up to a $d$-exact term owing to the acyclicity of the variational bicomplex in vertical form degree greater or equal than one. } 
    \[
    d{\bth}_{\iota_{X_{\bF}}\iota_{X_{\bG}}{\bom}} = \pm d(\iota_{X_{{\bF}}}\delta {\bth}_{{\bG}} \pm \iota_{X_{{\bG}}}\delta {\bth}_{{\bF}})
    \]
    and thus $\{F[f],G[g]\}^{\src}$ is again Hamiltonian as a generalised Lagrangian.
\end{proof}

\begin{proposition}\label{prop:fullbracketvsaction}
    Lef $F, G$ be two Hamiltonian generalised Lagrangian. Then
    \begin{align*}
\{F[f],G[f]\}^{\std} &= \{F[f],G[g]\}^\src\\
    &+(-1)^{|{\bF}|+|{\bom}|}(\iota_{X_{F[f]}}\theta_G)[fdf] + \sigma_\omega({\bF},{\bG})(-1)^{|{\bG}|+|{\bom}|}(\iota_{X_{F}}\theta_G)[fdf]\\
    &= \frac12 \left(X_{F[f]}G[f] +\sigma_\omega({\bF},{\bG})X_{G[f]}F[f]\right)\\
    &+\frac12\left((-1)^{|{\bF}|+|{\bom}|}(\iota_{X_{F}}\theta_G)[fdf] +\sigma_\omega({\bF},{\bG}) (-1)^{|{\bG}|+|{\bom}|}(\iota_{X_{G}}\theta_F)[fdf]\right).
\end{align*}
\end{proposition}

\begin{proof}

Recalling that a generalised Lagrangian is specified by a local density and 
\[
L[f]=\int_M\mathbf{L} f, \qquad \delta L[f] = \int_M \delta^\src{\bF} f + fd{\bth}_{\bF}
\]
we compute
    \begin{align*}
\{F[f],G[f]\}^\std
    &=\Pi\left(\int(\delta {\bF})f,\int (\delta {\bG})f\right) \\
    &= \Pi\left(\int(\delta^\src {\bF})f + d{\bth}_{\bF} f,\int (\delta^\src {\bG})f + d{\bth}_{\bG} f\right) \\
    &=\Pi\left(\int(\delta^\src {\bF})f,\int (\delta^\src {\bG})f\right) + \Pi\left(\int d{\bth}_{\bF} f,\int (\delta^\src {\bG})f\right) \\
        &+ \Pi\left(\int (\delta^\src {\bF})f,\int d{\bth}_{\bG} f\right) + \Pi\left(\int d{\bth}_{\bF} f,\int d{\bth}_{\bG} f\right)\\
    &=\{F[f],G[f]\}^\src + \Pi\left(\int d{\bth}_{\bF} f,\int (\delta^\src {\bG})f\right) \\
        &+ \Pi\left(\int (\delta^\src {\bF})f, \int d{\bth}_{\bG} f\right) + \Pi\left(\int {\bth}_{\bF} df,\int {\bth}_{\bG} df\right)
\end{align*}
Now, the term $\Pi\left(\int {\bth}_{\bF} df,\int {\bth}_{\bG} df\right)$ vanishes owing to the $df\wedge df $ term, while 
\begin{align*}
\Pi\left(\int (\delta^\src {\bF})f,\int d{\bth}_{\bG} f\right) 
    &= \iota_{X_{F[f]}}\int d{\bth}_{{\bG}}f = (-1)^{|{\bF}|+|{\bom}|}\int fdf \iota_{X_{\bF}}{\bth}_{{\bG}} \\
    &\doteq (-1)^{|{\bF}|+|{\bom}|} (\iota_{X_F}\theta_G)[fdf]
\end{align*}
and, similarly, with the sign $\sigma_\omega({\bF},{\bG})$ depending on the degree of ${\bom}$ and the densities\footnote{The degree of $\delta{\bF}$ and $d{\bth}$ agree, so the parity of ${\bF}$ and ${\bth}$ also agrees.} ${\bF}$ and ${\bG}$.
\begin{align*}
\Pi\left(\int d{\bth}_{\bF} f,\int (\delta^\src {\bG})f\right)
    &=\sigma_\omega({\bF},{\bG})\Pi\left(\int (\delta^\src {\bG})f,\int d{\bth}_{\bF} f\right)\\
    &= \sigma_\omega({\bF},{\bG}) (-1)^{|{\bG}|+|{\bom}|}\int  fdf\iota_{X_{{\bG}}}{\bth}_{{\bF}}\\
    &\doteq \sigma_\omega({\bF},{\bG}) (-1)^{|{\bG}|+|{\bom}|} (\iota_{X_G}\theta_F)[fdf].
\end{align*}
so that
\begin{align*}
\{F[f],G[f]\}^\std = \{F[f],G[f]\}^\src 
    &+ (-1)^{|{\bF}|+|{\bom}|}(\iota_{X_{F[f]}}\theta_G)[fdf] \\
    &+ \sigma_\omega({\bF},{\bG})(-1)^{|{\bG}|+|{\bom}|}(\iota_{X_{F}}\theta_G)[fdf]
\end{align*}

From the integrated Hamiltonian relation $\iota_{X_{F[f]}}\omega = \delta F[f] - \theta_F[df]$, we can compute the following (cf.\ the definitions in Equation \eqref{e:smearedtheta} and Lemma \ref{lem:Hamvfequiv}):
\begin{align*}
X_{F[f]}G[f] &= \iota_{X_{F[f]}}\delta G[f] = \iota_{X_{F[f]}}\iota_{X_{G[f]}}\omega + \iota_{X_{F[f]}}(\theta_G[df])\\
    &=\iota_{X_{F[f]}}\iota_{X_{G[f]}}\Omega + \iota_{X_{F[f]}}(\theta_G[df]) \\
    &= \{F[f],G[f]\}^\src + (-1)^{|{\bF}|+|{\bom}|}(\iota_{X_{F}}\theta_G)[fdf]
\end{align*}
and similarly exchanging $F\leftrightarrow G$. Hence we have
\begin{align*}
\{F[f],G[f]\}^\src 
    &= \frac12 \left(X_{F[f]}G[f] +\sigma_\omega({\bF},{\bG})X_{G[f]}F[f]\right) \\
    &-\frac12\left((-1)^{|{\bF}|+|{\bom}|}(\iota_{X_{F}}\theta_G)[fdf] +\sigma_\omega({\bF},{\bG}) (-1)^{|{\bG}|+|{\bom}|}(\iota_{X_{G}}\theta_F)[fdf]\right)
\end{align*}
Thus, putting everything together, we get the claim:
\begin{align*}
\{F[f],G[f]\}^{\std} &= \{F[f],G[g]\}^\src\\
    &+(-1)^{|{\bF}|+|{\bom}|}(\iota_{X_{F[f]}}\theta_G)[fdf] + \sigma_\omega({\bF},{\bG})(-1)^{|{\bG}|+|{\bom}|}(\iota_{X_{F}}\theta_G)[fdf]\\
    &= \frac12 \left(X_{F[f]}G[f] +\sigma_\omega({\bF},{\bG})X_{G[f]}F[f]\right)\\
    &+\frac12\left((-1)^{|{\bF}|+|{\bom}|}(\iota_{X_{F}}\theta_G)[fdf] +\sigma_\omega({\bF},{\bG}) (-1)^{|{\bG}|+|{\bom}|}(\iota_{X_{G}}\theta_F)[fdf]\right).
\end{align*}

\end{proof}

\begin{corollary}\label{cor:srcnormalbracket}
    Let $L$ be a generalised Lagrangian. Then
    \[
    \{L[f],L[f]\}^\std = \{L[f],L[f]\}^{\src} + 2\iota_{X_{L[f]}}\theta_{L}[df],
    \]
    as well as
    \[
    X_{L[f]}(L[f])=\{L[f],L[f]\}^\std - \iota_Q\theta_L[fdf],
    \]
    where we denoted
    \[
    \iota_Q\theta_{L}[fdf]\doteq \int fdf\iota_{X_\mathbf{L}}{\bth}_{\mathbf{L}}.
    \]
\end{corollary}

\begin{proof}
This is a consequence of Proposition \ref{prop:fullbracketvsaction} (cf. Equation \eqref{e:fullsrcBrel}), for $F[f]=G[f]=L[f]$ and $\sigma_\omega({\bF},{\bG})=\sigma_\omega(\mathbf{L},\mathbf{L})=-1$. 
\end{proof}

\end{appendices}

\begingroup
\sloppy
\printbibliography
\endgroup

\end{document}